%% file: equal-keys-main.tex
	\documentclass[
		paper=a4,%
		11pt,american,pagesize,%
		DIV=12,headinclude=true,headlines=0,footinclude=true,footlines=-5,twoside=semi%
	]{scrartcl}
	\def\version{arxiv}
\input{equal-keys-preamble}

\makeatletter
\ifproceedings{
	\title{%
		Quicksort Is Optimal For Many Equal Keys%
		\thanks{%
			\strut{}An extended version of this paper with full proofs
			is available online: \href{https://arxiv.org/abs/1608.04906}{\texttt{arxiv.org/abs/1608.04906}}.
		}%
	}
}{
	\title{%
		Quicksort Is Optimal For Many Equal Keys%
	}
}
\author{%
	Sebastian Wild%
	\thanks{%
		David R.\ Cheriton School of Computer Science, 
		University of Waterloo,
		Email: \texttt{wild@uwaterloo.ca}%
	}
}
\hypersetup{
	pdftitle={\@title},
	pdfauthor={\@author}
}

\makeatother

\begin{document}

\maketitle

\ifproceedings{%
	\fancyfoot[R]{\footnotesize{\textbf{Copyright \textcopyright\ 2017 by SIAM\\
	Unauthorized reproduction of this article is prohibited}}}

}{}

\begin{abstract}
\noindent
I~prove that 
the average number of comparisons for median-of-$k$ Quicksort 
(with fat-pivot \aka three-way partitioning) 
is asymptotically only a constant $\alpha_k$ times worse than 
the lower bound for sorting random multisets of $n$ elements with $\Omega(n^\epsilon)$ duplicates of each value (for any $\epsilon>0$).
The constant is
$\alpha_k = \ln(2) / \bigl(\harm{k+1}-\harm{\smash{(k+1)/2}} \bigr)$,
which converges to~$1$ as $k\to\infty$,
so median-of-$k$ Quicksort is asymptotically optimal for inputs with many duplicates.
This partially resolves a conjecture by 
Sedgewick and Bentley (1999, 2002)
and constitutes the first progress on the analysis of 
Quicksort with equal elements since Sedgewick's 1977 article.

\end{abstract}

\input{equal-keys-intro}

\input{equal-keys-preliminaries}

\input{equal-keys-stochastic-preliminaries}

\input{equal-keys-properties-of-entropy}

\input{equal-keys-input-models}

\input{equal-keys-previous-work}

\input{equal-keys-search-trees}

\input{equal-keys-quicksort-many-duplicates}

\input{equal-keys-separating-n-and-q}

\ifproceedings{}{\needspace{10\baselineskip}}

\input{equal-keys-entropy-bounds}
\ifproceedings{}{\needspace{5\baselineskip}}
\input{equal-keys-entropy-bounds-proof}

\input{equal-keys-lower-bound}
\input{equal-keys-lower-bound-proof}

\ifproceedings{}{\needspace{10\baselineskip}}
\input{equal-keys-proof-main-result}

\input{equal-keys-conclusion}

\section*{Acknowledgments}
I am very grateful for inspiring discussions
with Conrado Martínez, Markus Nebel, Martin Aumüller and Martin Dietzfelbinger,
which laid the ground for the present paper.
Moreover I thank Sándor Fekete for many helpful comments that improved 
the presentation and framing of the results.
Finally, I thank my anonymous referees for their scrutiny and thorough reviews;
their comments helped to clarify the presentation significantly.

\ifproceedings{}{

	\appendix
	\part*{Appendix}
	\pdfbookmark[0]{Appendix}{}
	
	\ifproceedings{}{%
		\manualmark%
		\markleft{\mytitle}%
		\automark*[section]{}%
	}

	\numberwithin{algorithm}{section}
	\numberwithin{table}{section}
	\numberwithin{figure}{section}
\input{equal-keys-notation} 

	\input{equal-keys-height}

}

{
\small
\bibliography{quicksort-refs}
}

\end{document}

%% file: equal-keys-preamble.tex
\RequirePackage{etex}
\RequirePackage{fixltx2e}

\makeatletter

\usepackage{lmodern}
\usepackage[utf8]{inputenc}
\usepackage[T1]{fontenc}
\usepackage{babel}
\input{ushyphex.tex} %

\usepackage{amsmath}
\usepackage{amssymb}
\usepackage{amsfonts}
\usepackage{mathtools}
\usepackage{colonequals}
\usepackage{relsize}
\usepackage{ifthen}
\usepackage{xcolor}
\usepackage{chngcntr}

\newcommand{\ifarxiv}[2]{\ifthenelse{\equal{\version}{arxiv}}{#1}{#2}}
\newcommand{\ifproceedings}[2]{\ifthenelse{\equal{\version}{proceedings-siam}}{#1}{#2}}

\ifproceedings{%
	\usepackage{ltexpprt}
}{}

\AtBeginDocument{%
	\let\mytitle\@title%
}

\ifproceedings{%
	\newenvironment{captionbeside}[1]{%
		\caption{#1}
	}{%
	}
}{
	\setlength\parindent{1.5em}
	\usepackage{colonequals}
	\usepackage[headsepline]{scrlayer-scrpage}
	\pagestyle{scrheadings}
	\clearscrheadfoot
	\AtBeginDocument{%
		\manualmark%
		\markleft{\mytitle}%
		\automark*[section]{}%
	}
	\ohead{\pagemark}
	\ihead{\headmark}
	\addtokomafont{caption}{\sffamily\small}
	\addtokomafont{captionlabel}{\sffamily\textbf}
	\setcapmargin{2em}
}

	\usepackage{microtype}

\ifproceedings{\date{~}}{\date{\small\today}}

\newboolean{colored}
\setboolean{colored}{false}

\AtBeginDocument{
	\ifthenelse{\boolean{colored}}{
		\colorlet{backgroundshading}{red!70!orange!20} %
	}{
		\colorlet{backgroundshading}{black!10}
	}
	\ifthenelse{\boolean{colored}}{
		\colorlet{headingscolor}{red!60!black} %
	}{
		\colorlet{headingscolor}{black!80}
	}
	\colorlet{softhighlight}{headingscolor!50!backgroundshading} 
}

\usepackage{tikz}
\usetikzlibrary{babel,quotes}
\usetikzlibrary{patterns}
\usetikzlibrary{positioning,external,calc,graphs}
\usetikzlibrary{decorations.pathreplacing,shapes.multipart,shapes.geometric}
\usetikzlibrary{fit}
\usepackage{ifluatex}
\ifluatex
	\usetikzlibrary{graphdrawing}
	\usegdlibrary{trees,layered}
\else \fi

\tikzsetexternalprefix{pics/externalized/}
\tikzset{external/export=false} %

\ifarxiv{
	\newcommand{\externalizedpicture}[1]{%
		\includegraphics{pics/externalized/#1}%
	}
}{
	\newcommand{\externalizedpicture}[1]{%
		\tikzset{external/export=true}%
		\tikzsetnextfilename{#1}%
	}
}

\usepackage[outline]{contour}

\usepackage{pgfplots}
\pgfplotsset{compat=1.12}
\usepgfplotslibrary{patchplots,ternary,dateplot}

\pgfdeclarelayer{foreground}
\pgfdeclarelayer{background}
\pgfsetlayers{background,main,foreground}

\tikzset{%
	every picture/.style={
		>=stealth,
		font={\sffamily\small},
	},
}

\tikzset{%
	highlightbox/.style={
		draw=softhighlight,
		very thick,
		rounded corners=3pt,
	},
}

\mathtoolsset{mathic=true}

\newcommand\separatedpar{%
	\medskip
	\plaincenter{\textcolor{headingscolor}{%
		$*$\hspace{3em}$*$\hspace{3em}$*$}%
	}\par
	\medskip
	\noindent
}

\usepackage[numbers,square]{natbib}
\renewcommand*{\NAT@spacechar}{~}

\let\oldthebibliography\thebibliography
\renewcommand\thebibliography[1]{%
	\oldthebibliography{#1}%
	\pdfbookmark[1]{\bibname}{}%
}

\usepackage{xspace}
\usepackage{enumitem}
\usepackage{array,multirow,multicol}
\usepackage{booktabs}
\usepackage{dcolumn}
\usepackage{multicol}
\usepackage{multirow}
\usepackage[referable]{threeparttablex}
\usepackage{colortbl}
\usepackage[TABBOTCAP]{subfigure}
\usepackage{rotating}

\usepackage{adjustbox}

\usepackage{xfrac}
\usepackage{dsfont}
\usepackage{bm}
\usepackage{mleftright}\mleftright %
\usepackage{stmaryrd}
\usepackage[hyphens]{url}
\usepackage{placeins}
\usepackage{needspace}

\usepackage[margin=.75em,font={small,rm},labelfont={bf}]{caption}
\DeclareCaptionStyle{ruled}{format=plain,%
						    labelfont=bf,%
                            labelsep=period,%
                            strut=on,%
                            margin=.3em,%
                            font={small,rm}%
}

\usepackage{clrscode3e}
\def\Foreach{\kw{for each} }
\def\EndFor{\End\li\kw{end for} }
\def\EndIf{\End\li\kw{end if} }
\def\EndWhile{\End\li\kw{end while} }

\usepackage{float}
\floatstyle{ruled}
\newfloat{algorithm}{tbp}{loa}
\floatname{algorithm}{Algorithm}
\usepackage[
	final,
	unicode=true, 
	bookmarks=true,
	bookmarksnumbered=true,
	bookmarksdepth=2,
	bookmarksopen=true,
	breaklinks=true,
	hidelinks,
	pdfhighlight={/P}
]{hyperref}

\usepackage{wref}

\usepackage{attachfile2}

\usepackage[amsmath,hyperref,thmmarks]{ntheorem}
\theorembodyfont{\slshape}
\theoremseparator{:}
\newtheoremstyle{proofstyle}%
  {\item[\theorem@headerfont\hskip\labelsep ##1\theorem@separator]}%
  {\item[\theorem@headerfont\hskip\labelsep ##1 of ##3\theorem@separator]}

\theorempreskip{\topsep} %

\ifproceedings{
	\renewtheorem{theorem}{Theorem}[section]
	
	\theoremstyle{plain}
	\theorempreskip{\topsep}
	
	\renewtheorem{proposition}[theorem]{Proposition}
	\renewtheorem{lemma}[theorem]{Lemma}
	\newtheorem{conjecture}[theorem]{Conjecture}
	\renewtheorem{corollary}[theorem]{Corollary}
	\newtheorem{definition}[theorem]{Definition}
	
	\theoremstyle{plain}
	\theorembodyfont{\upshape}
	
	\renewtheorem{property}[theorem]{Property}
	\renewtheorem{fact}[theorem]{Fact}

	\theoremsymbol{\raisebox{-.25ex}{$\Box$}}
	\qedsymbol{\raisebox{-.25ex}{$\Box$}}

	\theoremstyle{proofstyle}
	\renewtheorem{proof}{Proof}
	
}{%
	\newtheorem{theorem}{Theorem}[section]
	
	\theoremstyle{plain}
	\theorempreskip{\topsep}
	
	\newtheorem{proposition}[theorem]{Proposition}
	\newtheorem{lemma}[theorem]{Lemma}
	
	\newtheorem{corollary}[theorem]{Corollary}
	\newtheorem{definition}[theorem]{Definition}
	
	\theoremstyle{plain}
	\theorembodyfont{\upshape}

	\theoremsymbol{\raisebox{-.25ex}{$\Box$}}
	\qedsymbol{\raisebox{-.25ex}{$\Box$}}

	\theoremstyle{proofstyle}
	\newtheorem{proof}{Proof}
}

\newenvironment{thmenumerate}[1]{%
	\begin{enumerate}[
		label={\makebox[\widthof{(a)}][c]{\textup{(\alph*)}}},
		ref={\ref{#1}\kern.1em--\kern.1em(\alph*)},
		itemsep=0pt,
	]%
}{%
	\end{enumerate}%
}

\vfuzz \hfuzz	

\allowdisplaybreaks[3]

\pgfplotsset{
	legend cell align=left,
	every axis legend/.append style={%
		fill=backgroundshading!50!white,
		rounded corners=1pt,
	},
	xlabel near ticks,
	ylabel near ticks,
}
\pgfplotsset{
	every y tick scale label/.style={
		at={(-.03,1)},yshift=-1ex,anchor=south east,inner sep=0pt
	}
}

\def\mathematicaPlotAddPlot{%
	\@ifstar\@mathematicaPlotAddPlot\@@mathematicaPlotAddPlot%
}
\newcommand\@@mathematicaPlotAddPlot[4][]{%
	\addplot+[#1] table[x=#2,y=#3] {\tablefile} ;
	\ifthenelse{\equal{#4}{}}{}{\addlegendentry{#4}}
}
\newcommand\@mathematicaPlotAddPlot[4][]{%
	\addplot[#1] table[x=#2,y=#3] {\tablefile} ;
	\ifthenelse{\equal{#4}{}}{}{\addlegendentry{#4}}
}

\pgfplotscreateplotcyclelist{coloredlineplots}{%
	thick,headingscolor\\%
	thick,black\\%
	ultra thick,headingscolor,densely dotted\\%
	very thick,black,densely dashed\\%
	semithick,headingscolor,dashdotted\\%
	thin,black,loosely dashed\\%
}

\pgfplotscreateplotcyclelist{coloreddiscreteplots}{%
	mark=*,headingscolor,every mark/.append style={fill=headingscolor!50}%
\\%
	mark=square*,every mark/.append style={fill=white}%
\\%
	densely dashed,mark=square*,headingscolor,%
		every mark/.append style={solid,fill=red!10}%
\\%
	densely dashed,mark=*,%
		every mark/.append style={solid,fill=black!50}%
\\%
	mark=diamond*,every mark/.append style={%
		semithick,solid,fill=headingscolor!30!black},%
	headingscolor,very thick,densely dotted,%
\\%
}

\newdimen\makeboxdimen
\newcommand\makeboxlike[3][l]{%
\setbox0=\hbox{#2}%
\global\makeboxdimen=\wd0%
\setbox1=\hbox{\makebox[\makeboxdimen][#1]{%
\makebox[0pt][#1]{#3}%
}}%
\ht1=\ht0%
\dp1=\dp0%
\box1%
}

\newcommand\plaincenter[1]{%
	\mbox{}\hfill#1\hfill\mbox{}%
}

\newcounter{inlineenum}

\newcommand*\ie{\mbox{i.\hspace{.2ex}e.}}
\newcommand*\eg{\mbox{e.\hspace{.2ex}g.}}
\newcommand*\iid{\mbox{i.\hspace{.2ex}i.\hspace{.2ex}d.}\xspace}
\newcommand*\IID{\mbox{I.\hspace{.2ex}I.\hspace{.2ex}D.}\xspace}

\newcommand*\wrt{\mbox{w.\hspace{.2ex}r.\hspace{.2ex}t.}\xspace}
\newcommand*\aka{\mbox{a.\hspace{.2ex}k.\hspace{.2ex}a.}\xspace}
\newcommand*\withoutlossofgenerality{\mbox{w.\hspace{.23ex}l.\hspace{.2ex}o.\hspace{.17ex}g.}\xspace}

\newcommand*\whp{\mbox{w.\hspace{.2ex}h.\hspace{.2ex}p.}\xspace}

\newcommand*\phdthesis{\mbox{Ph.\hspace{.2ex}D.} thesis\xspace}

\newcommand\parenthesisclause[1]{%
	\unskip%
	\nobreak\hspace{.2ex}\mbox{---}\hspace{.2ex}%
	#1%
	\nobreak\hspace{.2ex}\mbox{---}\hspace{.2ex}%
	\ignorespaces%
}
\newcommand*\parenthesissign{%
	\unskip%
		\hspace{.2ex}---\hspace{.2ex}%
	\ignorespaces%
}

\newcommand\literalquote[3]{%
	{%
		\def\ellipsis{[\,\dots]\xspace}%
		\def\comment##1{[##1]}%
		\def\cite##1{##1}%
		\edef\citekey{#2}%
		{\itshape``\ignorespaces#1\unskip''}
		(%
			\ifthenelse{\equal{#3}{}}{%
			\citet{\citekey}%
			}{%
			\citet{\citekey}%
			, p.\,#3%
			}%
		)%
	}%
	\xspace%
}
\newcommand\literalquotegerman[3]{%
	{%
		\def\ellipsis{[\,\dots]\xspace}%
		\def\comment##1{[##1]}%
		\def\cite##1{##1}%
		\edef\citekey{#2}%
		{\selectlanguage{ngerman}%
		\itshape,,\ignorespaces#1\unskip``}
		(%
			\ifthenelse{\equal{#3}{}}{%
			\citet{\citekey}%
			}{%
			\citet{\citekey}%
			, p.\,#3%
			}%
		)%
	}%
	\xspace%
}
\newcommand\literalquotep[3]{%
	{%
		\def\ellipsis{[\,\dots]\xspace}%
		\def\comment##1{[##1]}%
		\def\cite##1{##1}%
		\edef\citekey{#2}%
		{\itshape``\ignorespaces#1\unskip''}
		(%
			\ifthenelse{\equal{#3}{}}{%
			\citep{\citekey}%
			}{%
			\citep{\citekey}%
			, p.\,#3%
			}%
		)%
	}%
	\xspace%
}

\newcommand{\ESymbol}{\mathbb{E}}

\newcommand\R{\mathbb R}
\newcommand\N{\mathbb N}
\newcommand\Z{\mathbb Z}
\newcommand\Q{\mathbb Q}
\renewcommand\C{\mathbb{C}}
\newcommand\Oh{O}
\def\.{\mskip1mu}

\DeclareMathOperator\ld{ld}

\newcommand{\ProbSymbol}{\ensuremath{\mathbb{P}}}
\providecommand{\given}{}
\DeclarePairedDelimiterXPP\Prob[1]{\ProbSymbol}[]{}{%
	\renewcommand\given{\nonscript\:\delimsize\vert\nonscript\:\mathopen{}}%
	#1%
}
\DeclarePairedDelimiterXPP\E[1]{\ESymbol}[]{}{%
	\renewcommand\given{\nonscript\:\delimsize\vert\nonscript\:\mathopen{}}%
	#1%
}
\DeclarePairedDelimiterXPP\Eover[2]{\ESymbol_{#1}}[]{}{%
	\renewcommand\given{\nonscript\:\delimsize\vert\nonscript\:\mathopen{}}%
	#2%
}
\providecommand{\Prob}{} %
\providecommand{\E}{} %
\providecommand{\Eover}{} %

\newcommand\total[1]{\Sigma #1}

\newcommand\ui[2]{#1^{\smash{(}#2\smash{)}}}

\usepackage{fixmath}
\newcommand{\vect}[1]{\boldsymbol{\mathbold{#1}}}

\newcommand\eqdist{	
	\mathchoice{%
		\mathrel{\overset{\raisebox{0ex}{$\scriptstyle \mathcal D$}}=}%
	}{%
		\mathrel{\like{=}{%
			\overset{\raisebox{-1ex}{$\scriptscriptstyle \mathcal D$}}=%
		}}%
	}{%
		\mathrel{\overset{\mathcal D}=}%
	}{%
		\mathrel{\overset{\mathcal D}=}%
	}%
}

\newcommand\ppe{\phantom{=}}

\newcommand\like[3][c]{%
	\mathchoice{%
		\makeboxlike[#1]{%
			\ensuremath{\displaystyle\relax#2}%
		}{%
			\ensuremath{\displaystyle\relax#3}%
		}%
	}{%
		\makeboxlike[#1]{%
			\ensuremath{\textstyle\relax#2}%
		}{%
			\ensuremath{\textstyle\relax#3}%
		}%
	}{%
		\makeboxlike[#1]{%
			\ensuremath{\scriptstyle\relax#2}%
		}{%
			\ensuremath{\scriptstyle\relax#3}%
		}%
	}{%
		\makeboxlike[#1]{%
			\ensuremath{\scriptscriptstyle\relax#2}%
		}{%
			\ensuremath{\scriptscriptstyle\relax#3}%
		}%
	}
}

\newcommand\bernoulli{\mathrm B}
\newcommand\hypergeometric{\mathrm{HypG}}
\newcommand\multinomial{\mathrm{Mult}}
\newcommand\binomial{\mathrm{Bin}}

\newcommand\betadist{\mathrm{Beta}}

\DeclarePairedDelimiterXPP\distFromWeights[1]{\mathcal D}(){}{#1}
\providecommand\distFromWeights{} %

\newcommand\BetaFun{\mathrm B}

\renewcommand\given{\mathbin{\mid}}
\newcommand\harm[1]{\ensuremath{H_{#1}}}

\newcommand\ce{\colonequals}

\newcommand{\surroundedmath}[3]{%
	\mathchoice{%
		#1{#2{#3}#2}%
	}{%
		#1{#3}%
	}{%
		#1{#3}%
	}{%
		#1{#3}%
	}%
}

\newcommand\rel[1]{\surroundedmath{\mathrel}{\:}{#1}}
\newcommand\wrel[1]{\surroundedmath{\mathrel}{\;}{#1}}
\newcommand\wwrel[1]{\surroundedmath{\mathrel}{\;\;}{#1}}
\newcommand\bin[1]{\surroundedmath{\mathbin}{\:}{#1}}
\newcommand\wbin[1]{\surroundedmath{\mathbin}{\;}{#1}}
\newcommand\wwbin[1]{\surroundedmath{\mathbin}{\;\;}{#1}}

\newcommand{\eqwithref}[2][c]{%
	\relwithref[#1]{#2}{=}%
}
\newcommand{\relwithref}[3][c]{%
	\mathrel{\underset{\mathclap{\makebox[\widthof{$=$}][#1]{\scriptsize\wref{#2}}}}{#3}}%
}
\newcommand{\relwithtext}[3][c]{%
	\mathrel{\underset{\mathclap{\makebox[\widthof{$=$}][#1]{\scriptsize#2}}}{#3}}%
}

\newcommand\toll[2][]{%
	\ensuremath{%
	\ifthenelse{\equal{#1}{}}{%
		T_{#2}%
	}{%
		T_{#2}({#1})%
	}}%
}

\newcommand\istoll[2][]{%
	\ensuremath{%
	\ifthenelse{\equal{#1}{}}{%
		\iscost_{\!#2}%
	}{%
		\iscost_{\!#2}({#1})%
	}}%
}

\newcommand\insertsortcost{W}
\newcommand\iscost{\insertsortcost}

\newcommand\indicator[1]{\indicatornobraces{\{#1\}}}
\newcommand\indicatornobraces[1]{\mathds 1_{#1}}
\newcommand\characteristicvector[2][]{%
	\ifthenelse{\equal{#1}{}}{%
		\mathds1_{#2}%
	}{%
		\ui{\mathds1_{#2}}{#1}%
	}
}

\newcommand\entropy[1][]{%
	\ensuremath{%
		\mathcal H_{#1}%
	}\xspace%
}

\newcommand\QSentropy{%
	\entropy[Q]%
}

\newcommand\discreteEntropy[1][]{%
	\ensuremath{
		\entropy[] \ifthenelse{\equal{#1}{}}{ }{ (#1) }
	}\xspace%
}

\newcommand\contentropy[1][]{%
	\ensuremath{%
		\entropy[\ln]%
		\ifthenelse{\equal{#1}{}}{ }{ (#1) }%
	}\xspace%
}

\newcommand\binaryEntropy[1][]{%
	\ensuremath{
		\entropy[\ld] \ifthenelse{\equal{#1}{}}{ }{ (#1) }%
	}\xspace%
}

\newcommand\convAS{	
	\mathchoice{%
		\mathrel{\overset{\raisebox{0ex}{\text{a.s.}}}\longrightarrow}%
	}{%
		\mathrel{\like{\longrightarrow}{%
			\overset{\raisebox{-1ex}{\kern-1pt\scriptsize\text{a.s.}}}\longrightarrow%
		}}%
	}{%
		\mathrel{\overset{\text{a.s.}}\longrightarrow}%
	}{%
		\mathrel{\overset{\text{a.s.}}\longrightarrow}%
	}%
}

\DeclareRobustCommand\arrayA{%
	\ensuremath{\mathchoice{%
		\smash{\raisebox{-.2pt}{\scalebox{1.25}[1.18]{$\mathtt{A}$}}}%
	}{%
		\smash{\raisebox{-.2pt}{\scalebox{1.25}[1.18]{$\mathtt{A}$}}}%
	}{%
		\smash{\raisebox{-.2pt}{\scalebox{1.25}[1.18]{$\scriptstyle\mathtt{A}$}}}%
	}{%
		\smash{\raisebox{-.2pt}{\scalebox{1.25}[1.18]{$\scriptscriptstyle\mathtt{A}$}}}%
	}}\xspace%
}
\providecommand\arrayA

\newcommand\bmprob[2][]{%
	\ifthenelse{\equal{#1}{}}{%
		\ifthenelse{\equal{#2}{}}{%
			\mathbb{P}_{\mkern-3mu\scriptscriptstyle\textrm{BM}}%
		}{%
			\ui { \mathbb{P}_{\mkern-3mu\scriptscriptstyle\textrm{BM}} } {#2}%
		}%
	}{%
		\ifthenelse{\equal{#2}{}}{%
			\mathbb{P}_{\mkern-3mu\scriptscriptstyle#1}%
		}{%
			\ui { \mathbb{P}_{\mkern-3mu\scriptscriptstyle#1} } {#2}%
		}%
	}%
}

\newcommand\steadystatemissrate[1]{%
	\ifthenelse{\equal{#1}{}}{%
		f
	}{%
		f_{#1}
	}
}

\usepackage{manfnt}

\DeclarePairedDelimiterXPP\dirichletExpectation[2]{\ESymbol_{D(#2)}}[]{}{%
	\renewcommand\given{\nonscript\:\delimsize\vert\nonscript\:\mathopen{}}%
	#1%
}
\providecommand\dirichletExpectation{} %

\newcommand*\virtualizename{\mathit{virtualClass}}
\newcommand*\virtualize[2][]{%
	\ifthenelse{\equal{#2}{}}{%
		\virtualizename_{#1}%
	}{%
		\virtualizename_{#1}(#2)%
	}%
}

\newcommand*\costof[1][]{%
	\ifthenelse{ \equal{#1}{} }{%
		\mathit{cost}%
	}{%
		\mathit{cost}_{#1}%
	}
}
\newcommand*\nextstate[1][]{%
	\ifthenelse{ \equal{#1}{} }{%
		\mathit{next}%
	}{%
		\mathit{next}_{#1}%
	}
}

\let\oldalign\align
\let\endoldalign\endalign

\newcommand*{\numberthis}{\stepcounter{equation}\tag{\theequation}}

\renewenvironment{align}{%
	\begingroup%
	\let\oldhalign\halign
	\def\halign{%
		\let\oldbreak\\%
		\def\nonnumberbreak{\nonumber\oldbreak*}%
		\def\\{%
			\@ifstar{\nonnumberbreak}{\oldbreak}%
		}%
		\oldhalign%
	}
	\oldalign%
}{%
	\endoldalign%
	\endgroup%
}

\newenvironment{galign}[1][]{%
	\subequations%
	#1%
	\def\theequation{\theparentequation{\protect\footnotesize\hspace{.1em}.\arabic{equation}}}%
	\align%
}{%
	\endalign%
	\endsubequations%
}

\let\oldparagraph\paragraph
\renewcommand\paragraph{%
    \@ifstar{\myparagraphStar}{\myparagraphNoStar}%
}
\newcommand\myparagraphStar[1]{%
	\oldparagraph*{#1.}%
}
\newcommand\myparagraphNoStar[2][]{%
	\ifthenelse{\equal{#1}{}}{%
		\oldparagraph[#2]{#2.}%
	}{%
		\oldparagraph[#1]{#2.}%
	}%
}

\colorlet{symmetriccolor}{black!10}
\colorlet{grayedout}{black!50}

\tikzset{
	inner node/.style={
		circle,draw,inner sep=1pt,
		minimum size=1.7em,
		fill=backgroundshading!50,
	},
	leaf node/.style={
		draw,rectangle,minimum size=1.25em,inner sep=1pt,
		fill=backgroundshading!75,
	},
	inner node concise/.style={
		circle,draw,fill,minimum size=.8ex,inner sep=0pt,
	},
	leaf node concise/.style={
		draw,fill,rectangle,minimum size=.5ex,inner sep=0pt,
	},
	leaf node empty/.style={
		draw=none,fill=none,rectangle,minimum height=.5ex,minimum width=.15ex,
	},
	partitioning node/.style={
		draw,inner sep=1pt,minimum height=1.5em,minimum width=1.5em,
		rectangle,rounded corners=2pt,
		fill=backgroundshading!50,
	},
	graphs/comparison tree detailed/.style={
		extended binary tree layout,
		nodes={inner node},
		leaf/.style={leaf node},
		math nodes,
		level distance=0ex, sibling distance=0ex,
		sibling sep=.6em, level sep=.4em,
	},
	graphs/comparison tree detailed compact/.style={
		extended binary tree layout,
		nodes={inner node,minimum size=1.05em,inner sep=0pt},
		leaf/.style={leaf node,minimum size=.6em,font=\scriptsize},
		math nodes,
		level distance=0ex, sibling distance=0ex,
		sibling sep=.6em, level sep=.4em,
	},
	graphs/recursion tree detailed/.style={
		tree layout,
		significant sep=.5em,
		missing nodes get space,
		nodes={partitioning node},
		leaf/.style={partitioning node},
		math nodes,
		level distance=0ex, sibling distance=0ex,
		sibling sep=.8em, level sep=.8em,
	},
	graphs/ternary comparison tree detailed/.style={
		tree layout,
		significant sep=.5em,
		missing nodes get space,
		nodes={inner node,minimum size=1.8em,},
		leaf/.style={leaf node},
		math nodes,
		edge quotes={
			pos=.33,
			draw,circle,fill=backgroundshading,
			inner sep=0pt,
			minimum size=8pt,
			font=\tiny,
			anchor=center,
		},
		level distance=0ex, sibling distance=0ex,
		sibling sep=1.1em, level sep=1.3em,
	},
	graphs/comparison tree concise/.style={
		binary tree layout,
		significant sep=0ex,
		nodes={inner node concise},
		empty nodes,
		leaf/.style={leaf node empty},
		level distance=0ex, sibling distance=0ex,
		sibling sep=.8ex, level sep=.4ex,
	},
}

\tikzset{
	bm edge node/.style={
		draw,
		diamond,
		inner sep=0pt,
		fill=backgroundshading,
		minimum size=8pt,
		font=\tiny,
		pos=.25,
	},
	st/.style={
		circle split,draw,
		inner sep=1pt,font={\footnotesize},
		minimum size=2em,
		fill=backgroundshading!50,
	},
}

\newcommand*\bmnode{\kern-.5pt$\lightning$\kern-1pt}

\usepackage{pbox}

\def\mydots{\xleaders\hbox to.5em{\hfill.\hfill}\hfill}

\newlength\tmpLenNotations
\newenvironment{notations}[1][10em]{%
	\small
	\newcommand\notationentry[1]{%
		\settowidth\tmpLenNotations{##1}%
		\ifthenelse{\lengthtest{\tmpLenNotations > \labelwidth}}{%
			\parbox[b]{\labelwidth}{%
				\makebox[0pt][l]{##1}\\%
			}%
		}{%
			\mbox{##1}%
		}%
		\mydots\relax%
	}%
	\begin{list}{}{%
		\setlength\labelsep{0em}%
		\setlength\labelwidth{#1}%
		\setlength\leftmargin{\labelwidth+\labelsep+1em}%
	}
		\newcommand\notation[1]{\item[{##1}]}%
	\raggedright
}{%
	\end{list}
}

\ifproceedings{
	
	\let\oldsubsection\subsection
	\renewcommand\subsection[1]{%
		\oldsubsection{#1.}
	}

}{
	\usepackage{noindentafter}
	
	\NoIndentAfterEnv{theorem}
	\NoIndentAfterEnv{lemma}
	\NoIndentAfterEnv{proposition}
	\NoIndentAfterEnv{proof}
	\NoIndentAfterEnv{definition}
	\NoIndentAfterEnv{corollary}
	\NoIndentAfterEnv{remark}
}

\newsavebox\tabletmp%
\newlength\tmpheight%
\newlength\tmpdepth%

\let\epsilon\varepsilon

\counterwithout{equation}{section}

\makeatother

%% file: equal-keys-intro.tex
\section{Introduction}

Sorting is one of the basic algorithmic tasks that is used as a fundamental
stepping stone for a multitude of other, more complex challenges.
Quicksort is the method of choice for sorting in practice.
Any undergraduate student learns as a justification for its quality that the average number 
of comparisons is $\Theta(n\log n)$ for a random permutation of $n$ distinct elements, 
putting it into the same complexity class as, \eg, Mergesort.
By choosing random pivots we can make this average the \emph{expected} behavior 
regardless of the order of the input.
Moreover, the hidden constants in $\Theta(n\log n)$ are actually small: 
Quicksort needs $2 \ln(2) n\ld(n) \pm \Oh(n)$
comparisons in expectation, \ie, asymptotically only a factor $2\ln(2)\approx 1.39$ more than the 
information-theoretic lower bound for any comparison-based sorting method.%
\footnote{%
	I~write $\ld$ for $\log_2$\protect\ifproceedings{}{ (\textsl{\underline logarithmus \underline dualis})},
	and $\ln$ for $\log_e$\protect\ifproceedings{}{ (\textsl{\underline logarithmus \underline naturalis})}.%
}
\ifproceedings{\par}{}%
By choosing the pivot as median of $k$ sample elements in each step, 
for $k$ a fixed, odd integer, this constant can be reduced to 
$\alpha_k = \ln(2) / \bigl(\harm{k+1}-\harm{\smash{(k+1)/2}} \bigr)$~\citep{Sedgewick1977},
where \ifproceedings{\smash}{}{$\harm k = \sum_{i=1}^k 1/i$}.
Note that $\alpha_k$ converges to $1$ as $k\to\infty$, so Quicksort's expected behavior 
is asymptotically optimal on random permutations.

\ifproceedings{}{
	It is also well-known that Quicksort rarely deviates much from this expected behavior.
	Folklore calculations show that Quicksort needs, \eg, at most 
	$7$ times the expected number of comparisons with probability at least $1-1/n^2$.
	This result can be strengthened~\citep{McDiarmidHayward1996} 
	to guarantee at most $\alpha$ times the expected costs for any $\alpha>1$
	with probability $1-\Oh(n^{-c})$ for all constants $c$ (independent of $\alpha$).
	Median-of-$k$ Quicksort is hence optimal not only in expectation, but almost always!
}

These guarantees certainly justify the wide-spread use of Quicksort, 
but they only apply to inputs with $n$ \emph{distinct} elements.
In many applications, duplicates (\ie, elements that are equal \wrt the order relation) 
appear naturally. 
The SQL clause \texttt{GROUP BY C}, for example, 
can be implemented by first sorting rows \wrt column \texttt{C}
to speed up the computation of aggregating functions like \texttt{COUNT} or \texttt{SUM}.
Many rows with equal \texttt{C}-entries are expected in such applications.

At the \textsl{KnuthFest} celebrating Donald Knuth's $1\.000\.000_2$th birthday,
Robert Sedgewick gave a talk titled 
\textsl{``Quicksort is optimal''}~\citep{SedgewickBentley2002},%
\footnote{%
	\strut{}Sedgewick presented the conjecture in a less widely-known talk already 
	in 1999~\citep{SedgewickBentley1999}.
	Since they never published their results in an article,
	I briefly reproduce their arguments in \wref{sec:previous-work}.
}
presenting a result that is at least very nearly so:\linebreak[3]
Quicksort with \textsl{fat-pivot} (\aka \textsl{three-way}) 
partitioning uses $2\ln 2 \approx 1.39$ times
the number of comparisons needed by \emph{any} comparison-based sorting method 
for \emph{any} randomly ordered input, with or without equal elements.
He closed with the conjecture that this factor can be made arbitrarily close to $1$
by choosing the pivot as the median of $k$ elements for sufficiently large~$k$.
This statement is referred to as the Sedgewick-Bentley conjecture.

The Sedgewick-Bentley conjecture is a natural 
generalization of the distinct-keys case (for which the statement is known to hold true),
and sorting experts will not find the claim surprising, since
it is the theoretical justification for a long-established best practice 
of general-purpose sorting:
as observed many times in practice,
Quicksort with the fat-pivot partitioning method by \citet{Bentley1993}%
\footnote{%
	The\ifproceedings{}{ classic} implementation by \citet{Bentley1993} from 1993 is used%
	\ifproceedings{}{ as default sorting methods} in important programming libraries, \eg, 
	the GNU C++ Standard Template Library (STL) and the Java runtime library (JRE).
	Since version 7, the JRE uses dual-pivot Quicksort instead, 
	but if the two sampled pivots are found \emph{equal}, 
	it falls back to fat-pivot partitioning by \citeauthor{Bentley1993}.
}
and \textsl{ninther} (\aka \textsl{pseudomedian of nine}) pivot sampling
comes to within \emph{ten percent} of the optimal expected comparison count 
for inputs with or without equal keys.

In this paper, I confirm the Sedgewick-Bentley conjecture for inputs with ``many duplicates'',
\ie, where every key value occurs $\Omega(n^\epsilon)$ times
for an arbitrary constant $\epsilon>0$.
To do so I show a bound on the constant of proportionality, namely
the same $\alpha_k$ mentioned above.
(Sedgewick and Bentley did not include a guess 
about the constant or the speed of convergence in their conjecture).
While Quicksort can certainly be outperformed for particular types of inputs,
the combination of simple, efficient code and almost universal proven optimality is unsurpassed;
it is good news that the latter also includes inputs with equal keys.

Confirming the Sedgewick-Bentley conjecture is not a surprising outcome,
but it has been surprisingly resistant to all attempts to formally prove it: 
no progress had been made in the nearly two decades since it was posed
(see \wref{sec:main-result} for some presumable reasons for that),
so restricting the problem might be sensible.
I would like to remark that my restriction of many duplicates 
is a technical requirement of the \emph{analysis},
but we have no reason to believe that Quicksort performs much different 
when it is violated.
More importantly, it is \emph{not} the purpose of this paper to suggest 
sorting inputs with many duplicates as a natural problem per se, 
for which tailored methods shall be developed.
(One is tempted to choose a hashing-based approach
when we \emph{know} that the number $u$ of different values is small,
but my results show that precisely for such instances, Quicksort will 
be competitive to any tailored algorithm!)
It must be seen as an idiosyncrasy of the present analysis one
might try to overcome in the future.

\ifproceedings{}{
	It is a strength of Quicksort is that it smoothly \emph{adapts} to the actual amount
	of duplication without requiring explicit detection and special handling of that case.
	The concept is analogous (but orthogonal) to \textsl{adaptive} sorting 
	methods~\citep{EstivillCastro1992} that similarly take 
	advantage of existing (partial) \emph{order} in the input.
	The purpose of this paper is thus to finally deliver a mathematical proof that
	median-of-$k$ Quicksort is an optimal ``entropy-adaptive''
	(\aka \textsl{distribution sensitive}~\citep{SandeepGupta1999})
	sorting method, at least for a large class of inputs.
}

\paragraph{Methods}
For the analysis we will work in an alternative input model:
instead of fixing the \emph{exact} multiplicities of a multiset, 
we fix a \emph{discrete probability distribution} over $[1..u]$
and draw $n$ elements independently and identically distributed (\iid) 
according to that distribution.
We will see that going from fixed to expected multiplicities 
only increases the average sorting costs, so the \iid model provides
an upper bound.

The analysis of Quicksort on discrete \iid inputs proceeds in two steps: 
First, I use that costs of (median-of-$k$) Quicksort correspond
to costs of searching in a ($k$-fringe-balanced) binary search tree (BST).
A~concentration argument shows that for inputs with many duplicates,
the search tree typically reaches a stationary state 
after a short prefix of the input,
and is thus independent of the (larger) rest of the input.

The second step is to determine the expected search costs
in a $k$-fringe-balanced BST built by repeatedly inserting \iid elements
until all $u$ values appear in the tree; (I~call the tree \textsl{saturated} then).
I show that the search costs are asymptotically $\alpha_k$ times the \textsl{entropy}
of the universe distribution.
To do so I derive lower and upper bounds from the recurrence of costs 
(with the \emph{vector} of probability weights as argument)
using the aggregation property of the entropy function.
To the best of my knowledge, the analysis of search costs in fringe-balanced trees 
with equal keys is novel, as well.

Most used techniques are elementary, but keeping the error terms under control
(so that $\Omega(n^\epsilon)$ duplicates suffice for any $\epsilon>0$) 
made some refinements of classical methods necessary 
that might be of independent interest.

\paragraph{Outline}
We start with some notation and a collection of known results (\wref{sec:preliminaries}),
and introduce the input models (\wref{sec:input-models}).
\wref{sec:previous-work} is devoted to related work on sorting in the presence 
of equal elements.
We continue in \wref{sec:main-result} with a formal statement of the 
main result and 
a briefly discussion why some previous attempts did not succeed.
In \wref[Sections]{sec:equals-quicksort-and-search-trees}{}\,--\,\ref{sec:expected-node-depth-asymptotic},
we study the cost of fat-pivot Quicksort on discrete \iid inputs as outlined 
above%
\ifproceedings{%
	.
}{%
	: \wref{sec:equals-quicksort-and-search-trees} demonstrates
	that it suffices to study search costs in trees to analyze Quicksort,
	and \wref{sec:saturated-trees} gives a stochastic description of these trees.
	In \wref{sec:equals-separating-n-and-q}, we exploit the assumption of
	many equal keys to separate the influence of the input size from the influence
	of different universe distributions (the ``first step'' from above).
	In \wref{sec:expected-node-depth-asymptotic} we then derive an asymptotic
	approximation for the search costs in fringe-balanced trees (second step).
}%
We establish a lower bound for the \iid model in \wref{sec:lower-bound}.
\wref{sec:proof-main-result} proves the main result by combining all pieces.
The results are summarized and put into context in \wref{sec:conclusion}.

\ifproceedings{
	There is an extended version of this paper available online
	\href{https://arxiv.org/abs/1608.04906}{\texttt{arxiv.org/abs/1608.04906}}
	that contains full proofs and more detailed comments.
}{%
	\wref{app:notations} collects all notation used in this paper for reference.
	In \wref{app:height-of-recursion-trees}, we formally prove
	that fringe-balanced search trees have logarithmic height with high probability
	in the presence of equal keys.
}

%% file: equal-keys-preliminaries.tex
\section{Preliminaries}
\label{sec:preliminaries}

We start with some common notation and 
a precise description of fat-pivot Quicksort with median-of-$k$ sampling,
as well as the corresponding a search tree variant, the $k$-fringe-balanced trees.
We collect further known results that will be used later on%
\ifproceedings{%
	.
}{%
	here for the reader's convenience: some classical tail bounds (\wref{sec:preliminaries-tail-bounds}),
	a fact on the height concentration of randomly built search trees 
	(\wref{sec:preliminaries-log-height-whp}),
	and a few properties of the entropy function used later on
	(\wref{sec:preliminaries-properties-entropy}).
}

\subsection{Notation}
\label{sec:notation}

\ifproceedings{}{
	This section contains the most important notation;
	a comprehensive list is given in \wref{app:notations} for reference.
}

I~write vectors in bold font $\vect x = (x_1,\ldots,x_n)$ and always understand them
as column vectors (even when written in a row in the text).
By default, all operations on vectors are meant component-wise.
By $\total{\vect x}$, I~denote $x_1+\cdots+x_n$, the \emph{total} of $\vect x$.
The \emph{profile} (or multiplicities vector) a multiset $M$ over $[u] = \{1,\ldots,u\}$
is a vector $\vect x = (x_1,\ldots,x_u)$ where every values $v\in[u]$ occurs $x_v$
times in $M$.

I use Landau-symbols ($\Oh$, $\Omega$ etc.)\ in the formal sense of \citet[Section A.2]{Flajolet2009}
and write $f(n) = g(n) \pm \Oh(h(n))$ to mean $|f(n) - g(n)| = \Oh(h(n))$.
(I use $\pm$ to emphasize the fact that we specify $f(n)$ up to an additive error of $h(n)$ 
without restricting the sign of the error.)

For a random variable $X$, $\E{X}$ is its expectation and $\Prob{X=x}$ denotes the probability of 
the event $X=x$. 
For $\vect q\in[0,1]^u$ with $\total{\vect q}=1$, 
I~write $U\eqdist \distFromWeights{\vect q}$ to say that $U$ is a 
random variable with $\Prob{U=i} = q_i$ for $i\in[u]$;
generally, $\eqdist$ denotes equality in distribution.
$\multinomial(n,\vect q)$ is a \textsl{multinomial distributed} variable
with $n$ trials from $\distFromWeights{\vect q}$, and
$\betadist(\alpha,\beta)$ is a \textsl{beta distributed} variable with parameters $\alpha,\beta>0$.
By $\entropy$ or $\entropy[\ld]$ I~denote the 
binary \textsl{Shannon entropy} $\entropy(\vect q) = \sum_{i=1}^u q_i \ld(1/q_i)$;
likewise $\entropy[\ln]$ is the base $e$ entropy.

I say an event $E=E(n)$ occurs \emph{with high probability (\whp)}
if for \emph{any} constant $c$ holds
$\Prob{E(n)} = 1 \pm \Oh(n^{-c})$ as $n\to\infty$.
Note that other sources use \whp if any such $c\in\N$ exists, 
which is a much weaker requirement.
Similarly for $X_1,X_2,\ldots$ a sequence of real random variables,
I say ``$X_n = \Oh(g(n))$ \whp'' if for every constant $c$ there is a constant $d$ so that
$\Prob{|X_n| \le d |g(n)|} = 1 \pm \Oh(n^{-c})$ as $n\to\infty$.

\subsection{Reference Version of Fat-Pivot Quicksort}
\label{sec:fat-pivot-quicksort}

By ``fat-pivot'' I mean a partitioning method that splits the input into three parts:
elements (strictly) smaller than the pivot, elements equal to the pivot, and
elements (strictly) larger than the pivot.
Instead of only one pivot element in binary partitioning, we now obtain a ``fat'' pivot segment 
that separates the left and right subproblems.
This idea is more commonly known as \textsl{three-way partitioning,}
but I prefer the vivid term fat pivot to make the distinction between fat-pivot partitioning and 
partitioning around \emph{two pivots} clear.
The latter has become fashionable again in recent years~\citep{javacoredevel2009,Wild2012,Wild2016}.
Since all duplicates of the pivot are removed before recursive calls,
the number of partitioning rounds is bounded by 
the number of different values in the input.

To simplify the presentation, I~assume in this work that partitioning 
\emph{retains the relative order} of elements that go to the same segment.
A~corresponding reference implementation operating on linked lists 
is given in \wref{alg:list-quicksort}.
It uses the median of $k$ sample elements as pivot, where 
the sample size $k=2t+1$, $t\in\N_0$, is a tuning parameter of the algorithm. 
We consider $k$ a fixed constant in the analysis.

\begin{algorithm}[tbhp]
	\caption{%
		Fat-pivot median-of-$k$ Quicksort.%
	}
	\label{alg:list-quicksort}
	\vspace*{-2ex}
	\small
	\begin{codebox}
	\Procname{$\proc{Quicksort}_{k}(L)$}
	\li	\If $(L.\id{length}) \le k - 1$
	\li	\Then
			\Return $\proc{Insertionsort}(L)$
		\EndIf
	\li	$P \gets \proc{Median}(L[1..k])$
	\li	$L_1,L_2,L_3 \gets \text{new List}$
	\li	\While $\lnot\, L.\id{empty}$
	\li	\Do
			$U \gets L.\id{removeFirstElement}()$
	\li		\kw{case distinction on} $\id{cmp}(U,P)$
			\Do
	\li			\kw{in case} $<$ \kw{do} $L_1.\id{append}(U)$
	\li			\kw{in case} $=$ \kw{do} $L_2.\id{append}(U)$
	\li			\kw{in case} $>$ \kw{do} $L_3.\id{append}(U)$
			\End
	\li		\kw{end cases}
		\EndWhile
	\li	$L_1 \gets \proc{Quicksort}_{k}(L_1)$
	\li	$L_3 \gets \proc{Quicksort}_{k}(L_3)$
	\li	\Return $L_1.\id{append}(L_2).\id{append}(L_3)$
	\end{codebox}
	\vspace*{-3ex}
\end{algorithm}

I always mean \wref{alg:list-quicksort} when speaking of Quicksort in this paper.
Moreover, by its costs I always mean the number of ternary key comparisons,
\ie, the number of calls to \id{cmp}.
Apart from selecting the median,
one partitioning step in \wref{alg:list-quicksort} uses exactly $n$
ternary comparisons.
\ifproceedings{%
}{
	Some of these will be redundant because they already happened while
	determining the median of the sample.
	We ignore this optimization here.
}

\wref{alg:list-quicksort} serves as precise specification for the analysis
but it is not the method of choice in practice.
However, our results apply to practical methods, as well
(see the discussion in \wref{sec:recursion-trees}).

\subsection{Quicksort and Search Trees}
\label{sec:quicksort-search-trees}

It has become folklore that the comparison costs of (classic) Quicksort are the same as
the internal path length of a BST, 
which equals the cost to construct the tree by successive insertions.
\Citet{Hibbard1962} first described this fact in 1962, right after Hoare's 
publication of Quicksort itself~\citep{Hoare1961a,Hoare1962}.
Formally we associate a \emph{recursion tree} to each execution of Quicksort:
Each partitioning step contributes a node labeled with the used pivot value.
Its left and right children are the recursion trees of the 
left and right recursive calls, respectively.
\wref{fig:equals-recursion-tree-simple} shows an example.

\begin{figure}[tbhp]
	\begin{captionbeside}{%
		Execution trace of Quicksort without pivot sampling
		and its recursion tree.
		\protect\ifproceedings{}{%
			The recursion tree coincides with the BST obtained by successively
			inserting the original input.%
		}%
		An animated version of this example%
		\protect\ifproceedings{}{ (and its extensions to finge-balancing and equal keys)}
		is available online:
		\href{https://youtu.be/yi6syj9nksk}{\texttt{youtu.be/yi6syj9nksk}}.
		\vspace{2ex}%
	}
	\newcommand\drawarray[4]{
		\begin{scope}[shift={(#1,-1.5*#2)}]
		\draw[fill=backgroundshading!50] (0.5,0) rectangle +(#3,1) ;
		\foreach \x/\u in {#4} {
			\node at (\x,0.5) {$\u$} ;
		}
		\end{scope}
	}%
	\newcommand\drawpivot[2]{
		\node[headingscolor] (P#1) at (#1,-1.5*#2+0.5) {$#1$} ;
		\node[draw,very thick,circle,headingscolor,inner sep=2.5pt] 
			(#1) at (#1+9.5,-1.5*#2+1) {$#1$} ;
	}%
	\plaincenter{%
	\begin{tikzpicture}[xscale=.45,yscale=.45,every node/.append style={font=\scriptsize}]
		\drawarray009{1/7,2/4,3/2,4/9,5/1,6/3,7/8,8/5,9/6}
		\drawarray016{1/4,2/2,3/1,4/3,5/5,6/6}
		\drawarray712{1/9,2/8}
		\drawpivot71
		\drawarray023{1/2,2/1,3/3}
		\drawarray422{1/5,2/6}
		\drawarray721{1/8}
		\drawpivot42
		\drawpivot92
		\drawarray031{1/1}
		\drawarray231{1/3}
		\drawarray531{1/6}
		\drawpivot23
		\drawpivot53
		\drawpivot83
		\drawpivot14
		\drawpivot34
		\drawpivot64
		
		\draw[thick,headingscolor]
			(7) -- (4) -- (2) -- (1)
			(7) -- (9) -- (8)
			(4) -- (5) -- (6)
			(2) -- (3)
		;
		\begin{pgfonlayer}{background}
			\draw[very thick,headingscolor!25!backgroundshading]
				(P7) -- (P4) -- (P2) -- (P1)
				(P7) -- (P9) -- (P8)
				(P4) -- (P5) -- (P6)
				(P2) -- (P3)
			;
		\end{pgfonlayer}
	\end{tikzpicture}%
	}\hspace*{-2em}
	\ifproceedings{\label{fig:equals-recursion-tree-simple}}{}
	\end{captionbeside}
	\ifproceedings{}{\label{fig:equals-recursion-tree-simple}}
\end{figure}

Recursion trees obviously fulfill the search tree property;
in fact, the recursion tree is \emph{exactly} the binary search tree that results 
from successively inserting the input elements into an initially empty tree 
(in the order they appear in the original input)
if Quicksort always uses the first element of the list as pivot
and partitions so that the relative order of elements smaller resp.\ larger 
than the pivot is retained (as is done in \wref{alg:list-quicksort}).
Even the \emph{same set} of comparisons 
is then used in both processes, albeit in a different order.

\subsection{Fringe-Balanced Trees}
\label{sec:fringe-balanced-trees}

The correspondence extends to median-of-$k$ Quicksort with 
an appropriate \textsl{fringe-balancing rule} for BSTs.
This is a little less widely known, but also well researched~\citep{Drmota2009}.
Upon constructing a $k$-fringe-balanced search tree, 
we \emph{collect} up to $k-1$ elements in a leaf.
(This corresponds to truncating the recursion in Quicksort 
for $n \le k-1$ and leave small subproblems for Insertionsort.)
Once a leaf has collected $k = 2t+1$ elements, it is \emph{split:}
Simulating the pivot selection process in Quicksort,
we find the median from the $k$ elements in the leaf and use it as 
the label of a new inner node.
Its children are two new leaves with the elements that did not become pivots:
the $t$ smaller elements go to the left, the $t$ larger to the right.
Because of the dynamic process of splitting leaves, the correspondence
is best shown in motion; the supplementary animation
(\href{https://youtu.be/yi6syj9nksk}{\texttt{youtu.be/yi6syj9nksk}})
includes fringe-balanced trees.

More formally,
a $k$-fringe-balanced tree, for $k=2t+1$ an odd integer, 
is a binary search tree whose leaves can store between $0$ and $k-1$ elements.
An empty fringe-balanced tree is represented as $\mathit{Leaf}()$, a single empty leaf.
The $k$-fringe-balanced tree $\mathcal T$ 
corresponding to a sequence of elements $x_1,\ldots,x_n$
is obtained by successively inserting the elements into an initially empty tree
using \wref{alg:fb-insert};
more explicitly, 
with $T_0 = \mathit{Leaf}()$ and $T_i = \proc{Insert}_k(\mathcal T_{i-1},x_i)$ for $1\le i\le n$,
we have $\mathcal T = \mathcal T_n$.

\begin{algorithm}[tbhp]
	\caption{%
		Insert into $k$-fringe-balanced tree.%
	}%
	\label{alg:fb-insert}
	\vspace*{-2ex}
	\small
	\begin{codebox}
	\Procname{$\proc{Insert}_{k}(\mathcal T, x)$}
	\li	\If $\mathcal T$ is $\mathit{Leaf}(\vect U)$
	\li	\Then
			Append $x$ to $\vect U$ \label{lin:stwl-insert-append-to-leaf}
	\li		\If $|\vect U| \le k - 1$ \kw{then}
				\Return $\mathit{Leaf}(\vect U)$ 
			\kw{end if}
	\zi		\Comment Else: Split the leaf
	\li		$P \gets \proc{Median}(U_1,\ldots,U_k)$
	\li		$C_1,C_2 \gets $ new empty list
	\li		\Foreach $U$ in $\vect U$%
	\li		\Do
				\kw{case distinction} on $\id{cmp}(U,P)$
				\label{lin:stwl-insert-classify-upon-split}
				\Do
	\li				\kw{in case} $U < P$ \kw{do} 
						append $U$ to $C_1$
	\li				\kw{in case} $U > P$ \kw{do} 
						append $U$ to $C_2$
	\zi				\Comment In case $U\isequal P$, we drop $U$.
				\End
	\li			\kw{end cases}
				\EndFor
	\li		\Return $\mathit{Inner}\bigl( P, \mathit{Leaf}(C_1), \mathit{Leaf}(C_2) \bigr)$
	\li \Else 
		 	$\mathcal T$ is $\mathit{Inner}(P,\mathcal T_1,\mathcal T_2)$
	\li		\kw{case distinction} on $\id{cmp}(x,P)$
				\label{lin:stwl-insert-classify-search}
			\Do
\ifproceedings{
	\li			\kw{in case} $x \isequal P$ \kw{do} 
				\Do
	\li				\kern-.5em
					\Return $\mathcal T$ 
					\quad \Comment tree unchanged
				\End
				\label{lin:stwl-insert-duplicate-discarded}
	\li			\kw{in case} $x \mathrel{\like{\isequal}{<}} P$ \kw{do} 
				\Do
	\li				\kern-.5em
					\Return $\mathit{Inner}\bigl(
						P, \proc{Insert}_k(\mathcal T_1,x),\mathcal T_2 
					\bigr)$
				\End
	\li			\kw{in case} $x \mathrel{\like{\isequal}{>}} P$ \kw{do} 
				\Do
	\li				\kern-.5em
					\Return $\mathit{Inner}\bigl(
						P, \mathcal T_1 , \proc{Insert}_k(\mathcal T_2,x)
					\bigr)$
				\End
			\End
}{
	\li			\kw{in case} $x \isequal P$ \kw{do} 
					\Return $\mathcal T$ 
					\quad \Comment tree unchanged
				\label{lin:stwl-insert-duplicate-discarded}
	\li			\kw{in case} $x \mathrel{\like{\isequal}{<}} P$ \kw{do} 
					\Return $\mathit{Inner}\bigl(
						P, \proc{Insert}_k(\mathcal T_1,x),\mathcal T_2 
					\bigr)$
	\li			\kw{in case} $x \mathrel{\like{\isequal}{>}} P$ \kw{do} 
					\Return $\mathit{Inner}\bigl(
						P, \mathcal T_1 , \proc{Insert}_k(\mathcal T_2,x)
					\bigr)$
			\End
}
	\li		\kw{end cases}
		\EndIf
	\end{codebox}
	\vspace*{-3ex}
\end{algorithm}

Growing trees with \wref{alg:fb-insert} enforces a certain balance upon 
the lowest subtrees, 
\ie, at the \emph{fringe} of the tree,
hence the name \emph{fringe-balanced}.
Searching an element in a fringe-balanced tree works as in an ordinary BST,
except for the leaves, where we sequentially search through the buffer%
\ifproceedings{%
	.
}{%
	;\wref{alg:fb-search} shows pseudocode for completeness.
}%
(Note that elements collected in leaves are not sorted.)

\ifproceedings{}{
	\begin{algorithm}[tbhp]
		\caption{%
			Search in fringe-balanced trees.%
		}%
		\label{alg:fb-search}
		\vspace*{-2ex}
		\small
		\begin{codebox}
		\Procname{$\proc{Search}(\mathcal T, x)$}
		\li	\If $\mathcal T$ is $\mathit{Leaf}(\vect U)$
		\li	\Then
			\Return $\proc{SequentialSearch}(\vect U, x)$
		\li \Else 
				$\mathcal T$ is $\mathit{Inner}(P, \mathcal T_1, \mathcal T_2 )$
		\li		\kw{case distinction} on $\id{cmp}(U,P)$
				\Do
		\li			\kw{in case} $U \isequal P$ \kw{do} \Return ``Found''
		\li			\kw{in case} $U \mathrel{\like{\isequal}{<}} P$ \kw{do} 
						\Return $\proc{Search}( \mathcal T_1,x )$
		\li			\kw{in case} $U \mathrel{\like{\isequal}{>}} P$ \kw{do} 
						\Return $\proc{Search}( \mathcal T_2,x )$
				\End
		\li		\kw{end cases}
			\EndIf
		\end{codebox}
		\vspace*{-3ex}
	\end{algorithm}
}

\paragraph{Duplicate Insertions}
Since our analysis of Quicksort with equal keys builds 
on the precise performance of fringe-balanced trees,
we put a particular emphasis is on the treatment of duplicate insertions
in \wref{alg:fb-insert}.
If a key $x$ is already present during insertion and
$x$ appears as key of an \emph{inner} node, the (repeated) insertion has no effect;
but if $x$ is found in a \emph{leaf}, another copy of $x$ is appended to the buffer of that leaf.

\ifproceedings{}{
	This unequal treatment might seem peculiar at first sight, 
	but does exactly what we need:
	duplicates do play a role for selecting pivots
	\parenthesisclause{likely values contribute more duplicates 
	to a random sample and are thus more likely to be selected as pivot}
	but once a pivot has been selected, all its copies are removed in this
	single partitioning step no matter how many there are.
}

%% file: equal-keys-stochastic-preliminaries.tex
\subsection{Tail Bounds}
\label{sec:preliminaries-tail-bounds}

\ifproceedings{
	We recall the classical Chernoff bound.
}{
	For the reader's convenience, we collect a few basic tail bounds here: 
	a classical and a less known Chernoff concentration bound, 
	and a bound for the far end of the lower tail of the binomial distribution $\binomial(n,p)$.
}

\begin{lemma}[Chernoff Bound]
\label{lem:chernoff-bound-binomial}
	Let \(X\eqdist\binomial(n,p)\) for \(n\in\N\) and \(p\in(0,1)\) and let \(\delta\ge0\). Then
	\begin{align*}
			\Prob[\Bigg]{ \biggl|\frac{X}n - p\biggr| \ge \delta }
		&\wwrel\le
			2\exp (-2 \delta^2 n ).
	\end{align*}
\qed\end{lemma}
This bound appears, \eg, as Theorem~2.1 of \citet{McDiarmid1998}.
\ifproceedings{%
}{%

	The following lemma is a handy, but less well-known bound for the \emph{multinomial} distribution
	that appears \parenthesisclause{indeed rather hidden} as Lemma~3
	in a paper by \citet{Devroye1983} from 1983.
	(Its proof is also discussed on \textsl{math stack exchange:} 
	\href{http://math.stackexchange.com/q/861058}{\texttt{math.stackexchange.com/q/861058}}.)

	\begin{lemma}[Chernoff Bound for Multinomial]
	\label{lem:chernoff-bound-multinomial}
		Let \(\vect X \eqdist \multinomial(n,\vect p)\) for \(n\in\N\) and \(\vect p \in (0,1)^u\)
		with \(\total{\vect p} = 1\).
		Further, let \(\delta\in(0,1)\) with \(\delta\ge \sqrt{ 20u/n }\) be given.
		Then
		\begin{align*}
				\Prob[\Bigg]{\.\sum_{i=1}^u \biggl|\frac{X_i}n - p_i\biggr| \rel\ge \delta }
			&\wwrel\le
				3 \exp( - \delta^2 n /25).
		\end{align*}
	\qed\end{lemma}

}%
\ifproceedings{%
	We also use the following elementary observation;
	a proof appears in the extended version of this paper.
}{
	Finally, we also use the following elementary observation.
}

\begin{lemma}[Far-End Left-Tail Bound]
\label{lem:binomial-not-very-small-whp}
\ifproceedings{~\\}{}
	Let $\ui Xn \eqdist \binomial(n,\ui pn)$ be a sequence of random variables and $k\in\N$
	a constant, where $\ui pn$ satisfies
	$\ui pn = \omega\bigl(\frac{\log n}{n}\bigr)$ as $n\to\infty$
	and is bounded away from $1$, \ie, there is a constant $\epsilon>0$ so that 
	$\ui pn \le 1-\epsilon$ for all $n$.
	Then $\Prob{\ui Xn \le k}  = o(n^{-c})$ as $n\to\infty$ for any constant $c$.
\end{lemma}

\ifproceedings{}{
	\begin{proof}%
	For better readability, we drop the superscript from $\ui pn$ when $n$ is clear from the context.
	Let $c$ be an arbitrary constant.
	\begin{align*}
			\Prob[\big]{\ui Xn\le k} \cdot n^c
		&\wwrel=
			n^c \sum_{i=0}^k \binom ni p^i (1-p)^{n-i}
	\\	&\wwrel\le
			n^c \sum_{i=0}^k \frac 1{i!} \underbrace{\biggl(\frac{p}{1-p}\biggr)^i}
					_{ \le (\frac{1-\epsilon}\epsilon)^i } 
				n^i (1-p)^n
	\\	&\wwrel\le
			n^{c+k} (1-p)^n \cdot \Oh(1)
	\\	&\wwrel=
			\exp\Bigl(n \ln(1-p) + (c+k) \ln(n) \Bigr)  \cdot \Oh(1)
	\shortintertext{using $\ln(x) \le x-1$, this is}
		&\wwrel\le
			\exp\Bigl(- np \pm \Oh(\log n) \Bigr)  \cdot \Oh(1)
	\\	&\wwrel\to 
			0
	\end{align*}
	since $p = \omega(\frac{\log n } n)$.
	This proves the claim.
	\end{proof}
	
	The requirement that $\ui pn$ is bounded away from $1$ can be lifted, 
	but it simplifies the proof. 
	For the application in the present paper the version is sufficient as is.
}

\subsection{Logarithmic Tree Height \whp}
\label{sec:preliminaries-log-height-whp}

It is a folklore theorem that randomly grown search trees have logarithmic height with high probability.
For the present work, the following result is sufficient.

\begin{proposition}[\ifproceedings{}{Probability of }Height-Degeneracy]
\label{pro:tree-log-height-weak-whp}
\ifproceedings{~\\}{}
	For any fixed (odd) $k$ holds: 
	the probability that a $k$-fringe-balanced search tree built from 
	$n$ randomly ordered elements (with or without duplicates)
	has height $>13 \ln n$ is in $\Oh(1/n^2)$ as $n\to\infty$.
\qed\end{proposition}
\ifproceedings{%
	The extended version of this paper gives a formal proof of \wref{pro:tree-log-height-weak-whp}
	using good-split-bad-split indicators and \wref{lem:chernoff-bound-binomial}.
	Curiously there is a technical pitfall in this proof
	that has not been rigorously addressed before (to my knowledge).
	
}{%
	We give a formal proof of \wref{pro:tree-log-height-weak-whp} and a more general discussion
	in \wref{app:height-of-recursion-trees}.
}

%% file: equal-keys-properties-of-entropy.tex
\subsection{Properties of the Entropy Function}
\label{sec:preliminaries-properties-entropy}

This section collects a few useful properties of the entropy function.
\ifproceedings{
	Proofs are given in the extended version.
}{
	We start with some observations that follow from well-known results of real analysis.
	Proofs are given in my \phdthesis~\citep{Wild2016}.
}

\begin{lemma}[\ifproceedings{Properties of Entropy}{Elementary Properties of the Entropy Function}]
\label{lem:entropy-function}
	~\\Let \(\entropy[\ln]: [0,1]^u \to \R_{\ge0}\) with 
	\(\entropy[\ln](\vect x) = \sum_{i=1}^u x_i \ln(1/x_i)\) be
	the base \(e\) entropy function.
	\begin{thmenumerate}{lem:entropy-function}
	\item \(\entropy[\ln](\vect x) = \ln(2) \entropy[\ld](\vect x)\).
	\item \label{lem:entropy-function-bounds}
		For all \(\vect x\in[0,1]^u\) with \(\total{\vect x}=1\) we have that
		\(0\le \entropy[\ln](\vect x) \le \ln(u)\).
	\item \label{lem:entropy-function-hölder}
		\(\entropy[\ln]\) is Hölder-continuous in \([0,1]^u\) for any exponent \(h\in(0,1)\),
		\ie, there is a constant \(C=C_h\) such that 
		\(|f(\vect y) - f(\vect x)| \le C_h \.u \cdot \|\vect y - \vect x\|_{\infty}^h\) for all 
		\(\vect x, \vect y \in [0,1]^u \).
		
		A possible choice for \(C_h\) is given by 
		\begin{align*}
		\numberthis\label{eq:C-h}
				C_h 
			&\wwrel= 
				\biggl( \int_0^1 \bigl|\ln(t)+1\bigr|^{\frac1{1-h}} \biggr)^{1-h}
		\end{align*}
		For example, \(h=0.99\) yields \(C_h \approx 37.61\).
	\end{thmenumerate}
\qed\end{lemma}

\noindent
A beta distributed variable $\Pi$ can be seen as a \emph{random} probability.
We can then ask for the \emph{expected} entropy of a Bernoulli trial with probability $\Pi$.
\ifproceedings{%
	Using that
		\begin{align*}
		\numberthis\label{eq:logarithmic-beta-integral}
		\begin{multlined}
				\int_0^1 z^{a-1} (1-z)^{b-1} \ln(z) \, dz
		\\[-1ex]	
			\wwrel=
				\BetaFun(a,b) \bigl( \psi(a) - \psi(a+b) \bigr)
				,\quad (a,b > 0) ,
		\end{multlined}
		\end{align*}
	(which is Equation~(4.253-1), p.\,540, of \citet{Gradshteyn2007} for $r=1$),
	we obtain the answer.
	(Full details are given in the extended version.)
}{%
	(The special case of \wref{lem:expected-entropy} when $t=0$
	appears Section 5.0 of
	\Citet{Bayer1975} and as Exercise 6.2.2--37 of \citet{Knuth1998}.)
}

\begin{lemma}[\ifproceedings{Expected Entropy Beta Variables}{Expected Entropy of Beta Variables}]
\label{lem:expected-entropy}
\ifproceedings{}{}
	Let $t\in\N_0$.
	For $\Pi\eqdist \betadist(t+1,t+1)$ and $k=2t+1$ we have
	\ifproceedings{\begin{align*}\SwapAboveDisplaySkip\qquad}{\begin{align*}}
	\numberthis\label{eq:expected-entropy}
			\E[\big]{\entropy[\ln](\Pi,1-\Pi)}
		&\wwrel=
			\harm{k+1}-\harm{t+1} .
	\end{align*}
\end{lemma}

\ifproceedings{}{
	\begin{proof}
	We need the following integral,
	which is a special case of Equation~(4.253-1), p.\,540, of \citet{Gradshteyn2007} with $r=1$:
	\begin{align*}
	\numberthis\label{eq:logarithmic-beta-integral}
			\int_0^1 z^{a-1} (1-z)^{b-1} \ln(z) \, dz
		\wwrel=
			\BetaFun(a,b) \bigl( \psi(a) - \psi(a+b) \bigr)
			,\qquad (a,b > 0) .
	\end{align*}
	Here $\psi(z) = \frac d{dz} \ln(\Gamma(z))$ is the digamma function. 
	
	The proof is now simply by computing. By symmetry we have 
	$\E{\entropy[\ln](\Pi,1-\Pi)} = -2\E{\Pi \ln(\Pi)}$;
	using the above integral and the relation $\psi(n+1) = \harm n - \gamma$
	(\href{http://dlmf.nist.gov/5.4.E14}{Equation~(5.4.14)} of the \citet{DLMF}) 
	we find that
	\begin{align*}
			\E{\Pi \ln(\Pi)}
		&\wwrel=
			\int_0^1 x \ln(x)  \frac{x^{t}(1-x)^{t}}
					{\BetaFun(t+1,t+1)} \, dx
	\\	&\wwrel=
			\frac{\BetaFun(t+2,t+1)}{\BetaFun(t+1,t+1)} 
				\int_0^1 \ln(x)  \frac{x^{t+1}(1-x)^{t}}
					{\BetaFun(t+2,t+1)} \, dx
	\\	&\wwrel{\eqwithref{eq:logarithmic-beta-integral}}
			\frac{t+1}{k+1} 
				\bigl(\psi(t+2)-\psi(k+2)\bigr)
	\\	&\wwrel=
			\frac{t+1}{2t+2} 
			\bigl(\harm{t+1} - \harm{k+1}\bigr)
	\\	&\wwrel=
			\frac12\bigl(\harm{t+1} - \harm{k+1}\bigr) \;.
	\end{align*}
	Inserting yields the claim.
	\end{proof}
}

\ifproceedings{}{
	Finally, using the Chernoff bound for the multinomial distribution 
	(\wref{lem:chernoff-bound-multinomial}),
	we obtain the following concentration property of the entropy of a normalized multinomial variable.
	
	\begin{lemma}[Concentration of Entropy of Multinomials]
	\label{lem:entropy-of-multinomial-bound}
		Let \(u\in\N\) and \(\vect p\in(0,1)^u\) with \(\total{\vect p} = 1\), 
		and \(\vect X \eqdist \multinomial(n, \vect p)\).
		Then it holds that
		\begin{align*}
		\numberthis\label{eq:entropy-of-multinomial-concentration}
				\E[\bigg]{ \entropy[\ln]\biggl(\frac{\vect X}n\biggr) }
			&\wwrel=
				\entropy[\ln](\vect p) \wbin\pm \rho,
		\end{align*}
		where we have for any \(\delta\in(0,1)\) with \(\delta \ge \sqrt{20u/n}\), 
		\(h\in (0,1)\) and \(C_h\) as in \weqref{eq:C-h} that
		\begin{align*}		
				\rho
			&\wwrel\le 
				C_h \delta^h \bigl(1 - 3 e^ {-\delta^2 n/25}\bigr)
				\bin+
				3 u \ln(u) e^{-\delta^2 n/25}.
		\end{align*}
		If \(u=\Oh(n^{\nu})\) as \(n\to\infty\) for a constant \(\nu\in[0,1)\), then
		\weqref{eq:entropy-of-multinomial-concentration} holds with an error of
		\(\rho = o(n^{-(1-\nu)/2+\varepsilon})\) for any fixed \(\varepsilon > 0\).
	\end{lemma}

	\begin{proof}
	We start with the multinomial Chernoff bound:
	for any $\delta \ge \sqrt{20u/n}$ we have that
	\begin{align*}
			\Prob[\big]{ \|\vect X - \vect x\|_\infty \rel\ge \delta n }
		&\wwrel\le
			\Prob[\Bigg]{ \,\sum_{i=1}^u\, \biggl|\frac{X_i}n - q_i\biggr| \rel\ge \delta }
	\\	&\wwrel{\relwithref[r]{lem:chernoff-bound-multinomial}\le}
			3 \exp (-\delta^2 n/25) .
	\numberthis\label{eq:chernoff-X-by-n-minus-p-greater-delta}
	\end{align*}
	To use this in estimating
	$\E[\big]{ \bigl| \entropy[\ln]\bigl(\frac{\vect X}n\bigr) - \entropy[\ln](\vect p) \bigr| }$, 
	we divide the domain of $\frac{\vect X}n$ into 
	the region of values with $\|\cdot\|_\infty$-distance 
	at most $\delta$ from $\vect p$, and all others.
	By \wref{lem:entropy-function}, $\entropy[\ln]$ is Hölder-continuous for any exponent $h\in(0,1)$
	with Hölder-constant $C_h$.
	Using this and the boundedness of $\entropy[\ln]$ (\wref{lem:entropy-function-bounds}) yields
	\begin{align*}
			\E[\bigg]{ \biggl| 
				\entropy[\ln]\biggl(\frac{\vect X}n\biggr) - \entropy[\ln](\vect p) 
			\biggr| } 
		&\wwrel{\relwithref{eq:chernoff-X-by-n-minus-p-greater-delta}\le}
			\sup_{\vect\xi \rel: \|\vect\xi\|_\infty < \delta} \mkern-10mu\bigl| 
				\entropy[\ln](\vect p + \vect\xi) - \entropy[\ln](\vect p) 
			\bigr| 
			\cdot \bigl(1 - 3 e^{- \delta^2 n / 25 }\bigr)
	\\*	&\wwrel\ppe\quad{}	
			\bin+
			\sup_{\vect x} \,\bigl| \entropy[\ln](\vect x) - \entropy[\ln](\vect p)\bigr|
			{}\cdot 3 e ^{- \delta^2 n / 25 }
	\\	&\wwrel{\relwithref[r]{lem:entropy-function}\le}
			C_h \delta^h \cdot \bigl(1 - 3 e^{-\delta^2 n / 25 }\bigr)
			\bin+
			3 \ln(u) e^{-\delta^2 n/25}.
	\end{align*}
	This proves the first part of the claim.
	
	For the second part, we assume $u=\Oh(n^\nu)$,
	thus $u\le d n^{\nu}$ for a constant $d$ and large enough $n$.
	We obtain an asymptotically valid choice of $\delta$ when $\delta = \omega(n^{(\nu-1)/2})$;
	then for large enough $n$ we will have $\delta > \sqrt{20} d  n^{(\nu-1)/2} \ge \sqrt{20 u/n}$.
	
	Let now an $\varepsilon > 0$ be given and
	set $\tilde\epsilon = \epsilon+\nu/2$.
	We may further assume that $\tilde\varepsilon<\frac12$ since the claim is vacuous for larger~$\tilde\epsilon$.
	We choose a Hölder exponent $h\in(0,1)$ so that
	$h > \frac{1-2\tilde\epsilon}{1-\nu}$ 
	(this is possible since $\frac{1-2\tilde\epsilon}{1-\nu}<1$ for $\tilde\epsilon>\nu/2$);
	this is equivalent to the relation
	\begin{align*}
			\frac{\nu-1}2
		&\wwrel<
			- \frac{\frac12-\tilde\epsilon}h.
	\end{align*}
	We can thus pick $c$ between these two values, \eg, 
	$c = \bigl(\frac{\nu-1}2 - \frac{1/2-\tilde\varepsilon}{h} \bigr)\big/2$.
	Since $c > -(\nu-1)/2$, the choice $\delta = n^c$ guarantees 
	$\delta \ge \sqrt{20u/n}$ for large enough $n$
	and we can apply \weqref{eq:entropy-of-multinomial-concentration}.
	
	As we now show, these choices are sufficient to prove the claim $\rho = o(n^{-1/2+\tilde\epsilon})$.
	To streamline the computations, we note that (by its definition) we can write $h$ as
	\begin{align*}
			h
		&\wwrel=
			\frac{1-2\tilde\varepsilon}{1-\nu} + \frac{4}{1-\nu}\lambda
			\mkern-150mu
			&&\text{for some constant $\lambda>0$ and}
	\\
			h
		&\wwrel=
			\frac{1-2\tilde\epsilon}{1-\nu\bin-2\lambda''}
		\wwrel=
			\frac{1-2\tilde\epsilon}{1+\nu\bin-2\lambda'}
			&&\text{for constants $\lambda''>0$ resp.\ $\lambda'>\nu$,}
	\end{align*}
	which implies $h\cdot c + (\frac12-\tilde\epsilon) = -\lambda < 0$ and
	$2c+1 = \lambda' > 0$. 
	With these preparations we find 
	(for $n$ large enough to have $u\le d n^{\nu}$ and $\delta \ge \sqrt{20u/n}$) 
	that
	\begin{align*}
			\rho \cdot n^{1/2-\tilde\epsilon}
		&\wwrel\le
			C_h\delta^h n^{1/2-\tilde\epsilon} \bigl(1-3\exp(-\delta^2 n/25)\bigr)
			\bin+
			3 n^{1/2-\tilde\epsilon} \ln(u) \exp(-\delta^2 n/25)
	\\	&\wwrel\le
			\underbrace{ C_h n^{-\lambda} }_{{}\to 0} 
				\cdot \bigl(1 - \underbrace{3 \exp (-n^{\lambda'}\mkern-4mu/25)}_{{} \to 0} \bigr)
			\bin+
			\underbrace{ 3\nu \ln(d\.n) \exp \Bigl(-3 n^{\lambda'} + (\tfrac12-\tilde\varepsilon) \ln(n)\Bigr) }_{{} \to 0}
	\\[-1ex]	&\wwrel\to
			0
			,\qquad(n\to\infty),
	\end{align*}
	which implies the claim.
	\end{proof}
}

%% file: equal-keys-input-models.tex
\section{Input Models}
\label{sec:input-models}
This section formally defines
the input models and some related notation.

\subsection{Multiset Model}
In the \emph{random multiset permutation model} (\aka \textsl{exact-profile model}),
we have parameters $u\in\N$, the \emph{universe size,} and 
$\vect x\in\N^u$, the fixed \emph{profile}.
An input under this model always has size $n = \total{\vect x}$,
and is given by a uniformly chosen random permutation of
\[
	\underbrace{1,\ldots,1}_{x_1\text{ copies}},\,
	\underbrace{2,\ldots,2}_{x_2\text{ copies}},\,\ldots,\,
	\underbrace{u,\ldots,u}_{x_u\text{ copies}},
\]
\ie, 
the multiset with $x_v$ copies of the number $v$ for $v=1,\ldots,u$.
\ifproceedings{}{%
	The random multiset permutation model is a natural generalization of the
	classical random permutation model 
	(which considers permutations of an ordinary set).
}%
I~write $\ui{C_{\vect x}}k$ for the (random) number of (ternary) comparisons
used by \wref{alg:list-quicksort}
to sort a random multiset permutation with profile~$\vect x$;
I will in the following use $C_{\vect x}$ (omitting $k$) for conciseness.

\subsection{\IID Model}
In the \emph{(discrete) \iid model} 
(\aka \textsl{probability model}~\citep{Kemp1996} or \textsl{expected-profile model}~\citep{Wild2016}) 
with parameters $u\in\N$ and $\vect q\in(0,1)^u$
with $\total{\vect q} = 1$,
an input of size $n$ consists of $n$ \iid (independent and identically distributed) 
random variables $U_1,\ldots,U_n$
with $U_i \eqdist \distFromWeights{\vect q}$ for $i=1,\ldots,n$.
The domain $[u]$ is called the \emph{universe,} and $\vect q$ the 
(probability vector of the) \emph{universe distribution}.

I~denote by $X_v$, for $v\in[u]$, the number of elements $U_i$ that have value $v$;
the vector $\vect X = (X_1,\ldots,X_u)$ of all these \emph{multiplicities} is called 
the \emph{profile} of the input $\vect U = (U_1,\ldots,U_n)$.
Clearly, $\vect X$ has a multinomial distribution, $\vect X \eqdist \multinomial(n,\vect q)$, 
with mean $\E{\vect X} = n \vect q$.

\ifproceedings{}{
	The discrete \iid model is a natural complement of the random-permutation model:
	we draw elements from a \emph{discrete} distribution in the former whereas 
	the latter
	is equivalent to drawing \iid elements from any \emph{continuous} distribution.
}
By $C_{n,\vect q} = \ui{C_{n,\vect q}}k$, I denote the number of (ternary) comparisons
used by \wref{alg:list-quicksort}
to sort $n$ \iid $\distFromWeights{\vect q}$ elements 
(again usually omitting the dependence on $k$).

\subsection{Relation of the Two Models}
\label{sec:relation-of-two-models}

The two input models 
\parenthesisclause{random multiset permutations with profile $\vect x$
resp.\ $n=\total{\vect x}$ \iid elements with distribution $\vect q = \vect x/n$}
are closely related.
Both can be described using an urn with (initially) $n$ balls, where
for every $v\in[u]$ exactly $x_v$ balls bear the label $v$.
The multiset model corresponds to drawing $n$ balls from this urn without replacement
(ultimately emptying the urn completely), so that the profile of the input is always $\vect x$.
The \iid model corresponds to drawing $n$ balls from the the urn \emph{with} replacement;
this leaves the urn unchanged and the $n$ draws are mutually independent.

In our reference Quicksort (\wref{alg:list-quicksort}),
we use the first $k$ elements, \ie, the first $k$ balls drawn from the urn, 
to determine the pivot:
it is the median of this sample. 
By that, we try to estimate the median of all $n$ balls (initially) in the urn,
so as to obtain a balanced split.
In the multiset model, the $k$ elements are drawn \emph{without} replacement, 
whereas in the \iid model, they are drawn \emph{with} replacement 
from the urn.
The algorithm is of course the same in both cases 
(and indeed chooses sample elements \emph{without} replacement), 
but the mutual independence of elements in the \iid model implies that
the sampling process is (in this latter case) equivalent to sampling \emph{with} replacement.

Now it is well-known that sampling without replacement yields  
a \emph{better} (\ie, less variable) estimate of the true median than
sampling with replacement, 
and thus leads to more favorable splits in Quicksort!
The above reasoning applies for recursive calls, as well,
so that on average the multiset model is no more costly for Quicksort than the \iid model:
$\E{C_{\vect x}} \le \E{C_{n,\vect x/n}}$.
We will thus focus on analyzing the \iid model.%
\footnote{%
	The full version of this paper will discuss the multiset model
	and the above reasoning in more detail.
}

%% file: equal-keys-previous-work.tex
\section{Previous Work}
\label{sec:previous-work}

\ifproceedings{
	We briefly discuss previous results relevant for this paper;
	the extended version contains a more comprehensive summary of related work.
	
}{}
In the \textsl{multiset sorting problem}, 
we are given a permutation of a \emph{fixed} multiset,
where value $v\in[u]$ appears $x_v\in\N$ times,
for a total of $n = \total{\vect x}$ elements.
Neither $u$ nor $\vect x$ are known to the algorithm.
We consider only comparison-based sorting;
unless stated otherwise, all results concern the number of 
\emph{ternary} comparisons, \ie, one comparison has as result $<$, $=$ or~$>$.

\subsection{Multiset Sorting}
Multiset sorting attracted considerable attention in the literature.
\citet{MunroRaman1991} prove a lower bound of 
$n\ld n - \sum_{i=1}^u x_i \ld x_i - n\ld e \pm \Oh(\log n)$ (ternary) comparisons,
which can equivalently be written in terms of the entropy as
$\entropy(\vect x / n) n - n\ld e \pm \Oh(\log n)$.
We reproduce their main argument in \wref{sec:lower-bound} and extend it 
to \iid inputs.

\ifproceedings{
	Algorithms to optimally sort multisets in the worst case up to an $\Oh(n)$ error term 
	are known~\citep{MunroSpira1976,MunroRaman1991}.
}{
	The conceptually simplest algorithm coming close to this bound 
	is to insert elements into a splay tree, collecting all 
	duplicates in linear lists inside the nodes.
	By ``static optimality'' of splay trees (Theorem~2 of \citet{SleatorTarjan1985}), this needs
	$\Oh(\entropy(\vect x/n) n)$ comparisons and so is optimal up to a constant factor.
	That factor is at least~$2$ (using \textsl{semi-splaying}), and we need linear extra space.
	
	Already in 1976, \citet{MunroSpira1976} described simple variants of Mergesort and Heapsort 
	that \emph{collapse} duplicate elements whenever discovered.
	They are optimal up to an $\Oh(n)$ error term \wrt comparisons, 
	but do not work in place. 
	(Their Heapsort requires a non-standard extract-min variant that does not work in place.)
	
	The first in-place method was the adapted Heapsort of \citet{MunroRaman1991};
	it does not use the element-collapsing technique, but rather removes all duplicates from
	the heap in one bulk extract-min operation.
}
None of these methods made it into practical library implementations
since they incur significant overhead \wrt existing sorting methods
when there are not many equal keys in the input.

\subsection{Quicksort on Multiset Permutations}

Building on \citeauthor{Burge1976}'s analysis of BSTs~\citep{Burge1976},
\citeauthor{Sedgewick1977a} analyzed 
several Quicksort variants on random permutations of multisets 
in his 1977 article~\citep{Sedgewick1977a}.
For fat-pivot Quicksort without sampling, he found the exact average-case result 
(when every permutation of the multiset is equally likely):
$2\QSentropy(\vect x)+n-u$ ternary comparisons, 
where $\QSentropy(\vect x) = \sum_{1\le i<j\le u} x_i x_j / (x_i +\cdots+x_j)$
is the so-called ``Quicksort entropy''.
\ifproceedings{}{
	Interestingly, Sedgewick found fat-pivot partitioning not advisable for practical use at that time;
	this only changed with the success of the implementation of \citet{Bentley1993}.
}

Two decades later,
\citet{SedgewickBentley1999,SedgewickBentley2002} combined this exact, 
but somewhat unwieldy
result with the bound $\QSentropy(\vect q) \le \entropy(\vect q) \ln 2$
and concluded that with at most $\bigl(2\ln(2)\. \entropy(\vect x/n)+1\bigr)n$
ternary comparisons on average, 
fat-pivot Quicksort is asymptotically optimal in the average case 
for sorting a random permutation of any fixed multiset%
\ifproceedings{}{%
\footnote{
	\Citet{SedgewickBentley1999,SedgewickBentley2002} compare this number against 
	$\ld\bigl(\binom{n}{x_1,\ldots,x_u}\bigr) = \ld \bigl( n!/(x_1!\cdots x_u!) \bigr) $, 
	\ie, the logarithm of the number of different input
	orderings (given by the multinomial coefficient).
	This information-theoretic argument lower bounds the number of needed yes/no questions 
	(\ie, binary comparisons), 
	but more elaboration is necessary for \emph{ternary} comparisons.
	The lower bound of \citet{MunroRaman1991} (cf.\ \wref{sec:lower-bound})
	uses a reduction to distinct elements and yields the desired bound for ternary comparisons.
}
}%
\parenthesissign up to the constant factor~$2\ln 2$.
The bound $\QSentropy(\vect q) \le \entropy(\vect q) \ln 2$ was noted
in a seemingly unrelated context by \citet[Theorem~3.2]{AllenMunro1978} that appeared just one year after 
Sedgewick's Quicksort paper~\citep{Sedgewick1977a}.
\Citeauthor{AllenMunro1978} studied the move-to-root heuristic for
self-organizing BSTs, which they found to have the \emph{same} search costs in the long run 
as a BST built by inserting elements drawn \iid according to the access distribution 
until saturation.
We will consider this connection between Quicksort and search trees in detail
in \wref{sec:equals-quicksort-and-search-trees}.

\Citeauthor{KatajainenPasanen1992} considered Quicksort-based approaches for multiset sorting.
\ifproceedings{
	Their analysis only shows that we use at most
	$
			\entropy(\vect x/n) n + (\tfrac2{\ld e}-1) n\ln n \bin\pm \Oh(n)
	$
	comparisons, which is \emph{not} entropy-optimal.
}{
	They argued (indirectly) that a fat-pivot Quicksort
	uses on average $2n\ln(n) - \sum_{v=1}^u x_v \ld(x_v) \wbin\pm \Oh(n)$ comparisons (their Theorem~3),
	since ``Due to the three-way partitions, all redundant comparisons between a pivot and
	elements equal to the pivot are avoided''~\citep{KatajainenPasanen1992}.
	Note however that this only shows that we use at most
	$
			\entropy(\vect x/n) n + (\tfrac2{\ld e}-1) n\ln n \bin\pm \Oh(n)
	$
	comparisons, which is \emph{not} entropy-optimal.
	
	In a companion paper~\citep{KatajainenPasanen1994} they described a stable Quicksort version 
	with exact median selection
	and showed that it needs $\Oh(\entropy(\vect x/n)n)$ comparisons even 
	in the worst case;
	however the constant of proportionality is one plus the constant for deterministic median selection, 
	and thus at least $3$~\citep{DorZwick2001}.
}

\subsection{Fringe-Balanced Trees}

\ifproceedings{}{
	The concept of fringe balancing (see \wref{sec:fringe-balanced-trees}) 
	appears under a handful of other names in the literature:
	\textsl{locally balanced search trees}~\citep{Walker1976},
	\textsl{diminished trees}~\citep{Greene1983}, and
	\textsl{iR / SR trees}~\citep{HuangWong1983,HuangWong1984}.
	I~use the term \emph{fringe-balanced trees}
	since it is the most vivid term and since it is by now widely adopted in the analysis-of-algorithms
	community, see, \eg, the relatively recent monograph~\citep{Drmota2009} by \citeauthor{Drmota2009}.
	The name ``fringe balanced'' itself has its origins in a technique called \textsl{fringe analysis}, 
	which \citet{PobleteMunro1985} applied to BSTs that use
	what they called a ``fringe heuristic''.
	The earliest occurrence of ``fringe balanced'' seems to be in the title
	of a paper by \citet{Devroye1993}; 
	curiously enough, \citeauthor{Devroye1993} did not use this term in the 
	main text of the paper.
	
	Along with the different names come slight variations in the definitions;
	I~remark that our definition (deliberatively) differs a bit in the base cases from usual definitions
	to precisely mimic Quicksort recursion trees.
}

Many parameters like path length, height and profiles of fringe-balanced 
trees have been studied
when the trees are built from a random permutation of $n$ distinct elements,
see, \eg, \citet{Drmota2009}.
The case of equal elements has not been considered except for 
the unbalanced case $k=1$, \ie, ordinary BSTs; see \citet{Kemp1996,Archibald2006}.

\section{New Results}
\label{sec:main-result}

We now state the main claim of this paper.
\ifproceedings{%
	We
}{%
	Note that the error terms of asymptotic approximations in the results mentioned above 
	only involved $n$, so there was no need to specify an explicit relation between 
	the profile of the multiset $\vect x = (x_1,\ldots,x_u)$, the universe size $u$, and the number
	of elements $n$;
	here we are not so fortunate. 
	We hence 
}%
include $n$ as sub- or superscript whenever the dependence on the input size is important.

\begin{theorem}[Main Result]
\label{thm:main-result}
	Let $(u_n)_{n\in\N}$ be a sequence of integers and
	$(\ui{\vect q}n)_{n\in\N}$ be a sequence of vectors with
	$\ui{\vect q}n \in (0,1)^{u_n}$ and $\total{\ui{\vect q}n} = 1$
	for all $n\in\N$.
	Assume further that $(\ui{\vect q}n)$ has ``many duplicates'', \ie, 
	there is a constant $\epsilon>0$ so that 
	$\min_{v\in[u_n]} \ui{q_v}nn = \Omega(n^\epsilon)$ as $n\to\infty$.
	Abbreviate the corresponding (binary) entropy by $\mathcal H_n = \entropy (\ui{\vect q}n)$.
	
	The number $C_{n,\ui{\vect q}n}$ of (ternary) comparisons used by median-of-$k$ Quicksort
	with fat-pivot partitioning to sort $n$ elements drawn \iid according to 
	$\distFromWeights{\ui{\vect q}n}$
	fulfills
	\begin{align*}
			\E{C_{n,\ui{\vect q}n}}
		&\wwrel=
			\alpha_k \. \mathcal H_n \. n
			\wwbin\pm
			\Oh\Bigl( \bigl( \mathcal H_n ^ {1-\delta} +1 \bigr) n \Bigr),
			\ifproceedings{}{\qquad(n\to\infty),}
	\end{align*}
	\ifproceedings{as $n\to\infty$}{}%
	for any constant $\delta < \frac2{k+5}$.
	This number is asymptotically optimal up to the factor $\alpha_k$.
\end{theorem}

\paragraph{Previous Approaches And Why They Fail for \boldmath$k>1$}

\label{sec:why-previous-approaches-fail}

Sedgewick's analysis~\citep{Sedgewick1977a}
is based on explicitly solving the recurrence for the expected number of comparisons.
Since it has a \emph{vector,} namely the profile $\vect x$, as parameter, 
tricky differencing operations are required to obtain a telescoping recurrence.
They rely on symmetries that are only present for the most basic version
of Quicksort:
it has now been 40 years since Sedgewick's article appeared, and
not the slightest generalization of the analysis to, say, median-of-3 Quicksort, 
has been found.

Following our approach outlined in the introduction, we can alternatively compute
the expected costs to search each element of the multiset 
in a BST built by inserting the same elements in random order.
The random number of comparisons can thus be written as 
the scalar product $\vect \Gamma^T \vect x$, where $\vect\Gamma$
is the node-depth vector of the BST (cf.\ \wref{sec:equals-quicksort-and-search-trees}).
For an ordinary BST, once an element is present in the tree,
any further insertions of the \emph{same} value are without effect;
so we obtain the same tree no matter how many duplicates of this element
later follow.
This means that the resulting tree has \emph{exactly the same shape} as
when we insert elements drawn \iid according to $\distFromWeights{\vect q}$
with $\vect q= \vect x / n$ until saturation.
\ifproceedings{}{%
	The expected search costs in the latter case are found to be precisely $2\QSentropy(\vect q)+1$
	by a comparatively simple argument%
	\footnote{%
		Assume we search $v\in[u]$.
		We always have one final comparison with outcome $=$;
		the remaining comparisons are on the search path and compare 
		$v$ to some $j\ne v$.
		We compare $v$ to $j$ iff among the values between $v$ and $j$, 
		$j$ was the \emph{first} to be inserted into the tree,
		which happens with probability ${q_j}/({q_v+\cdots+q_j})$.
		Adding up the indicator variables for these events and multiplying by the 
		probability $q_v$ to search that value $v$, we obtain
		\(
			1 + \sum_{v=1}^u q_v \sum_{j\ne v} \frac{q_j}{q_v+\cdots+q_j}
			\wwrel=
			1 + 2 \sum_{1\le i < j \le u} \frac{q_i q_j}{q_i+\cdots+q_j}
			\wwrel=
			2 \QSentropy(\vect x / n) + 1
		\).
	}%
	~\citep{AllenMunro1978};
	multiplying by $n$ gives the Quicksort costs.
}%

For median-of-$k$ Quicksort we obtain $k$-fringe-balanced trees,
and now a certain number of duplicate insertions \emph{do} affect the shape of the tree%
\ifproceedings{%
	.
}{%
	; after all, this is the way the balancing is achieved in first place
	(see \wref{sec:fringe-balanced-trees}).
}%
As the multiset model corresponds to drawing elements without replacement,
the probabilities for the values \emph{change} after each insertion.
Analyzing the search cost then essentially reduces to solving the vector-recurrence
for Quicksort with pivot sampling that has resisted all attempts for 40 years.
One might conclude that the case $k=1$ can be explicitly solved precisely because 
we were effectively working in the \iid model instead of the multiset model.

\paragraph{My Assumption}
The only hope I see to make progress for $k > 1$ is thus to cling to the \iid model
even though it is not equivalent to the multiset model anymore.
We thereby retain independent insertions and the analysis of search costs 
in fringe-balanced trees becomes conceivable.
However, we face a new problem:
How often we search each value $v$ is now a random variable $X_v$
and the random number of comparisons is $\vect\Gamma^T \vect X$,
the \emph{product} of two random variables.
Since $\vect\Gamma$ is the node-depth vector of a tree built by inserting a multiset
with profile $\vect X$ (in random order), 
the two random variables are \emph{not independent}.

My assumption 
\parenthesisclause{$\Omega(n^{\epsilon})$ expected occurrences of each value}
is a simple sufficient condition to circumvent this complication 
(see \wref{sec:equals-separating-n-and-q}).
\ifproceedings{}{
	It is possible to slightly weaken it (at the price of a more clumsy criterion):
	We can tolerate values with minuscule multiplicities in $\Oh(n^{\delta})$
	as long as $\delta < \epsilon$ (so that we have a nonempty \emph{separating range} $(\delta,\epsilon)$),
	and the total number of these rare values is $\Oh(n^{1-\epsilon})$.
}

%% file: equal-keys-search-trees.tex
\section{Quicksort and Search Trees with Duplicates}
\label{sec:equals-quicksort-and-search-trees}

The correspondence discussed in 
\wref[Sections]{sec:quicksort-search-trees}/\ref{sec:fringe-balanced-trees}
extends to inputs with equal keys
if we consider a \emph{weighted} path length in the tree,
where the weight of each node is the multiplicity $X_v$ of 
its key value $v$.
This is best seen in a concrete example.

\paragraph{Example}
\label{sec:equals-quicksort-costs-value-wise}

Let the universe size be $u=5$
with profile $\vect X = (X_1,\ldots,X_5)$.
Assume that $\vect X \ge k$ so that
each of the five values is used as a pivot in \emph{exactly} one partitioning step,
and the leaves in the final tree will be empty.
Assume we obtain the recursion tree shown in \wref{fig:recursion-tree-example}.

\begin{figure}[tbhp]
	\parbox{.33\linewidth}{%
	\plaincenter{%
		\externalizedpicture{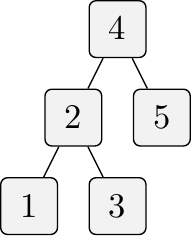}
	}%
	}%
	\parbox{.66\linewidth}{%
	\caption{%
		Exemplary recursion tree for $u=5$.
		Each node represents a partitioning step, 
		with the given pivot value.
		\protect\ifproceedings{}{
			Child links correspond to child recursive calls.
		}%
		Empty leaves are not shown.
		The node-depths vector for this tree is $\vect\Gamma=(3,2,3,1,2)$.
	}
	\label{fig:recursion-tree-example}
	}%
\end{figure}

The traditional recurrence for Quicksort sums up the 
costs of all partitioning steps.
Each such uses one comparison per element, plus the 
comparisons for selecting the pivot, say, at most $c\cdot k$;
ignoring the latter, we essentially sum up the subproblem sizes of all recursive calls.
\ifproceedings{\par}{}
For our recursion tree this yields
\begin{align*}
	\left\{ 
	\begin{array}{*{12}{@{\hspace{5pt}}l}}
		X_1 & + & X_2 & + & X_3 & + & X_4 & + & X_5 &  & \pm & c\cdot k \\
		X_1 & + & X_2 & + & X_3 &   &     &   &     &  & \pm & c\cdot k \\
		X_1 &   &     &   &     &   &     &   &     &  & \pm & c\cdot k \\
		    &   &     &   & X_3 &   &     &   &     &  & \pm & c\cdot k \\
		    &   &     &   &     &   &     &   & X_5 &  & \pm & c\cdot k
	\end{array}
	\right. 
\end{align*}
comparisons, (each row corresponding to one partitioning step).
For example, the step with pivot $2$ gets as input all elements smaller
than $4$, \ie, $X_1+X_2+X_3$ many, so
the number of comparisons used in this step is $X_1+X_2+X_3 \pm c\cdot k$
(for some constant $c$ depending on the median-selection algorithm).

The key observation is that we can also read the result \emph{column-wise,}
aligning the rows by key values:
up to the given error terms, we find that
\textsl{sorting costs are the cost of searching each input element in the (final) recursion tree!}
\ifproceedings{}{
	For example, searching $3$ in the tree from \wref{fig:recursion-tree-example}, 
	we first go left, then right and then find $3$ as the pivot, so
	the costs are $3$ \parenthesissign which is precisely the coefficient of $X_3$ in the overall costs.
}%
In vector form, we can write the search costs as $\vect\Gamma^T \vect X$, where 
$\vect\Gamma$ is the \emph{node-depths vector} of the recursion tree,
\ie, the vector of depths of nodes sorted by their keys.
In the example
$
		\vect\Gamma
	\wwrel=
		(3,2,3,1,2)
$.
\ifproceedings{}{
	This is nothing else than a weighted path length; 
	(more precisely, 
	a weighted internal path length where we include \emph{both} endpoints 
	to determine the length of a path).
}

\subsection{Recursion Trees}
\label{sec:recursion-trees}

We formalize the result of the above example in two lemmas.
The first one formally captures the correspondence of Quicksort and search trees.

\begin{lemma}[Recursion Trees]
\label{lem:equals-equivalence-recursion-tree-search-tree-exact}
	For any input \(\vect U = (U_1,\ldots,U_n)\) 
	(with or without duplicates)
	the following processes execute the \textit{same set} of (ternary) key comparisons
	and produce the same tree shape:
	\begin{enumerate}[label=(\arabic*)]
	\item 
		sorting \(\vect U\) with \wref{alg:list-quicksort} and storing the recursion tree,
		ignoring any calls to Insertionsort,
	\item 
		inserting \(\vect U\) successively into an initially
		empty $k$-fringe-balanced tree using \wref{alg:fb-insert}.
	\end{enumerate}
\end{lemma}

\begin{proof}
We prove the equivalence by induction on $n$.
If $n\le k-1$, Quicksort stops 
(it passes control directly to Insertionsort which we ignore),
so the recursion tree consists of one leaf only.
Likewise in the search tree, 
all elements are gathered in the single leaf and no comparisons happen.
So assume the claim holds for inputs with less than $k$ elements.

If now $n\ge k$, Quicksort chooses a pivot $P$ from the first $k$ elements,
compares all elements to $P$,
and divides the input into segments $\ui{\vect U}1$ and $\ui{\vect U}2$ containing
the elements strictly smaller resp. strictly larger than $P$.
All duplicates of $P$ are put in place and vanish from the recursion tree.

Now consider what happens upon inserting the $k$th element in the search tree.
This is the first time a leaf is split and
the key for the inner node (the root of the tree) 
is chosen from the first $k$ inserted elements overall.
We thus choose same value $P$ that was chosen as pivot in Quicksort.
The other elements from the leaf are compared to $P$ and 
inserted into one of the two new leaves (unless they are duplicates of $P$).
Any later insertions must start at the root, 
so each of these elements are also compared to $P$ before
the insertion continues in one of the subtrees, or stops if a duplicate of $P$ is found.

So we execute the same set of comparisons in both processes at the root of the tree.
Towards applying the inductive hypothesis for recursive calls resp.\ subtrees,
we note that the relative order of elements is retained in both processes,
so the elements inserted in the left/right child of the root are exactly 
$\ui{\vect U}1$ resp.\ $\ui{\vect U}2$ in both cases.
The claim thus follows by induction.
\end{proof}

The proof relies on the fact that \wref{alg:list-quicksort} retains
the relative order of elements,
but practical non-stable, in-place implementations
(\eg, those in \citep{Bentley1993,Sedgewick2011algorithms}) 
do not fulfill this requirement.
However, a weaker version of 
\wref{lem:equals-equivalence-recursion-tree-search-tree-exact}
remains valid for any fat-pivot partitioning method:
the two processes always have the same \emph{distribution} 
of the number of comparisons and tree shapes over \emph{randomly ordered} inputs.
\ifproceedings{}{%
	In this general case, sorting the input corresponds
	to inserting a (different, but uniquely determined) permutation of the input
	into a fringe-balanced tree;
	my \phdthesis~\citep{Wild2016} gives some more details on that.
}
Such a distributional version of \wref{lem:equals-equivalence-recursion-tree-search-tree-exact}
is sufficient for all results in this paper, 
so our results apply to practical implementations, as well.

\subsection{Search Costs}

The second lemma relates the costs to build a search tree to
the cost of searching in the final tree.

\begin{lemma}[Search Costs]
\label{lem:insert-costs-approx-search-costs}
	Let \(\vect U = (U_1,\ldots,U_n)\) be an input (with or without duplicates),
	let $\mathcal T$ be built from $\vect U$ by successive insertions
	using \wref{alg:fb-insert} and let $B_{\vect U}$ be the number of comparisons
	done in this process.
	Assume there are $I$ inner nodes in $\mathcal T$.
	Searching each element of \(\vect U\) in $\mathcal T$ 
	\ifproceedings{}{using \wref{alg:fb-search}}
	(ignoring the sequential search in leaves)
	uses $B_{\vect U} \pm c\cdot Ik$ comparisons.
\end{lemma}

\begin{proof}
	The elements in the leaves of $\mathcal T$ require the exact same comparisons 
	for insert and search to find the path to their leaf.
	(We ignore the sequential search within the leaf).
	So consider the key $v\in[u]$ of an inner node.
	The first occurrences of $v$ in $\vect U$ are simply added to the leaf
	that later becomes the inner node with label $v$.
	Each of these entails one comparison more when searching for it in $\mathcal T$
	as we paid when inserting it
	(namely the last comparison with $v$ that identifies them as equal).
	However, there are be at most $k$ such elements before the leaf overflows,
	and for the remaining duplicate insertions of $v$, we \emph{do} the last comparison
	(same as in the search). So searching pays up to $I\cdot k$ comparisons more.
	On the other hand,
	whenever a leaf overflows, we need a certain number of comparisons to select the
	median, say at most $c\cdot k$ for some constant $c$.
	So insertion pays up to $I\cdot ck$ comparisons more than insertion.
\end{proof}

In general, $I$ can be as large as $n/2$, but it is certainly bounded 
by the number $u$ of distinct keys in the input.
Together this yields the formal generalization of our example.

\begin{corollary}[Quicksort and Search Costs]
\label{cor:quicksort-costs-search-costs}
\ifproceedings{\mbox{}\\}{}
	Let \(\vect U = (U_1,\ldots,U_n)\) be an input over universe $[u]$.
	Let $\mathcal T$ be built from $\vect U$ by successive insertions
	using \wref{alg:fb-insert} and denote by $\vect\Gamma$ its node-depth vector.
	The cost of \wref{alg:list-quicksort} on $\vect U$ (ignoring Insertionsort)
	is $\vect\Gamma^T \vect X \pm c uk$ for $\vect X$ the profile of $\vect U$.
\qed\end{corollary}

\ifproceedings{}{
	Up to an error term of $\Oh(u)$, we can thus consider search costs in fringe-balanced
	trees instead of sorting costs in Quicksort.
}%
For a special class of such trees,
we will be able to determine the search costs:
the \emph{saturated fringe-balanced trees}.
They are the subject of the next section.

\section{Saturated Fringe-Balanced Trees}
\label{sec:saturated-trees}

Consider a $k$-fringe-balanced tree $\mathcal T$ built by successively inserting elements drawn 
\iid $\distFromWeights{\vect q}$ (for a fixed universe distribution $\vect q$)
into an initially empty tree.
How does this tree evolve if we continue inserting indefinitely?
Since the universe $[u]$ is finite and duplicate insertions do not alter $\mathcal T$,
the process reaches a stationary state almost surely.
We call the trees corresponding to such stationary states \emph{saturated trees} 
(\wrt the given, fixed universe).
The expected search costs in saturated trees
will play a key role in the analysis of Quicksort.

\input{equal-keys-saturated-trees-recurrence}

%% file: equal-keys-saturated-trees-recurrence.tex
We start by developing a stochastic description of the shape
of a random saturated tree $\mathcal T$ and set up a recurrence equation
for the expected search costs.

\subsection{Stochastic Description}
\label{sec:stochastic-model}

Let $\vect q\in (0,1)^u$ with $\total{\vect q} = 1$ be given.
The distribution function of the universe distribution $U\eqdist \distFromWeights{\vect q}$ is 
$F_U(v) = \Prob{U\le v} = \smash{\sum_{i=1}^{\lfloor v\rfloor} q_i} $ for $v\in[0,u+1)$,
and I~denote its (generalized) inverse by
$F^{-1}_U: (0,1) \to [1..u]$ with 
$F^{-1}_U(x) = \inf\{ v\in[1..u]  : F_U(v) \ge x \}$.

Let $P$ be the label of the root of $\mathcal T$;
the distribution of $P$ is a key ingredient to (recursively) describe saturated trees.
\ifproceedings{}{
	($P$ is also the pivot chosen in the first partitioning step of median-of-$k$ Quicksort.)
}%
When the first leaf overflows, $P$ is chosen as the median of the first $k=2t+1$ inserted values,
which are \iid $\distFromWeights{\vect q}$ distributed,
so it is given by
\begin{align*}
\numberthis\label{eq:dist-P}
		P
	\wwrel\eqdist
		f_U^{-1}(\Pi),
\end{align*}
where $\Pi$ has a $\betadist(t+1,t+1)$ distribution,
\ie, it has density 
$f_\Pi(z) = z^t(1-z)^t / \BetaFun(t+1,t+1)$.
$\BetaFun(a,b) = \Gamma(a)\Gamma(b) / \Gamma(a+b)$
is the \textsl{beta function}.
This is the generalized \textsl{inversion method of random sampling,} 
see \citet[Sec.\,V.3.4]{Devroye1986}, 
illustrated in our \wref{fig:equals-stochastic-model-sampling},
which is based on the fact that $\betadist(t+1,t+1)$ is the distribution of the 
median of $2t+1$ \iid uniformly in $(0,1)$ distributed random variables 
(see, \eg, \citep[Sec.\,I.4.3]{Devroye1986}).
For convenient notation, we write $\vect D = (D_1,D_2) = (\Pi,1-\Pi)$ for the induced 
\textsl{spacings}; see \wref{fig:equals-stochastic-models-probs}.

\ifproceedings{\begin{figure}[t]}{\begin{figure}[t]}
	\begin{captionbeside}{%
		Illustration of pivot sampling in the \iid model%
		\protect\ifproceedings{.}{with $u=6$.}
		$\Pi$~is the $x$-coordinate of a point uniformly chosen in the gray area (the area under the curve),
		and $P$ is the index of the interval this point lies in.%
		\protect\ifproceedings{\vspace*{-2ex}}{\vspace*{3ex}}%
	}
	\plaincenter{%
	\externalizedpicture{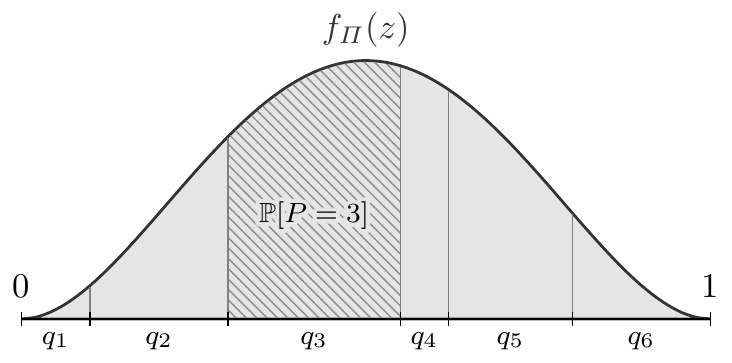}
%
%
%
%
%
%
%
%
%
%
	}
	\ifproceedings{\label{fig:equals-stochastic-model-sampling}}{}
	\end{captionbeside}
	\ifproceedings{}{\label{fig:equals-stochastic-model-sampling}}
\end{figure}

\ifproceedings{\begin{figure}[t]}{\begin{figure}[t]}
	\ifproceedings{}{\vspace*{-2ex}}%
	\begin{captionbeside}{%
		Relation of the different quantities in the stochastic description.%
		\protect\ifproceedings{\vspace*{-2ex}}{\vspace{3ex}}%
	}
	\plaincenter{%
	\externalizedpicture{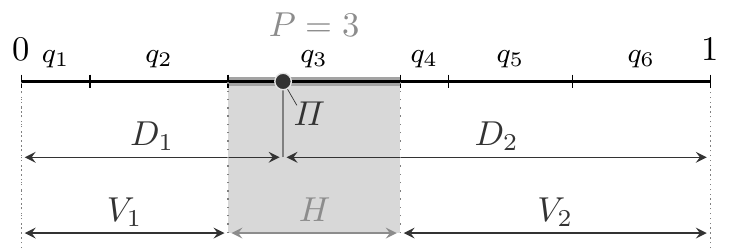}
%
%
%
%
%
%
%
%
	}
	\ifproceedings{\label{fig:equals-stochastic-models-probs}}{}
	\end{captionbeside}
	\ifproceedings{}{\label{fig:equals-stochastic-models-probs}}
\end{figure}

We further denote by $V_1$ and $V_2$ the probability that a random element 
$U\eqdist\distFromWeights{\vect q}$ belongs to the left resp.\ right subtree of the root,
and by $H = \Prob{U=P}$ the probability to ``hit'' the root's value.
These quantities are fully determined by $P$ (see also \wref{fig:equals-stochastic-models-probs}):
\begin{galign}[\label{eq:def-V1-V2-H}]
		V_1
	&\wwrel=
		q_1 +\cdots+ q_{P-1},
\\
		V_2
	&\wwrel=
		q_{P+1} +\cdots+ q_{u},
\\
		H
	&\wwrel=
		q_P.
\end{galign}
(In the boundary case $P=1$, we have $V_1 = 0$, and similarly $V_2=0$ for $P=u$.)
Finally, we denote by $\vect{Z_1}$ and $\vect{Z_2}$ the ``zoomed-in'' 
universe distributions in the left resp.\ right subtree:
\begin{galign}[\label{eq:def-Z}]
		\vect{Z_1}
	&\wwrel=
		\left(\frac{q_1}{V_1} ,\ldots, \frac{q_{P-1}}{V_1}\right),
\\
		\vect{Z_2}
	&\wwrel=
		\left(\frac{q_{P+1}}{V_2} ,\ldots, \frac{q_{u}}{V_2}\right).
\end{galign}
$\vect{Z_1}$ is not well-defined for $P=1$; we set it to the \emph{empty} vector $\vect{Z_1} = ()$ in this case.
Similarly $\vect{Z_2}=()$ for $P=u$.

\subsection{Search Costs}
\label{sec:equals-stwl-path-length-distributional}

Let $\mathcal T$ be a random $k$-fringe-balanced tree resulting from
inserting \iid $\distFromWeights{\vect q}$ elements until saturation.
Each value $v\in[u]$ then appears as the key of one inner node of $\mathcal T$;
let $\Gamma_v$ denote its depth, \ie, the (random) number of nodes on the path (including endpoints)
from the root to the node containing $v$ in the (random) tree $\mathcal T$.
The vector
$
		\vect\Gamma
	\wwrel=
		(\Gamma_1,\ldots,\Gamma_u)
$
is called the (random) \emph{node-depths vector} of $\mathcal T$
(see \wpref{fig:recursion-tree-example} for an example).
Finally, we write $A_{\vect q} = \vect\Gamma^T \vect q$.
This is the average depth of a node drawn according to $\distFromWeights{\vect q}$ in
the (random) tree~$\mathcal T$;
note that we \emph{average} the costs over the searched \emph{key}, but consider the \emph{tree} fixed;
so $A_{\vect q}$ is a random variable since $\mathcal T$ remains random: 
the (weighted) average node depth in a random saturated $k$-fringe-balanced tree.

The expected node depth, or equivalently, the expected search cost in the tree,
can be described recursively:
The root contributes one comparison to any searched element, 
and with probability $H$ the search stops there.
Otherwise, the sought element is in the left or right subtree with probabilities $V_1$ resp.\ $V_2$, 
and the expected search costs in the subtrees are given recursively by $A_{\vect{Z_1}}$ and $A_{\vect{Z_2}}$.
With the notation from above, 
this yields a \textsl{distributional recurrence} for~$A_{\vect q}$:
\begin{galign}[\label{eq:equals-A-q-recurrence}]
\label{eq:equals-A-q-recurrence-u-ge-1}
		A_{\vect q}
	&\wwrel\eqdist
		1 \bin+ V_1 \ui{A_{\vect{Z_1}}}1 + V_2 \ui{A_{\vect{Z_2}}}2 \,,
	\qquad(u\ge1),
\\
\label{eq:equals-A-q-recurrence-initial-condition}
		A_{()}
	&\wwrel=
		0,
\end{galign}
where
$(\ui{A_{\vect q}}1)$ and $(\ui{A_{\vect q}}2)$ are independent copies of $(A_{\vect q})$, 
which are also independent of $(V_1,V_2,\allowbreak \vect{Z_1},\vect{Z_2})$.

%% file: equal-keys-quicksort-many-duplicates.tex
\section{Quicksort With Many Duplicates}
\label{sec:equals-separating-n-and-q}

\ifproceedings{
	We consider 
}{
	As discussed in \wref{sec:why-previous-approaches-fail},
	the attempts to analyze median-of-$k$ Quicksort 
	in full generality have failed;
	I therefore confine myself to
}%
the following restricted \iid model.

\begin{definition}[Many Duplicates]
\label{def:equals-expected-profile-model-many-duplicates}
\ifproceedings{\mbox{}\\}{}
	Let \((\ui{\vect q}n)_{n\in\N}\) be a sequence of stochastic vectors, where
	\(\ui{\vect q}n\) has \(u_n\) entries, \ie, \(\ui{\vect q}n \in (0,1)^{u_n}\) and
	\(\total{\ui{\vect q}{n}} = 1\), for all \(n\in\N\).
	An input of size \(n\in\N\) under the\/ \textit{\iid model} for \((\ui{\vect q}n)_{n\in\N}\)
	consists of the \(n\) \iid \(\distFromWeights{\ui{\vect q}n}\) 
	distributed random variables.
\ifproceedings{\par}{}	
	$(\ui{\vect q}n)_{n\in\N}$ is said to have \textit{many duplicates} if
	there is a constant \(\varepsilon>0\) so that
	\(\mu_n= \Omega(n^{-1+\varepsilon})\) as $n\to\infty$ where
	\(\mu_n \ce \min_r \ui{q_r}{n}\) is the smallest probability.
\end{definition}

This condition ensures that every value occurs $\Omega(n^{\epsilon})$ times
in expectation.
\ifproceedings{}{
	(It might hence be more appropriate to say many duplicates \emph{of each kind,} 
	but I~refrain from doing so for conciseness.)
}%
With many duplicates, we expect few \emph{degenerate} inputs
in the following sense.

\begin{definition}[Profile-Degenerate Inputs]
\label{def:equals-degenerate-profile}
\ifproceedings{\mbox{}\\}{}
	Let \(\nu\in[0,1)\) and \(k\in\N\).
	An input vector \(\vect U = (U_1,\allowbreak\ldots,\allowbreak U_n) \in [u]^n\) of size \(n\) is 
	called \textit{\((\nu,k)\)-profile-degenerate} if 
	not all \(u\) elements of the universe appear at least \(k\) times 
	in the first \(n_T = \lceil n^\nu\rceil\) elements \(U_1,\ldots,U_{n_T}\) of \(\vect U\).
	If the parameters are clear from the context or are not important, we call \(\vect U\) 
	simply \textit{profile-degenerate}.
\end{definition}

For non-degenerate inputs, the recursion tree will depend only on the first $n_T$ elements,
and the profile of the remaining $n-n_T$ elements is \emph{independent of this tree.}
By choosing $n_T$ large enough to have non-degenerate inputs \whp, but 
small enough to keep the contribution of the first $n_T$ elements to the search costs negligible,
we obtain the following theorem.

\begin{theorem}[Separation Theorem]
\label{thm:separating-n-and-q}
\ifproceedings{~\\}{}
	Consider median-of-$k$ Quicksort with fat-pivot partitioning 
	under a discrete \iid model with many duplicates.
	The expected number of comparisons fulfills
	\begin{align*}
	\numberthis\label{eq:equals-expected-profile-stateless-costs-expectation}
			\E{C_{n,\ui{\vect q}n}}
		&\wwrel=
			\E{A_{\ui{\vect q}n}} \cdot n \bin\pm \Oh(n^{1-\varepsilon}),
			\ifproceedings{}{\qquad(n\to\infty),}
	\end{align*}
	for a constant \(\varepsilon>0\)
	\ifproceedings{as \(n\to\infty\)}{};
	more precisely,
	we need \(\varepsilon\in(0,\tilde\epsilon)\) when
	\(\mu_n = \min_r \ui{q_r}n = \Omega(n^{-1+\tilde\epsilon})\).
\end{theorem}

\ifproceedings{}{
	Recall that \(\E{A_{\ui{\vect q}n}}\) is the expected search cost in
	a \emph{saturated} $k$-fringe-balanced tree built from $\distFromWeights{\ui{\vect q}n}$.
	It depends only on the universe distribution \(\ui{\vect q}n\) (and \(k\)), 
	but \emph{not} (directly) on \(n\).
	We have therefore separated the influence of $n$ and $\vect q$ (for inputs with many duplicates),
	and can investigate $\E{A_{\vect q}}$ in isolation in the next section.
	The remainder of this section is devoted to the proof of the separation theorem.
}

%% file: equal-keys-separating-n-and-q.tex
\ifproceedings{\begin{proof}}{\begin{proof}[\wref{thm:separating-n-and-q}]}
By \wref{cor:quicksort-costs-search-costs} we can study search costs 
in fringe-balanced trees;
the main challenge is that this tree is built form the same
input that is used for searching.
In other words, $\vect\Gamma$ and $\vect X$ are not independent;
for non-degenerate inputs, they are however almost so.

We start noting the following basic fact that we will use many times:
	In any discrete \iid model, \(u_n \le \frac1{\mu_n}\) for all \(n\).
	In particular, an \iid model with many duplicates has
	\(u_n = \Oh(n^{1-\varepsilon})\).
\ifproceedings{}{%
	This is easy to see:
	Since $\mu_n$ is the smallest entry of $\ui{\vect q}n$ we have 
	$1=\total{\ui{\vect q}n} \ge u_n \mu_n$, so $u_n \le 1/\mu_n$.
	The second part follows directly from \wref{def:equals-expected-profile-model-many-duplicates}.
}

\paragraph{Probability of degenerate profiles}

\ifproceedings{%
	Since no element of the universe is very unlikely,
	we can bound the probability of degenerate inputs.
}{%
	We can bound the probability of degenerate inputs
	using Chernoff bounds.
	An elementary far-end lower-tail bound (\wref{lem:binomial-not-very-small-whp}) 
	actually yields the following slightly stronger asymptotic result.
}

\begin{lemma}[Non-Degenerate \whp]
\label{lem:equals-bound-for-degenerate-profile-probability}
	Assume an \iid model with 
	\(\mu_n = \min_v \ui{q_v}{n} = \Omega(n^{-\rho})\) for \(\rho\in[0,1)\)
	and let \(k\in\N\) and \(\rho<\nu<1\).
	Then the probability of an input of size \(n\) to be
	\((\nu,k)\)-profile-degenerate is in \(o(n^{-c})\) as $n\to\infty$ 
	for any constant \(c\).
\end{lemma}

\ifproceedings{
	The lemma follows from a standard application of the union bound and
	tail inequalities; in this case \wref{lem:binomial-not-very-small-whp}
	yields a stronger results than classical Chernoff bounds.
	The detailed proof is in the extended version.
}{
	\begin{proof}
	Let $\rho\in[0,1)$, $k$ and $\nu\in (\rho,1)$ be given.
	Set $\varepsilon = \nu - \rho > 0$ and
	denote by $\vect Y = \ui{\vect Y}n$ the 
	profile of the first $n_T=\lceil n^{\nu}\rceil$ elements of the input. 
	Clearly \(\ui{\vect Y}n \eqdist \multinomial(n_T;\ui{\vect q}n)\).
	Assume \withoutlossofgenerality that the minimal probability is always $\ui{q_1}n = \mu_n$.
	A standard application of the \textsl{union bound} yields
	\begin{align*}
			\Prob[\big]{\neg \ui{\vect Y}n \ge k}
		&\wwrel=
			\Prob[\Bigg]{\bigvee_{v=1}^{u_n} \ui{Y_v}n < k}
	\\	&\wwrel\le
			\sum_{v=1}^{u_n}\Prob[\big]{\ui{Y_v}n < k}
	\\	&\wwrel\le
			u_n \cdot \Prob{\ui{Y_1}n < k} .
	\numberthis\label{eq:prob-nondegenerate-le-u-times-more-than-k-ones}
	\end{align*}
	Now $\ui{Y_1}n \eqdist \binomial(n_T,\mu_n)$
	with $\mu_n = \Omega(n^{-\rho}) = \Omega(n_T^{-\rho/\nu})=\omega\bigl(\frac{\log n_T}{n_T}\bigr)$,
	and we always have $\mu_n \le \frac12 < 1$,
	so we can apply \wref{lem:binomial-not-very-small-whp}: 
	for any given constant~$c$, we have
	$\Prob{\ui{Y_1}n < k} = o\bigl(n_T^{-(c+1)/\nu}\bigr)$.
	Since $u_n = o(\frac n{\log n}) = o(n)$ we find that
	\begin{align*}
			n^c \cdot \Prob[\big]{\neg \ui{\vect Y}n \ge k}
		&\wwrel{\relwithref{eq:prob-nondegenerate-le-u-times-more-than-k-ones}\le}
			n^c \, u_n \, \Prob{\ui{Y_1}n < k}
	\\	&\wwrel=
			o\bigl(n^{c+1}\bigr) \cdot o\Bigl(n_T^{-\frac{c+1}\nu}\Bigr)
	\\	&\wwrel=
			o(1),
	\end{align*}
	since $n_T\sim n^\nu$.
	So the input is $(\nu,k)$-degenerate with high probability.
	\end{proof}
}

In \wref{thm:separating-n-and-q}, we assume an \iid model with many duplicates, \ie,
\(\mu_n = \Omega(n^{-1+\tilde\epsilon})\) with $\tilde\epsilon\in(0,1]$,
and an $\epsilon \in (0,\tilde\epsilon)$ is given.
We set
\begin{align*}
\ifproceedings{\let\wwrel\rel}{}
		\nu 
	\wwrel\ce 
		\frac{(1-\tilde\epsilon) + (1-\epsilon)}2
	\wwrel=
		1 - \frac{\tilde\epsilon + \epsilon}2
	\wwrel\in 
		(1-\tilde\epsilon,1-\epsilon).
\end{align*}
Then, by \wref{lem:equals-bound-for-degenerate-profile-probability}, an input 
of size $n$ is $(\nu,k)$-degenerate with probability in $o(n^{-c})$ for all~$c$.
This also means that the overall cost contribution of degenerate inputs to expected costs is in 
$o(n^{-c})$ for all $c$,
and hence covered by the error term in \weqref{eq:equals-expected-profile-stateless-costs-expectation},
since costs for any input are at most quadratic in $n$.

We will thus, for the remainder of this proof, assume that the input is not 
$(\nu,k)$-degenerate,
\ie, each of the values of the universe appears at least $k$ times among the first
$n_T = \lceil n^\nu \rceil$ elements.

\paragraph{Independence of Profiles and Trees}
\ifproceedings{}{
	We now turn to the distribution of the recursion trees.
}%
The shape of the recursion tree is determined by at most $u\cdot k$ elements:
we have at most $u$ partitioning rounds
since each of the $u$ elements of the universe becomes a pivot in at most one partitioning step,
and each partitioning step inspects $k$ elements for choosing its pivot.

Also, for each of the $u$ values in the universe, at most the first $k$ occurrences
in the input, reading from left to right, can influence the tree:
if a value $v\in[u]$ is already contained in an inner node, 
all further duplicates of $v$ are ignored.
Otherwise, all occurrences of $v$ must appear in a single leaf, which can hold up to $k-1$
values, so there are never more than $k-1$ copies of $v$ in the tree.
The leaf will overflow at the latest upon inserting the $k$th occurrence of $v$, 
and then a new internal node with pivot $v$ is created.

In a non-degenerate input, the first $k$ duplicates appear among the first $n_T$ 
elements $U_1,\ldots,U_{n_T}$ of the input, so all pivots are chosen based on these elements only.
Moreover, after these $n_T$ insertions, all $u$ values appear as labels of inner nodes.
\ifproceedings{}{
	All leaves are empty then and remain so for good:
	the recursion tree has reached a saturated state.
}

We denote by $\vect\Gamma = \vect\Gamma(\ui{\vect q}n)$ the node-depths vector 
for the final recursion tree and 
by $\vect{\tilde X}$ the profile of $U_{n_T+1},\ldots,U_n$.
Since they are derived from disjoint ranges of the \iid input,
$\vect\Gamma$ and $\vect{\tilde X}$ are stochastically independent.

\paragraph{Overall Result}
\ifproceedings{}{%
	We now have all ingredients to compute the overall costs of Quicksort.
}%
Recall that $u_n = \Oh(n^{1-\tilde\epsilon})$ and $n_T\sim n^\nu$ 
with $1-\tilde\epsilon < \nu < 1-\epsilon$.
\ifproceedings{
	By \wref{pro:tree-log-height-weak-whp}, our trees have 
	height $>13\ln n$ with probability in $\Oh(n^{-2})$, \ie,
	for most inputs, $\Gamma_v = \Oh(\log n)$ for all $v$.
}{
	Since a recursion tree cannot have a path longer than $u$,
	we always have $\vect\Gamma \le c u$ for a fixed constant $c$ that 
	depends only on the cost measure, \ie, $\Gamma_v = \Oh(n^{1-\tilde\epsilon})$ for all $v\in[u]$.
	However, this estimate is very pessimistic; we
	know that randomly grown trees have logarithmic height \whp
	(\wref{sec:preliminaries-log-height-whp} resp.\ \wref{app:height-of-recursion-trees}),
	so $\Gamma_v = \Oh(\log n)$ for all $v$ with high probability.
	To be concrete, the height is $>13\ln n$ with probability in $\Oh(n^{-2})$
	by \wref{pro:tree-log-height-weak-whp}.
}

We therefore further split the set of non-profile-degenerate inputs into 
\emph{``height-degenerate''} ones
where the height 
\ifproceedings{}{
	of the resulting recursion tree
}%
is $> 13\ln n$ and all other ones.
This gives the following stochastic representation conditional on the input
being not $(\nu,k)$-profile-degenerate.
\begin{align*}
		C_{n,\ui{\vect q}n}
	&\wwrel=
		\vect\Gamma^T \vect{X} \wbin\pm \Oh(u)
\\	&\wwrel=
		\vect\Gamma^T \vect{\tilde X} \wwbin\pm
		n_T \|\vect\Gamma\|_{\infty}
		\wbin\pm\Oh(u)
\\	&\wwrel\eqdist
		\vect\Gamma^T \vect{\hat X} \wwbin\pm 
		2 n_T \|\vect\Gamma\|_{\infty}
		\wbin\pm\Oh(u)
\\	&\wwrel=
		\vect\Gamma^T \vect{\hat X} \wwbin\pm 
		\Oh\Bigl(
			n_T\bigl(
				\indicator{\text{not height-deg.}}\cdot \log n
\ifproceedings{\\*&\wwrel\ppe\quad{}}{}
				+ 
				\indicator{\text{height-deg.}} \cdot u
			\bigr)
		\Bigr)
\\	&\wwrel=
		\vect\Gamma^T \vect{\hat X} \wwbin\pm 
		\Oh\Bigl(
			n^\nu \bigl(
				\indicator{\text{not height-deg.}}\cdot \log n
\ifproceedings{\\*&\wwrel\ppe\quad{}}{}
				+ 
				\indicator{\text{height-deg.}} \cdot n^{1-\tilde\epsilon}
			\bigr)
		\Bigr)
,
\end{align*}
where $\vect{\hat X}$ is independent of $\vect \Gamma$ and 
$\vect{\hat X} \eqdist \multinomial(n,\ui{\vect q}n)$.
\ifproceedings{\par}{}
We thus find that
\begin{align*}
\begin{multlined}
		C_{n,\ui{\vect q}n}
	\wwrel\eqdist
		\vect\Gamma^T \vect{\hat X} \wwbin\pm 
		\Oh(n^{1-\epsilon}),
\ifproceedings{\\}
		\qquad(\text{input neither profile- nor height-degenerate}).
\end{multlined}
\numberthis\label{eq:equals-stochastic-representation-non-degenerate}
\end{align*}
Taking expectations over all non-degenerate inputs in 
\weqref{eq:equals-stochastic-representation-non-degenerate},
exploiting independence,
and inserting $\Prob{\text{height-deg.}} = \Oh(1/n^2)$
(\wref{pro:tree-log-height-weak-whp})
yields
\begin{align*}
		\E{C_{n,\ui{\vect q}n}}
	&\wwrel=
		\E[\big]{\vect\Gamma^T \vect{\hat X}} \wwbin\pm 
		\Oh\bigl( n^{1-\epsilon} \bigr)
\\	&\wwrel=
		\E{\vect \Gamma}^T \cdot \E{\vect{\hat X}} \wwbin\pm 
		\Oh\bigl( n^{1-\epsilon} \bigr)
\\	&\wwrel=
		\underbrace{\bigl(\E{\vect \Gamma}^T \cdot \ui{\vect q}n \bigr) }_{\E{A_{\ui{\vect q}n}}} 
		\cdot\, n \wwbin\pm \Oh(n^{1-\varepsilon}),
\numberthis\label{eq:equals-Cnq-linear-in-n}
\end{align*}
with $A_{\ui{\vect q}n}$ as given in \wpeqref{eq:equals-A-q-recurrence}.
As argued above, 
the contribution of profile-degenerate inputs is in $o(n^{-c})$ 
for any $c$ and thus covered by $\Oh(n^{1-\varepsilon})$,
so \weqref{eq:equals-Cnq-linear-in-n} holds also for the unconditional expectation.
This concludes the proof of \wref{thm:separating-n-and-q}.
\end{proof}

%% file: equal-keys-entropy-bounds.tex
\section{Expected Search Costs in Saturated Trees}
\label{sec:expected-node-depth-asymptotic}

The remaining step of the analysis consists in computing $\E{A_{\vect q}}$,
the expected search costs in saturated $k$-fringe-balanced trees
built from $\iid$ $\distFromWeights{\vect q}$ elements.
\ifproceedings{}{%
	Previous work only covers the unbalanced BST case ($k=1$), where
	the result is known \emph{exactly}:
	$\E{A_{\vect q}} = 2 \QSentropy(\vect q) + 1 \le 2\entropy[\ln](\vect q) + 1$ 
	where $\QSentropy(\vect q) = \sum_{1\le i<j\le u} q_i q_j / (q_i +\cdots+q_j)$%
	~\citep[Theorems~3.1 and~3.4]{AllenMunro1978}.
}

Since the recurrence equation for the expected search costs 
(\wpeqref{eq:equals-A-q-recurrence}) seems hard to solve,
we try to relate it to a quantity that we already know: 
the \emph{entropy of the universe distribution.}
The entropy function satisfies an \textsl{aggregation property:}
intuitively speaking, we do not change the entropy of the final outcomes
if we \emph{delay} some decisions in a random experiment
by first deciding among whole groups of outcomes and then continuing within the chosen group.
\ifproceedings{}{%
	Indeed, this property was Shannon's third fundamental requirement when he introduced 
	his entropy function in 1948~\citep[p.\,393]{Shannon1948}.
}

The aggregation property can nicely be formalized in terms of trees, 
see Lemma~6.2.2E of \citet[p.\,444]{Knuth1998},
and we can use it in a search tree to express the entropy of 
node access probabilities as the entropy of the split corresponding to the root
plus the entropies of its subtrees, weighted by their respective total probabilities.
More specifically in our setting, we have 
$\vect q\in[0,1]^u$ with $\total{\vect q} = 1$ and 
we condition on the value $P\in[u]$ in the root.
Using the notation from \wref{sec:stochastic-model}
\ifproceedings{}{
	\parenthesisclause{%
	the zoomed-in distributions of the subtrees $\vect Z$, 
	the probability to go into these subtrees $\vect V$, 
	and the probability to access the root $H$%
	}
}%
we find that
\begin{align*}
\numberthis\label{eq:entropy-aggregation-P}
		\entropy(\vect q)
	&\wwrel=
		\entropy(V_1, H, V_2)
		\bin+ \sum_{r=1}^2 V_r \entropy(\vect{Z_r}) \,.
\end{align*}
(By linearity, $\entropy$ can be \wrt any base; we will use it with $\entropy[\ln]$ below.)

If we likewise condition on the root value $P$ in the recursive description of the search costs,
we find for the expected search costs
that
\begin{align*}
\numberthis\label{eq:equals-E-A-q-recurrence}
		\E{A_{\vect q}}
	&\wwrel=
		1 \bin+ \sum_{r=1}^2 V_r \E{A_{\vect{Z_r}}} \,.
\end{align*}
$\entropy(\vect q)$ and $\E{A_{\vect q}}$ fulfill the \emph{same form} of recurrence, 
only with a different toll function: $1$ instead of $\entropy(V_1, H, V_2)$!
We thus try to relate these two toll functions by obtaining bounds on $\entropy(V_1, H, V_2)$,
and then extend this relation inductively to $\entropy(\vect q)$ and $\E{A_{\vect q}}$.
The technical difficulties in doing so are that 
	$\entropy(V_1, H, V_2)$ is very sensitive to $\vect q$, 
	so we have to do a case distinction to obtain bounds on it.
	We hence cannot give matching upper and lower bounds for $\entropy(V_1, H, V_2)$,
	which necessitates the introduction of second order terms to account for the slack
	(cf.\ the constant $d$ below).
Separately deriving upper and lower bounds on $\E{A_{\vect q}}$ from that
in terms of $\entropy(\vect q)$, and making them match asymptotically 
in the leading term, we obtain the following result.

\begin{theorem}[Expected Search Costs]
\label{thm:A-q-entropy-bounds}
\ifproceedings{\mbox{}\\}{}
	Let a sequence of universe distributions $(\ui{\vect q}n)_{n\in\N}$
	be given for which $\mathcal H_n \ce \entropy[\ld](\ui{\vect q}n) \to \infty$ as $n\to\infty$.
	The expected search costs of a saturated $k$-fringe-balanced tree (with $k=2t+1$)
	built from \iid $\distFromWeights{ \ui{\vect q}n }$ keys is given by
	\begin{align*}
	\begin{multlined}
			\E{A_{\ui{\vect q}n}}
		\wwrel=
			\alpha_k \entropy[\ld](\ui{\vect q}n)
			\bin\pm \Oh\Bigl(
				\mathcal H_n{}\bigr.^{\frac{t+2}{t+3}}
				\log\bigl( \mathcal H_n \bigr)
			\Bigr),
	\ifproceedings{\\*[-.5ex]}{}
			\qquad (n\to\infty).
	\end{multlined}
	\end{align*}
\end{theorem}
\ifproceedings{\vspace{-2ex}}{}

%% file: equal-keys-entropy-bounds-proof.tex
\label{sec:proof-entropy-bounds}

\begin{proof}%
We start with the upper bound.
We actually derive a whole class of upper bounds characterized by a parameter $\epsilon$.

\begin{lemma}[Upper Bound]
\label{lem:equals-entropy-upper-bound-eps}
	Let $\E{A_{\vect q}}$ satisfy \weqref{eq:equals-E-A-q-recurrence},
	and let $\epsilon\in(0,1)$ be given. Define
	\begin{align*}
			c \wrel= c_\epsilon
		&\wwrel=
			\frac1{\tilde H
			\wbin- 4\epsilon \tilde h}
			\,,
	\\
			d \wrel= d_\epsilon
		&\wwrel=
			\frac {(t+1) \, \BetaFun(t+1,t+1) }
				{ \epsilon^{t+2}(1-\epsilon)^{t} } \,,
	\\ 	\text{where}\quad
			\tilde H
		&\wwrel=
			\harm{k+1} - \harm{t+1} 
	\\ 	\text{and}\quad
			\tilde h
		&\wwrel=
			\harm k - \harm{t} \;.
	\end{align*}
	If $c\ge 0$, we have that 
	$
			\E{A_{\vect q}}
		\wwrel\le
			c \cdot \entropy[\ln](\vect q) \bin+ d
	$
	for all stochastic vectors $\vect q$.
\end{lemma}

\begin{proof}
Let $\epsilon$ with $c=c_\epsilon\ge 0$ be given. 
Note that $d=d_\epsilon\ge 0$ holds for all $\epsilon$.
The proof is by induction on $u$, the size of the universe.
If $u=0$, \ie, $\vect q = ()$, we have $\E{A_{\vect q}} = 0$,
see \weqref{eq:equals-A-q-recurrence-initial-condition}.
Since $d\ge 0$ and here $\entropy[\ln](\vect q) = 0$, the claim holds.

Now assume that $u\ge1$ and the claim holds for all (strictly) smaller universe sizes.
We start by taking expectations in \weqref{eq:equals-A-q-recurrence} and
conditioning on the pivot value $P$:
\ifproceedings{\begin{align*}\SwapAboveDisplaySkip}{\begin{align*}}
		\E{A_{\vect q}}
	&\wwrel=
		1 \bin+
		\Eover* {P} {\sum_{r=1}^2 \E{V_r A_{\vect{Z_r}}\given P}}
\\	&\wwrel=
		1 \bin+
		\Eover* {P} {\sum_{r=1}^2 V_r \,\E{ A_{\vect{Z_r}}\given P}}\
\shortintertext{using the inductive hypothesis}
	&\wwrel\le
		1 \bin+
		\Eover* {P} {\sum_{r=1}^2 V_r \,(c \entropy[\ln](\vect{Z_r})+d)}
\\	&\wwrel=
		1 \bin+
		c\cdot \E* {\sum_{r=1}^2 V_r \,\entropy[\ln](\vect{Z_r})}
		\bin+ d \cdot \E{\total{\vect V}}
\\	&\wwrel{\eqwithref{eq:entropy-aggregation-P}}
		c\cdot \entropy[\ln](\vect q) 
\ifproceedings{\\* &\wwrel\ppe{}}{}
		\bin+
		\underbrace{
		1 \bin-
		c\cdot \E[\big] {\entropy[\ln](\vect V, H)}
		\bin+ d \cdot \E{\total{\vect V}}
		} _ {\varpi}
	\numberthis\label{eq:equals-def-varpi}
\end{align*}
It remains to show that $\varpi \le d$.
We consider two cases depending on the maximal probability 
$\lambda = \max_{1\le v\le u} q_v$.
\begin{enumerate}%
\item Case $\lambda < \epsilon$:\\
	In this case, all individual probabilities are smaller than $\epsilon$,
	so it is plausible that we can bound the 
	expected entropy of partitioning
	$\E{\entropy[\ln](\vect V, H)}$
	from below.
	The subuniverse probabilities $V_r$ are quite close to the continuous
	spacings $D_r$:
	by definition (see also \wtpref{fig:equals-stochastic-models-probs})
	we have in interval-arithmetic notation
	\begin{align*}
	\numberthis\label{eq:equals-Dr-Vr-all-qr-le-eps}
			V_r
		&\wwrel=
			D_r \bin+ (-2\epsilon,0)
			,\qquad (r=1,2).
	\end{align*}
	For the expected partitioning entropy, this means
	\begin{align*}
			\E{\entropy[\ln](\vect V\mkern-5mu,\mkern -1mu H)}
		&\wwrel\ge
			\sum_{r=1}^2 \E[\big]{V_r \ln(1/V_r)}
	\intertext{%
		using \weqref{eq:equals-Dr-Vr-all-qr-le-eps} and
		$x\log(1/x) \ge (x-\epsilon)\log(1/(x+\epsilon'))$ 
		for $\epsilon,\epsilon'\ge0$ and $x\in[0,1]$, this is
	}
	\ifproceedings{\\*[-2\baselineskip]}{}
		&\wwrel\ge
			\sum_{r=1}^2 \E[\big]{(D_r-2\epsilon) \ln(1/D_r)}
	\\[-1ex]	&\wwrel=
			\E[\big]{\entropy[\ln](\vect D)} 
			\wbin+ 2\epsilon \sum_{r=1}^2 \E[\big]{\ln(D_r)}
	\\[-1ex]	&\wwrel{\relwithtext{%
					\llap{\wref{lem:expected-entropy}, }%
					\wref{eq:logarithmic-beta-integral}}=%
				}
			\tilde H
			\ifproceedings{\bin+}{\wbin+}
			2\epsilon \sum_{r=1}^2 \bigl(\psi(t+1)-\psi(k+1)\bigr)
	\\ &\wwrel=
			\tilde H
			\ifproceedings{\bin-}{\wbin-} 
			4\epsilon  (\harm k - \harm{t}) 
	\\ &\wwrel=
			\tilde H
			\ifproceedings{\bin-}{\wbin-}  
			4\epsilon \tilde h 
	\\ &\wwrel=
			1/c
			\;.
	\end{align*}
	Hence $c$ satisfies 
	$c \ge 1 \big/ \E{\entropy[\ln](\vect V, H)} \ge 0$,
	which implies
	\begin{align*}
			\varpi 
		&\wwrel\le 
			1 \bin-
			\frac{1}{\E{\entropy[\ln](\vect V, H)}} \cdot 
				\E[\big] {\entropy[\ln](\vect V, H)}
	\ifproceedings{\\* &\wwrel\ppe{}}{}
			\bin+ d \cdot \underbrace{\E{\total{\vect V}}}_{\le 1}
	\\[-2ex]	&\wwrel\le
			d.
	\end{align*}
	The inductive step is proven in this case.
	
\item Case $\lambda \ge \epsilon$: \\
	In the second case, 
	there is a \emph{likely} value $v\in[u]$ with $q_v \ge \epsilon$.
	We will show a lower bound for having this value as label of the root.

	We have 
	$\Pi \eqdist \betadist(t+1,t+1)$ and
	hence 
	$f_{\Pi}(z) = \frac{z^{t}(1-z)^{t}}{\BetaFun(t+1,t+1)}$.
	Recall that $P = f_U^{-1}(\Pi)$ (\wref{fig:equals-stochastic-model-sampling}),
	so we can bound the probability to draw $v$ from below by the smallest value of the
	integral over any $\epsilon$-wide strip of the density:
	\ifproceedings{\vspace{-2ex}}{}
	\begin{align*}
			\Prob{P = v}
		&\wwrel\ge
			\min_{0\le\zeta\le 1-\varepsilon} \int_{z=\zeta}^{\zeta+\varepsilon} f_{\Pi_r}(z) \, dz
	\\	&\wwrel=
			\min_{0\le\zeta\le 1-\varepsilon} \int_{z=\zeta}^{\zeta+\varepsilon} 		
					\frac{z^{t}(1-z)^{t}}{\BetaFun(t+1,t+1)} \, dz
	\\	&\wwrel=
			\int_{z=0}^{\varepsilon} \frac{z^{t}(1-z)^{t}}{\BetaFun(t+1,t+1)} \, dz
	\\	&\wwrel\ge
			\int_0^\epsilon 
				\frac{z^{t}(1-\epsilon)^{t}}{\BetaFun(t+1,t+1)} \, dz
	\\	&\wwrel=
			\frac{\epsilon^{t+1}(1-\epsilon)^{t}}
				{(t+1)\,\BetaFun(t+1,t+1)}
			\;.
	\numberthis\label{eq:equals-prob-Pr-eq-v-when-qv-ge-eps}
	\end{align*}
	For the expected hitting probability, we thus
	have for any $\vect q$ with a $q_v \ge \epsilon$ that
	\begin{align*}
			\E{H}
		&\wwrel\ge
			q_v \cdot \Prob{P = v}
		\ifproceedings{\\&}{}
		\wwrel{\relwithref{eq:equals-prob-Pr-eq-v-when-qv-ge-eps}\ge}
			\frac { \epsilon^{t+2}(1-\epsilon)^{t} } 
				{(t+1) \, \BetaFun(t+1,t+1) }
		\ifproceedings{\\&}{}
		\wwrel= 1/d,
	\numberthis\label{eq:equals-E-H-ge}
	\end{align*}
	so
	$d \ge 1 \big/ \E{H}$.
	This implies
	\begin{align*}
			\varpi - d 
		&\wwrel\le 
			1 -
			c\cdot \E[\big] {\entropy[\ln](\vect V, H)}
			+ d \cdot \E{\total{\vect V}} - d
	\\	&\wwrel\le
			1 - (1-\E{\total{\vect V}}) \cdot \frac{1}{\E{H}}
	\\[-.5ex]	&\wwrel=
			0.
	\end{align*}
	\ifproceedings{This concludes the inductive step in case 2.}{This concludes the inductive step for the second case.}
\end{enumerate}
The inductive step is thus done in both cases, and the claim holds for
all stochastic vectors $\vect q$ by induction.
\end{proof}

\noindent
The lower bound on $\E{A_{\vect q}}$ uses basically the same techniques;
only a few details differ.

\begin{lemma}[Lower Bound]
\label{lem:equals-entropy-lower-bound-eps}
	Let $\E{A_{\vect q}}$ satisfy \weqref{eq:equals-E-A-q-recurrence},
	and let $\epsilon\in(0,1/e)$ be given. Define
	\begin{align*}
			c \wrel= c_\epsilon
		&\wwrel=
			\frac1{\tilde H
			\wbin+ 4\epsilon \bin+\epsilon\ln(1/\epsilon)}
			\,,
	\\
			d \wrel= d_\epsilon
		&\wwrel=
			\bigl(c_\epsilon \ln(3) \bin- 1 \bigr)
				\frac {(t+1) \, \BetaFun(t+1,t+1) }
					{ \epsilon^{t+2}(1-\epsilon)^{t} } \,,
	\\ 	\text{where}\quad
			\tilde H
		&\wwrel=
			\harm{k+1} - \harm{t+1}\;.
	\end{align*}
	If $d\ge 0$, we have that
	$
			\E{A_{\vect q}}
		\wwrel\ge
			c \cdot \entropy[\ln](\vect q) \bin- d
	$
	all stochastic vectors $\vect q$.
\end{lemma}
\ifproceedings{
	The proof is similar to the one for the upper bound and omitted;
	it appears in full in the extended version of this paper.
}{
	\begin{proof}
	Let $\epsilon\in(0,1/e)$ with $d=d_\epsilon\ge 0$ be given;
	$c=c_\epsilon\ge 0$ holds for any $\epsilon$.
	The proof is similar to that of \wref{lem:equals-entropy-upper-bound-eps},
	so we emphasize the differences and skip identical parts.
	If $u=0$, \ie, $\vect q = ()$, the claim holds since $d\ge 0$.
	
	Now assume $u\ge1$ and that the claim holds for all (strictly) smaller universe sizes.
	As for the upper bound, we find from the recurrence using the inductive hypothesis
	\begin{align*}
			\E{A_{\vect q}}
		&\wwrel\ge
			1 \bin+
			\Eover* {P} {\sum_{r=1}^2 V_r \,(c \entropy[\ln](\vect{Z_r})-d)}
	\\[1ex]	&\wwrel{\eqwithref{eq:entropy-aggregation-P}}
			c\cdot \entropy[\ln](\vect q) 
			\bin+
			\underbrace{
				1 \bin-
				c\cdot \E[\big] {\entropy[\ln](\vect V, H)}
				\bin- d \cdot \E{\total{\vect V}}
			} _ {\varpi},
		\numberthis\label{eq:equals-def-varpi-lower-bound}
	\end{align*}
	and it remains to show that $\varpi \ge -d$.
	We consider the same two cases for $\lambda = \max_{1\le v\le u} q_v$.
	\begin{enumerate}%
	\item Case $\lambda < \epsilon$: \\
		In this case, we bound the expected entropy of partitioning,
		$\E{\entropy[\ln](\vect V, H)}$,
		from \emph{above}.
		Similar to the computation for the upper bound, we find
		\begin{align*}
				\E{\entropy[\ln](\vect V, H)}
			&\wwrel=
				\sum_{r=1}^2 \E[\big]{V_r \ln(1/V_r)}
				\bin+
				\E[\big]{H \ln(1/H)}
		\intertext{%
			using \weqref{eq:equals-Dr-Vr-all-qr-le-eps} and 
			$(x-\epsilon)\log(1/(x-\epsilon)) \le x \ln(1/x) + \epsilon'$ 
			for $0 \le \epsilon\le \epsilon'$ and $x\in[0,1]$, and 
			that $x\ln(1/x)$ is increasing for $x\in[0,1/e]$,
			this is
		}
			&\wwrel\le
				\sum_{r=1}^2 \E[\big]{D_r \ln(1/D_r) + 2\epsilon}
				\wbin+
				\epsilon \ln(1/\epsilon)
		\\	&\wwrel{\eqwithref[r]{eq:expected-entropy}}
				\tilde H
				\wbin+ 4\epsilon
				\bin+ \epsilon \ln(1/\epsilon)
		\end{align*}
		So $c$ satisfies
		$c \le 1 \big/ \E{\entropy[\ln](\vect V, H)}$,
		which implies
		\begin{align*}
				\varpi 
			&\wwrel\ge 
				1 \bin-
				\frac{1}{\E{\entropy[\ln](\vect V, H)}} \cdot 
					\E[\big] {\entropy[\ln](\vect V, H)}
				\bin- d \cdot \underbrace{\E{\total{\vect V}}}_{\le 1}
		\\[-1ex]	&\wwrel\ge
				-d.
		\end{align*}
		The inductive step is proven in this case.
		
	\item Case $\lambda \ge \epsilon$: \\
		In the second case, 
		there is a \emph{likely} value $v\in[u]$ with $q_v \ge \epsilon$.
		By the same arguments as in the proof of \wref{lem:equals-entropy-upper-bound-eps},
		we find
		\begin{align*}
				\E{H}
			&\wwrel{\relwithref{eq:equals-E-H-ge}\ge}
				\frac { \epsilon^{t+2}(1-\epsilon)^{t} } 
					{(t+1) \, \BetaFun(t+1,t+1) },
		\end{align*}
		so that
		$d \ge \bigl(c \ln(3) \bin- 1 \bigr) \big/ \E{H}$.
		This implies
		\begin{align*}
				\varpi + d 
			&\wwrel= 
				1 \bin-
				c\cdot \E[\big] {\underbrace{\entropy[\ln](\vect V, H) }_{\le \ln(3)}}
				\bin+ d \bigl( 1 - \E{\total{\vect V}}\bigr)
		\\	&\wwrel\ge
				1 \bin-
				c\cdot \ln(3)
				\bin+ \frac{c \ln(3) \bin- 1}{\E{H}} 
				\cdot \E{H}
		\\	&\wwrel=
				0.
		\end{align*}
		This concludes the inductive step also in the second case,
		so the claim holds for all stochastic vectors $\vect q$ by induction.
	\end{enumerate}
	\end{proof}
}%
We observe that $c_\epsilon$ converges to $1/\tilde H$ for both bounds as $\epsilon\to0$,
so there is hope to show $\E{A_{\vect q}} \sim \entropy[\ln](\vect q) / \tilde H$.
We have to be precise about the limiting process and error terms, though.

Let $(\ui{\vect q}i)_{i\in\N}$ be a sequence of universe distributions
for which $\mathcal H_i \ce \entropy[\ln](\ui{\vect q}i) \to \infty$ as $i\to\infty$.
Now, consider $c_\epsilon$ and $d_\epsilon$ from \wref{lem:equals-entropy-upper-bound-eps}.
As functions in $\epsilon$, they satisfy for $\epsilon\to0$
\begin{align*}
\numberthis\label{eq:c-d-asymptotic-upper-bound}
		c_\epsilon
	&\wwrel=
		\frac{1}{\tilde H} \bin\pm \Oh(\epsilon)\,,
&
		d_\epsilon
	&\wwrel=
		\Oh(\epsilon^{-t-2})
\end{align*}
Since the bounds above hold simultaneously for all feasible values of $\epsilon$,
we can let $\epsilon$ depend on $\mathcal H_i$.
If we set
\ifproceedings{\begin{align*}\SwapAboveDisplaySkip}{\begin{align*}}
\numberthis\label{eq:choice-of-epsilon}
		\epsilon \wrel= \epsilon_i 
	&\wwrel= 
		\mathcal H_i{}^{-\frac{1}{t+3}} 
\end{align*}
we have $\epsilon_i \to 0$ as $i\to\infty$ and so $c_{\epsilon_i} > 0$
for large enough~$i$. 
Then we have by \wref{lem:equals-entropy-upper-bound-eps}
\begin{align*}
		\E{A_{\ui{\vect q}i}}
	&\wwrel\le
		c_{\epsilon_i} \mathcal H_i \bin+ d_{\epsilon_i}
	\wwrel{\eqwithref{eq:c-d-asymptotic-upper-bound}}
		\frac{\mathcal H_i}{\tilde H}
		\wrel\pm \Oh\Bigl(
			\mathcal H_i{}^{\frac{t+2}{t+3}}
		\Bigr)
		\ifproceedings{}{,\qquad(i\to\infty).}
\numberthis\label{eq:equals-entropy-upper-bound-H-to-inf}
\end{align*}
\ifproceedings{as $i\to\infty$.}{}

Now consider the lower bound.
For $c_\epsilon$ and $d_\epsilon$ from 
\wref{lem:equals-entropy-lower-bound-eps} we similarly find as $\epsilon\to0$ that
\begin{align*}
\numberthis\label{eq:c-d-asymptotic-lower-bound}
		c_\epsilon
	&\wwrel=
		\frac{1}{\tilde H} \bin\pm \Oh(\epsilon \log \epsilon)\,,
&
		d_\epsilon
	&\wwrel=
		\Oh(\epsilon^{-t-2}) .
\end{align*}
With the same $\epsilon_i$ as above (\weqref{eq:choice-of-epsilon})
we have $d_{\epsilon_i} \ge 0$ for large enough $i$, so by \wref{lem:equals-entropy-lower-bound-eps}
it holds that
\begin{align*}
		\E{A_{\ui{\vect q}i}}
	&\wwrel\ge
		\frac{\mathcal H_i}{\tilde H} 
		\wrel\pm \Oh\Bigl(
			\mathcal H_i{}^{\frac{t+2}{t+3}}
			\log \mathcal H_i
		\Bigr)
		,\qquad(i\to\infty).
\numberthis\label{eq:equals-entropy-lower-bound-H-to-inf}
\end{align*}
Together with \weqref{eq:equals-entropy-upper-bound-H-to-inf},
this concludes the proof of \wref{thm:A-q-entropy-bounds}.
\end{proof}

%% file: equal-keys-lower-bound-proof.tex
\section{Lower Bound For \IID Sorting}
\label{sec:lower-bound}

We follow the elegant argument of \citet{MunroRaman1991} for multiset sorting
to obtain a lower bound for the discrete \iid model.
\ifproceedings{
	Since profiles are concentrated around the mean
	and the entropy function is smooth,
}{
	By averaging over the profiles, 
}%
we obtain essentially the result of their Theorem~4,
but with a weaker error term.

\begin{theorem}[Lower Bound]
\label{thm:lower-bound}
	Let $u = \Oh(n^\nu)$ for a constant $\nu\in[0,1)$
	and $\vect q \in (0,1)^u$ with $\total{\vect q} = 1$.
	For any constant $\varepsilon>0$,
	$
			\entropy[\ld](\vect q) n
			- n/\ln(2) 
			\wbin\pm o\bigl(n^{(1+\nu)/2+\varepsilon}\bigr)
	$
	ternary comparisons 
	are necessary in expectation as $n\to\infty$ 
	to sort $n$ \iid $\distFromWeights{\vect q}$ elements
	by any comparison-based algorithm.
\end{theorem}
\ifproceedings{}{%
	We remark that one might expect to require at least $\entropy[\ld](\vect q) n$ comparisons 
	since this is the entropy of a vector of $n$ \iid $\distFromWeights{\vect q}$ elements;
	but such entropy arguments must be taken with care:
	for $u$ much larger than $n$, $u\gg n$, we might have $\entropy(\vect q) \gg \ld n$; then
	$n \entropy(\vect q)\gg n\ld n$ is certainly \emph{not} a lower bound
	for sorting. 
	(\wref{thm:lower-bound} does not make a statement for such a case since the 
	error bound dominates then.)
}

\ifproceedings{
	The proof is omitted; it is given in the extended version of this paper.
}{
	
	\begin{proof}[\wref{thm:lower-bound}]
	Let $U_1,\ldots,U_n$ be $n$ \iid $\distFromWeights{\vect q}$ numbers
	and $V_1,\ldots,V_n$ be a random permutation of $[n]$.
	The $n$ vectors $(U_i,V_i)$ are then all distinct, and
	all $n!$ relative rankings \wrt lexicographic order are equally likely.
	
	Assume we can sort $U_1,\ldots,U_n$ with $\E{C_{n,\vect q}}$ 
	ternary comparisons on average.
	We use this method to partially sort the $n$ vectors $(U_1,V_1),\ldots,(U_n,V_n)$
	according to the first component only.
	We can then complete the sorting using
	Mergesort (separately) on each of the $u$ classes of elements with same first component.
	(Any sorting method must already have determined the borders between these classes 
	while sorting according to $U_1,\ldots,U_n$.)
	The total number of comparisons we use is then no more than
	\begin{align*}
			\E{C_{n,\vect q}} \bin+ \sum_{v=1}^u \E[\big]{ X_v \ld(X_v) }
		&\wwrel=
			\E{C_{n,\vect q}} \bin+ n \E[\bigg]{ 
					{\underbrace{ \sum_{v=1}^u \frac{X_v}n \ld\left(\frac{X_v}n\right) } _ {=\entropy[\ld](\vect X/n)}} 
				}
			\bin+ n\ld(n)
		\shortintertext{%
			with $\rho$ as in \wref{lem:entropy-of-multinomial-bound}
		}
		&\wwrel=
			n\ld(n) + \E{C_{n,\vect q}} 
			+ n \entropy[\ld](\vect q)
			\wbin\pm n \, \frac{\rho}{\ln(2)}
	\\	&\wwrel\ge
			n\ld(n) - n/\ln(2) \bin\pm \Oh(\log n)
	\end{align*}
	since the latter is the well-known lower bound on the average number of 
	(ternary or binary) comparisons for 
	sorting a random permutation of $n$ distinct elements, see, \eg, 
	Equation~5.3.1--(37) of \citet{Knuth1998}.
	It follows that 
	\begin{align*}
			\E{C_{n,\vect q}}
		&\wwrel\ge
			\entropy[\ld](\vect q) n
			- \frac{n}{\ln(2)}
			\wbin\pm n \, \frac{\rho}{\ln(2)}.
	\end{align*}
	For $u=\Oh(n^\nu)$ with $\nu\in[0,1)$
	we get asymptotically for any $\varepsilon>0$
	\begin{align*}
			\E{C_{n,\vect q}}
		&\wwrel\ge
			\entropy[\ld](\vect q) n
			- \frac{n}{\ln(2)}
			\wbin\pm o\bigl( n^{\frac{1+\nu}2+\varepsilon} \bigr).
	\end{align*}
	\end{proof}
}

%% file: equal-keys-proof-main-result.tex
\section{Proof of Main Result}
\label{sec:proof-main-result}

With all these preparations done,
the proof of our main result reduces to properly combining the ingredients 
developed above.

\begin{proof}[\wref{thm:main-result}]
Let the sequence $(\ui{\vect q}n)_{n\in\N}$ be given.
By assumption the model has many duplicates, \ie, $\mu_n = \Omega(n^{-1+\epsilon})$.
We thus obtain from \wref{thm:separating-n-and-q} that
	\begin{align*}
			\E{C_{n,\ui{\vect q}n}}
		&\wwrel=
			\E{A_{\ui{\vect q}n}} \cdot n \bin\pm \Oh(n^{1-\delta'})
			,\qquad(n\to\infty),
	\end{align*}
for any $\delta' \in (0,\epsilon)$.
If $\mathcal H_n = \entropy(\ui{\vect q}n)\to\infty$ as $n\to\infty$,
we can continue with \wref{thm:A-q-entropy-bounds} right away.
To cover the case that
$(\mathcal H_n)_{n\in\N}$ contains an infinite subsequence that is bounded, 
we add an error bound of $\Oh(n)$; 
this dominates $\E{A_{\ui{\vect q}n}} \cdot n$
(making the claim is essentially vacuous) for such inputs.
So in any case we have that
	\begin{align*}
			\E{C_{n,\ui{\vect q}n}}
		&\wwrel=
			\alpha_k \mathcal H_n \cdot n \bin\pm \Oh(\mathcal H_n^{1-\delta} n + n + n^{1-\delta'})
	\\	&\wwrel=
			\alpha_k \mathcal H_n \cdot n \bin\pm \Oh(\mathcal H_n^{1-\delta} n + n)
			\ifproceedings{}{,\qquad(n\to\infty)}
	\end{align*}
\ifproceedings{as $n\to\infty$}{}%
for any $\delta \in (0,\frac1{t+3}) = (0,\frac2{k+5})$.
The optimality claim follows directly by comparing with the lower bound in \wref{thm:lower-bound}.
\end{proof}

%% file: equal-keys-conclusion.tex
\section{Conclusion}
\label{sec:conclusion}

\FloatBarrier

Computer scientists are so accustomed to the random-permutation model 
that the possibility of duplicate keys in sorting is easily overlooked.
The following formulation (which is equivalent to random permutations)
makes the presence of an
underlying restriction more obvious:
we sort $n$ numbers drawn independently from the same \emph{continuous}
probability distribution.
The assumption of \iid samples is as natural as declaring all orderings equally likely,
and certainly adequate when no specific knowledge is available;
even more so for Quicksort where randomly shuffling the input prior to sorting is best practice.
The use of a \emph{continuous} distribution, though, is a restriction worth questioning; 
what happens if we use a discrete distribution instead?

The most striking difference certainly is that with a discrete distribution,
we expect to see equal elements appear in the input.
There are more differences, though.
If the distribution is (absolutely) continuous (\ie,
attains values in a real interval and has a continuous density), 
almost surely all elements are different and the ranks form a random permutation~\citep{Mahmoud2000sorting},
no matter how the continuous distribution itself looks like.
By contrast, different discrete universe distributions each yield a \emph{different} input model;
in particular the universe size $u$ can have a great influence.

In this paper, I~presented the first analysis of median-of-$k$ Quicksort
under the model of independently and identically distributed numbers from a
\emph{discrete} universe distribution.
I~showed that it uses asymptotically a factor 
$\alpha_k=\ln(2) / \bigl(\harm{k+1}-\harm{\smash{(k+1)/2}} \bigr)$ more comparisons 
than necessary in expectation.
Analytical complications necessitated the restriction to the case where 
every element is expected to appear $\Omega(n^{\epsilon})$ times
for some $\epsilon>0$, but I~conjecture that the result generalizes
beyond the limitations of the current techniques (see the comments below).

The very same statement 
\parenthesisclause{asymptotically a factor $\alpha_k$ over the lower bound} 
holds in the \iid model with a continuous distribution,
so, apart from the above restriction, we can say that median-of-$k$ Quicksort is 
asymptotically optimal to within the constant factor $\alpha_k$ for \emph{any}
randomly ordered input.
It is reassuring (and somewhat remarkable given the big differences in the derivations that lead to it) 
that we obtain the same constant in both cases: 
the relative benefit of pivot sampling is independent of the presence or absence of equal keys.
\wref{tab:alpha-k} shows the first values of $\alpha_k$;
most savings happen for small~$k$.

\begin{table}[tbhp]
	\begin{captionbeside}{%
		A few values of $\alpha_k$ and the relative improvement
		over the case without sampling.%
	}
	\qquad\qquad
	\plaincenter{%
	\small
	\begin{tabular}{clc}
	\toprule
		$k$ & $\alpha_k$ & saving over $k=1$\\
	\midrule
		$1     $ & $1.38629 $ & ---\\
		$3     $ & $1.18825 $ & $14.3\,\%$ \\
		$5     $ & $1.12402 $ & $18.9\,\%$ \\
		$7     $ & $1.09239 $ & $21.2\,\%$ \\
		$9     $ & $1.07359 $ & $22.6\,\%$ \\[.5ex]
		$\infty$ & $1       $ & $27.9\,\%$ \\
	\bottomrule
	\end{tabular}%
	}
	\qquad\qquad
	\ifproceedings{\label{tab:alpha-k}}{}
	\end{captionbeside}
	\ifproceedings{}{\label{tab:alpha-k}}
\end{table}

\ifproceedings{}{
	Exploiting the correspondence of Quicksort and search trees,
	we can also express the result in terms of trees:
	The expected search costs in search trees built from \emph{continuous} \iid 
	(or randomly permuted and \emph{distinct}) 
	data is of order $\log n$,
	whereas for \emph{discrete} \iid data (with many duplicates), 
	it is of order $\entropy(\vect q)$.
	Fringe-balancing allows to lower the constant of proportionality, $\alpha_k$,
	and has the \emph{same relative effect} in both cases.

	A particularly simple example of distributions covered by our analysis is
	the uniform distribution, $\vect q = (\frac1u,\ldots,\frac1u)$, 
	where $u = u_n = \Oh(n^{1-\epsilon})$.
	This model coincides with the random $u$-ary files studied by \citet{Sedgewick1977a}
	in his 1977 article.
	Since the entropy of the discrete uniform distribution is simply $\entropy(\vect q) = \ld (u)$, 
	we obtain $\E{C_{n,\vect q}} \sim \alpha_k n \ld (u)$
	\parenthesissign the same form as for random permutation only with $\ld n$ replaced by $\ld u$.
	For the special case of the uniform distribution, we can also strengthen
	the error bound of \wref{thm:A-q-entropy-bounds} for $\E{A_{\vect q}}$ 
	to $\Oh(1)$ using different techniques; 
	see my \phdthesis~\citep[Sec.\,8.7.7]{Wild2016} for details.
}

I surveyed previous results on multiset sorting, which uses a different input model.
We cannot directly relate the result,
but we have seen (in \wref{sec:relation-of-two-models}) that the
\iid model yields at least an upper bound.
(A more fine-grained discussion of the relation of the two models
is deferred to a full version of this article.)

Since
$\harm{k+1} - \harm{(k+1)/2} \sim \ln(k+1) - \ln((k+1)/2) = \ln(2)$
as $k\to\infty$, we have $\alpha_k \to 1$,
so we have confirmed the conjecture of Sedgewick and Bentley for inputs with many duplicates:
by increasing $k$, we obtain sorting methods that are arbitrarily close to the 
asymptotically optimal number of comparisons for randomly ordered inputs, and
in particular for random permutations of a multiset.

\paragraph{Extensions and Future Directions}

The methods used in this paper can readily be generalized to
analyze skewed sampling,
\ie, when we pick an order statistic other than the median of the sample.
It is known that the number of comparisons is minimal for the median~\citep{Martinez2001},
so this might not be very interesting in its own right,
but it can be used to analyze \textsl{ninther} sampling
and similar schemes~\citep[Sec.\,6.4]{Wild2016}.

\ifproceedings{}{
	I conjecture that $\alpha_k \.\entropy(\vect x/n) \.n$ is the correct leading term
	for any profile $\vect x \in \N^u$.
	If we have, however, a discrete \iid model with $q_v \ll n^{-1}$ so that $v$ is unlikely to
	be present in the input \emph{at all}, a modified ``entropy'' must appear instead of $\entropy(\vect q)$.
	(This case cannot occur in the multiset model since $x_v \ge 1$.)
	In the limit when all individual $q_i$ are small, this modified entropy would have to equal $\ld(n)$.
	Identifying such a unified expression is an interesting challenge.
}

Another line of future research is to extend the present analysis to
multiway Quicksort, \eg, the Yaroslavskiy-Bentley-Bloch dual-pivot Quicksort used in Java.
The uniform case is known~\citep{Wild2016}, but the techniques of the present paper 
do not carry over:
the partitioning costs for such methods also depend on $\vect q$, 
which means that we do not get matching lower and upper bounds for $\E{A_{\vect q}}$ 
using the method of \wref{sec:proof-entropy-bounds} any more.

%% file: equal-keys-notation.tex
\section{Index of Notation}
\label{app:notations}

In this appendix, I collect the notations used in this work.

\newlength\notationwidth
\setlength\notationwidth{9em}

\subsection{Generic Mathematical Notation}
\begin{notations}[\notationwidth]
\notation{$\N$, $\N_0$, $\Z$, $\Q$, $\R$, $\C$}
	natural numbers $\N = \{1,2,3,\ldots\}$, 
	$\N_0 = \N \cup \{0\}$,
	integers $\Z = \{\ldots,-2,-1,0,1,2,\ldots\}$,
	rational numbers $\Q = \{p/q : p\in\Z \wedge q\in\N  \}$,
	real numbers $\R$, and complex numbers $\C$.
\notation{$\R_{>1}$, $\N_{\ge3}$ etc.}
	restricted sets $X_\mathrm{pred} = \{x\in X : x \text{ fulfills } \mathrm{pred} \}$.
\notation{$\ln(n)$, $\ld(n)$}
	natural and binary logarithm; $\ln(n) = \log_e(n)$, $\ld(n) = \log_2(n)$.
\notation{$\vect x$}
	to emphasize that $\vect x$ is a vector, it is written in \textbf{bold;}\\
	components of the vector are not written in bold: $\vect x = (x_1,\ldots,x_d)$;\\
	unless stated otherwise, all vectors are column vectors.
\notation{$X$}
	to emphasize that $X$ is a random variable it is Capitalized.
\notation{$[a,b)$}
	real intervals, the end points with round parentheses are excluded, 
	those with square brackets are included.
\notation{$[m..n]$, $[n]$}
	integer intervals, $[m..n] = \{m,m+1,\ldots,n\}$;
	$[n] = [1..n]$.
\notation{$\|\vect x\|_p$}
	$p$-norm;
	for $\vect x\in\R^d$ and $p\in\R_{\ge1}$ we have 
	$\|\vect x\|_p = \bigl(\sum_{r=1}^{d} |x_r|\bigr)^{1/p}$.
\notation{$\|\vect x\|_\infty$}
	$\infty$-norm or maximum-norm; 
	for $\vect x\in\R^d$ we have $\|\vect x\|_\infty = \max_{r=1,\ldots,d} |x_r|$.
\notation{$\vect x+1$, $2^{\vect x}$, $f(\vect x)$}
	element-wise application on vectors; $(x_1,\ldots,x_d) + 1  = (x_1+1,\ldots,x_d+1)$ and 
	$2^{\vect x} = (2^{x_1},\ldots,2^{x_d})$; 
	for any function $f:\C\to\C$ write $f(\vect x) = (f(x_1),\ldots,f(x_d))$ etc.
\notation{$\total{\vect x}$}
	``total'' of a vector; for $\vect x = (x_1,\ldots,x_d)$, 
	we have $\total{\vect x} = \sum_{i=1}^d x_i$.
\notation{$\vect x^T$, $\vect x^T \vect y$}
	``transpose'' of vector/matrix $\vect x$; 
	for $\vect x, \vect y \in \R^n$, I write
	$\vect x^T \vect y = \sum_{i=1}^n x_i y_i$.
\notation{$\harm n$}
	$n$th harmonic number; $\harm n = \sum_{i=1}^n 1/i$.
\notation{$\Oh(f(n))$, $\pm\Oh(f(n))$, $\Omega$, $\Theta$, $\sim$}
	asymptotic notation as defined, \eg, by \citet[Section A.2]{Flajolet2009};
	$f=g\pm\Oh(h)$ is equivalent to $|f-g| \in \Oh(|h|)$.
\notation{$x \pm y$}
	$x$ with absolute error $|y|$; formally the interval $x \pm y = [x-|y|,x+|y|]$;
	as with $\Oh$-terms, I use ``one-way equalities'': $z=x\pm y$ instead of $z \in x \pm y$.
\notation{$\Gamma(z)$}
	the gamma function, $\Gamma(z) = \int_0^\infty t^{z-1}e^{-t} \, dt$.
\notation{$\psi(z)$}
	the digamma function, $\psi(z) = \frac d{dz} \ln(\Gamma(z))$.
\notation{$\BetaFun(\alpha,\beta)$}
	beta function; 
	$\BetaFun(\alpha,\beta) = \int_0^1 z^{\alpha-1}(1-z)^{\beta-1} \,dz = \Gamma(\alpha)\Gamma(\beta) / \Gamma(\alpha+\beta)$.
\notation{$I_{a,b}(\alpha,\beta)$}
	incomplete regularized beta function;
	$I_{a,b}(\alpha,\beta) = \int_a^b z^{\alpha-1}(1-z)^{\beta-1} \,dz \,\big/\, \BetaFun(\alpha,\beta)$.
\end{notations}

\subsection{Stochastics-related Notation}
\begin{notations}[\notationwidth]
\notation{$\Prob{E}$, $\Prob{X=x}$}
	probability of an event $E$ resp.\ probability for random variable $X$ to
	attain value $x$.
\notation{$\E{X}$}
	expected value of $X$; I write $\E{X\given Y}$ for the conditional expectation
	of $X$ given $Y$, and $\Eover X{f(X)}$ to emphasize that expectation is taken 
	\wrt random variable~$X$.
\notation{$X\eqdist Y$}
	equality in distribution; $X$ and $Y$ have the same distribution.
\notation{$\indicatornobraces{E}$, $\indicator{X\le 5}$}
	indicator variable for event $E$, \ie, $\indicatornobraces{E}$ is $1$ if $E$
	occurs and $0$ otherwise;
	$\{X\le 5\}$ denotes the event induced by the expression $X \le 5$.
\notation{$\bernoulli(p)$}
	Bernoulli distributed random variable;
	$p\in[0,1]$.
\notation{$\distFromWeights{\vect p}$}
	discrete random variable with weights $\vect p$; for $\vect p\in[0,1]^d$,
	for $I\eqdist \distFromWeights{\vect p}$, we have $I \in [1..d]$ and
	$\Prob{I=i} = p_i$ for $i\in[d]$ and $0$ otherwise.
\notation{$\betadist(\alpha,\beta)$}
	beta distributed random variable with shape parameters $\alpha\in\R_{>0}$ and $\beta\in\R_{>0}$.
\notation{$\binomial(n,p)$}
	binomial distributed random variable with $n\in\N_0$ trials and success probability $p\in[0,1]$;
	$X\eqdist \binomial(n,p)$ is equivalent to $(X,n-X)\eqdist\multinomial(n;p,1-p)$.
\notation{$\multinomial(n,\vect p)$}
	multinomially distributed random variable; 
	$n\in\N_0$ and $\vect p \in [0,1]^d$ with $\total{\vect p} = 1$.
\notation{$\hypergeometric(k,r,n)$}
	hypergeometrically distributed random variable;
	$n\in\N$, $k,r,\in\{1,\ldots,n\}$.  
\notation{$\entropy(\vect p)$, $\entropy[\ld](\vect p)$, $\entropy[\ln](\vect p)$}
	Shannon entropy of information theory; 
	$\entropy[\ld](p_1,\ldots,p_d) = \sum_{r=1}^d p_r \ld (1/p_r)$;
	similarly $\entropy[\ln]$ is the base-$e$ entropy.
	$\entropy[\ln](p_1,\ldots,p_d) = \sum_{r=1}^d p_r \ln (1/p_r)$;\\
	I write $\entropy$ for $\entropy[\ld]$.
\notation{stochastic vector}
	A vector $\vect p$ is called stochastic if $0\le \vect p\le 1$ and $\total{\vect p} = 1$.
\notation{\whp (event)}
	Let $E=E(n)$ be an event that depends on a parameter $n$. 
	I say ``$E$ occurs \whp'' if $\Prob{E(n)} = 1 \pm \Oh(n^{-c})$ as $n\to\infty$ for \emph{any} constant $c$.
\notation{\whp (bound)}
	Let $X_1,X_2,\ldots$ be a sequence of real random variables.
	I say ``$X_n = \Oh(g(n))$ \whp'' if for every constant $c$ there is a constant $d$ so that
	$\Prob{|X_n| \le d |g(n)|} = 1 \pm \Oh(n^{-c})$ as $n\to\infty$.
\end{notations}

\subsection{Notation for the Algorithm}
\begin{notations}[\notationwidth]
\notation{$n$}
	length of the input array, \ie, the input size.
\notation{$u$}
	universe size $u\in\N$.
\notation{$\vect x$}
	fixed profile in the multiset model; 
	$\vect x \in \N^u$, $\total{\vect x} = n$
\notation{$\vect q$}
	probability weights of the (discrete) \iid model;
	$\vect q\in(0,1)^u$, $\total{\vect q} = 1$.
\notation{$U_i$}
	$i$th element of the input; 
	in the \iid model, $U_i \eqdist \distFromWeights{\vect q}$,
	for all $i$ and $(U_1,\ldots,U_n)$ are (mutually) independent.
\notation{$k$, $t$}
	sample size $k = 2t+1$ for $t\in \N_{0}$;
\notation{ternary comparison}
	operation to compare two elements, outcome is either $<$, $=$ or $>$.
\notation{$\vect X$}
	profile of the input; $X_v$ is the number of occurrences of value $v$
	in $U_1,\ldots,U_n$;
	$\vect X \eqdist \multinomial(n,\vect q)$.
\end{notations}

\subsection{Notation for the Analysis}
\label{sec:notation-analysis-equals}
\begin{notations}[\notationwidth]
\notation{$C_{\vect x}$}
	(random) number of (ternary) comparisons needed by \wref{alg:list-quicksort}
	to sort a random permutation of a multiset with profile $\vect x$;
	the dependence on $k$ is left implicit.
\notation{$C_{n,\vect q}$}
	(random) number of (ternary) comparisons needed by \wref{alg:list-quicksort}
	to sort $n$ \iid $\distFromWeights{\vect q}$ numbers;
	the dependence on $k$ is left implicit.
\notation{$\Pi$}
	continuous pivot value; $\Pi\eqdist\betadist(t+1,t+1)$
\notation{$\vect D = (D_1,D_2)$}
	continuous spacings of the unit interval $(0,1)$ induced by $\Pi$,
	\ie, $\vect D = (\Pi,1-\Pi)$.
\notation{$P$}
	(random) value of chosen pivot in the first partitioning step;
	$
			P
		\wwrel=
			F^{-1}_U(\Pi)
	$, for $F_U$ the cumulative distribution function
	of $\distFromWeights{\vect q}$.
\notation{$\vect V = (V_1,V_2)$}
	non-pivot class probabilities, see \weqref{eq:def-V1-V2-H}.
\notation{$H$}
	hitting probability, \weqref{eq:def-V1-V2-H}.
\notation{$\vect Z = (\vect{Z_1},\vect{Z_2})$}
	zoomed-in distributions; see \weqref{eq:def-Z}.
\notation{$\vect\Gamma$}
	$\vect\Gamma\in\N_{0}^u$;
	node-depths vector in search tree; 
	$\Gamma_v$ is the random cost of searching value $v\in[u]$.
\notation{$A_{\vect q}$}
	random $\vect q$-weighted average node depth of a $k$-fringe-balanced tree 
	built from inserting $\distFromWeights{\vect q}$ elements till saturation, 
	$A_{\vect q} = \vect\Gamma^T \vect q$;
	see \weqref{eq:equals-A-q-recurrence}.
\end{notations}

%% file: equal-keys-height.tex
\FloatBarrier

\section{Height of Recursion Trees}
\label{app:height-of-recursion-trees}

In this appendix, we recapitulate the well-known folklore result that
Quicksort recursion trees
\parenthesisclause{or equivalently binary search trees}
have logarithmic height with high probability.
We will show that the same is true for fringe-balanced trees 
(and thus median-of-$k$ Quicksort)
and inputs with equal keys;
this is not surprising as the latter are more balanced than ordinary BSTs,
but a formal proof is in order for two reasons:
we 
(a) employ a stricter notion of ``\whp'' (cf.\ \wref{sec:notation}) 
than often used, and 
(b) there is a pitfall in the folklore proof using good-split-bad-split indicators
that has not been rigorously addressed anywhere (to my knowledge).

We start with a result of \citet{Mahmoud1992evolution} (page 101) 
for ordinary binary search trees
that is a prototypical for the types of results we are after in this section.
(We will confine ourselves to slightly less accurate statements, though.)
\begin{lemma}[BSTs Have Log-Height with High Probability]
\label{lem:bst-height-hosam}
	\mbox{}\\
	For any $\epsilon>0$ it holds:
	The probability that a binary search tree built from a random permutation of
	$n$ distinct elements has height $\ge (\alpha+\epsilon) \ln n$ is 
	at most $K c_{\alpha+\epsilon} n^{-\eta_{\alpha+\epsilon}} = 
	\Oh(n^{-\eta_{\alpha+\epsilon}})$ where
	$K>0$ is some constant (independent of $\epsilon$), 
	$c_x = 1 / \bigl( \Gamma(x)(1-2/x)\bigr)$,
	$\eta_x = -\bigl(x\ln(2/x) +x -1 \bigr)$
	and $\alpha\approx4.31107$ is the unique root of $\eta_x$ for $x\in(4,5)$.
\qed\end{lemma}
By \wref{lem:equals-equivalence-recursion-tree-search-tree-exact} this result
translates immediately to the height of Quicksort recursion trees.

\Citeauthor{Mahmoud1992evolution}'s derivation is based on extensive knowledge on random BSTs
(also given in~\cite{Mahmoud1992evolution}), in particular he uses the exact 
expected number of leaves at any given level.
Generalizing this for fringe-balanced trees with duplicates seems a daunting task.

Variants of \wref{lem:bst-height-hosam} are sometimes covered in (advanced) 
algorithms courses and textbooks,
see, \eg, 
Exercise\,4.20 of \citet{MitzenmacherUpfal2005}
Section\,2.4 of \citet{Dubhashi2009},
and Section~4.7 of \citet{Erickson2017}.
There, a more intuitive argument is given by bounding the probability of many ``bad''
splits at nodes using Chernoff bounds, however neither of these
resources gives a detailed formal proof.
The argument is appealingly simple, and fairly easy to extend, but it indeed requires some care
to really guarantee the needed independence. 
We therefore give a detailed proof below and obtain the following result.

\begin{lemma}[Logarithmic Height of Recursion Trees With High Probability]
\label{lem:equals-recursion-trees-log-height-whp}
	~\\
	Let $\mathcal T$ be a $k$-fringe-balanced search tree built by
	inserting $n$ \iid $\distFromWeights{\vect q}$ numbers into an initially empty tree.
	Then $\mathcal T$ has height $\Oh(\log n)$ \whp; more precisely
	$\mathcal T$ has height $\ge c \ln n$ with probability $\le 2 n^{\eta}$ for $n \ge n_0 = e^{30\,000}$
	(\ie, with probability $\Oh(n^\eta)$),
	where
	\begin{align*}
			\eta
		&\wwrel=
			1-2c\delta^2
	\\
			\delta
		&\wwrel=
			p - \frac1c \cdot \biggl( \frac1{\ln(1/\alpha)}+1 \biggr)
	\\
			p
		&\wwrel=
			0.99 - 2 I_{\alpha-0.01,1}(t+1,t+1)
		\wwrel=
			0.99 - 2\,\frac{k!}{t!\, t!} \int_{\alpha-0.01}^1 x^t (1-x)^t \,dx
	\end{align*}
	where $\alpha \in (\frac12,1)$ is a parameter that can be chosen arbitrarily
	as long as $\delta > 0$.
	This result is independent of \(\vect q\), and holds
	in particular when \(\vect q\) depends on \(n\).
\end{lemma}

\begin{table}[tbhp]
	\begin{captionbeside}{%
		A few exemplary values for the constants in \wref{lem:equals-recursion-trees-log-height-whp}.%
	}
		\hspace*{1em}\begin{tabular}{ccc}
		\toprule
			  $k$    & $c$  &  $\eta$  \\
		\midrule
			  $1$    & $12$ & $-1.72$  \\
			  $1$    & $13$ & $-2.86$  \\
			  $1$    & $20$ & $-11.53$  \\[.5ex]
			  $3$    & $9$ & $-1.94$  \\
			  $3$    & $10$ & $-3.37$  \\
			  $3$    & $20$ & $-19.02$ \\
		\bottomrule
		\end{tabular}
		\quad
		\begin{tabular}{ccc}
		\toprule
			  $k$    & $c$  &  $\eta$  \\
		\midrule
			  $5$    & $7$  & $-0.70$  \\
			  $5$    & $8$  & $-2.24$  \\
			  $5$    & $20$ & $-22.45$ \\[.5ex]
			$\infty$ & $4$  & $-2.03$  \\
			$\infty$ & $5$  & $-4.01$  \\
			$\infty$ & $20$  & $-33.71$  \\
		\bottomrule
		\end{tabular}%
		\hspace*{1em}
	\end{captionbeside}
	\label{tab:values-lemrecursion-trees-log-height-whp}
\end{table}

The given value for $n_0$ is clearly ridiculous;
the intention of \wref{lem:equals-recursion-trees-log-height-whp} is 
to give a rigorous proof of the asymptotic result.
As we will see in the proof, one can trade larger $c$ values for smaller $n_0$,
and all constants could be improved by a stronger version of the Chernoff bound.
It is not my intention to do so here.

We note that for Quicksort, an alternative route is to analyze a modified version of the algorithm
that makes some technicalities vanish and performs no better than the original Quicksort; 
see, \eg, \citet{Seidel1993}.
Moreover, much stronger concentration results are known for the overall number of \emph{comparisons} in 
Quicksort, see \citet{McDiarmidHayward1996} or the streamlined description in 
Section~7.6 of \citet{Dubhashi2009}.
There the exponent of the probability bound is arbitrarily large for one \emph{fixed} 
bound $c \ln n$.
It seems not possible to obtain such a stronger bound for the \emph{height} of the tree, though.

Concentration results are typically framed in terms of randomized Quicksort.
To emphasize the relation between Quicksort recursion trees and search trees,
our reference Quicksort (\wref{alg:list-quicksort}) is not randomized,
but deterministically chooses the first elements for the sample.
Causing an exceptionally high recursion tree can hence be attributed to a specific 
input in our scenario; the following definition expresses that idea.

\begin{definition}[Height-degenerate Inputs]
\label{def:height-degenerate}
	An input of length $n$ is called \textit{\(h\)-height-degenerate} (\wrt our reference
	fat-pivot median-of-$k$ Quicksort) if the recursion tree of $\proc{Quicksort}_k$
	on this input has a height $> h \ln(n)$.
\end{definition}
From \wref{lem:equals-recursion-trees-log-height-whp} we immediately obtain
the following fact.

\begin{corollary}[Probability of Height-Degeneracy]
\label{cor:recursion-tree-log-height-weak-whp}
	For any $k$, the probability that an input of $n$ elements
	is $13$-height-degenerate is in $\Oh(1/n^2)$ as $n\to\infty$.
\qed\end{corollary}
We will in the following simply use ``height-degenerate'' to mean $13$-height-degenerate.
\smallskip
\smallskip

\separatedpar

\begin{proof}[\wref{lem:equals-recursion-trees-log-height-whp}]
	Let $k=2t+1$ be given.
	We will follow the folklore proof that the height of randomly grown BSTs 
	is typically logarithmic:
	we determine a constant probability $p>0$ (independent of $n$) so that a single
	partitioning step yields a reasonably balanced split;
	since long paths in the recursion tree cannot contain more than 
	a certain number of such balanced nodes, we can bound the probability of seeing 
	such a long path in a recursion tree.

	\paragraph{Outline}
	A few technicalities need to be addressed in the formal proof:
	\begin{enumerate}[itemsep=.5ex]
	\item 
		To tweak constants, we will introduce a parameter $\alpha\in(1/2,1)$ 
		and require subproblems at a node $v$ to contain at most $\alpha n(v)$ elements,
		where $n(v)$ is the size of the sublist that $v$ corresponds to.
	\item 
		The probability to lose a given fraction of elements depends on the 
		subproblem size (a discretization effect) and
		in our case of inputs with duplicates also on $\vect q$.
		We can therefore only \emph{lower bound} the probability for a 
		balanced node by a constant $p$.
		To obtain a term for $p$ that we can actually evaluate,
		we resort to asymptotic approximations, which 
		are only a valid lower bound for sublist sizes larger than a threshold $n_0$.
	\item 
		Since the precise probability to lose a fraction of elements depends
		on the sublist size, also the events that a certain
		node be balanced are dependent.
	\end{enumerate}
	The last point is a problem since the standard Chernoff bounds requires 
	mutual independence; nevertheless this issue is not addressed in the sources cited above.
	Since the pivot choices themselves are done independently, we only need a little trick
	to obtain independent events:
	A node is only considered good when it is balanced and additionally a biased
	coin flip yields heads, where the bias is chosen so that the overall probability for a
	good node is the same for all nodes.
	This gives us a sequence of independent indicator variables to which we can apply
	our machinery as usual.
	
	\paragraph{Balanced nodes}
	We use the following notation in this proof.
	$v$ denotes a node of the recursion tree.
	By $n(v)$ we mean the sublist size at $v$, \ie, the \underbar number of elements in the subproblem
	of the recursive call that $v$ corresponds to.
	$d(v)$ denotes the \underbar depth of $v$, \ie, 
	the number of nodes on the path from the root to $v$, including endpoints.
	Finally, if $n(v) \ge k$, 
	we use $J_r(v)$, $r=1,2$, to the denote the size of the $r$th subproblem at $v$; 
	this is the subproblem size of $v$'s left resp.\ right child.
	
	We are now in the position to formalize the notion of \emph{balanced} nodes:
	Let $\alpha\in(1/2,1)$ be a fixed number and $n_0\ge k$ a constant 
	(to be chosen later).
	We call an inner node \emph{\(\alpha\)-balanced} if $n(v)\le n_0$
	or if $J_r(v)/n(v) \le \alpha$ for both $r=1,2$.
	($(\alpha,n_0)$-balanced would be more appropriate; 
	the dependence on $n_0$ is understood implicitly).
	An $\alpha$-balanced node hence has
	no subproblem with more than an $\alpha$-fraction of the elements,
	or has a negligibly small sublist anyway.
	
	The key idea is now that any path in a recursion tree for $n$ elements can contain at most 
	\begin{align*}
			\log_{\alpha}(n_0/n) 
		\wwrel= 
			\log_{1/\alpha}(n/n_0) 
		\wwrel= 
			\frac1{\ln(1/\alpha)}\ln(n/n_0)
		\wwrel\le 
			\frac1{\ln(1/\alpha)}\ln(n)
	\end{align*}
	$\alpha$-balanced nodes before reaching a node $v$ with $n(v) \le n_0$
	since considering only the size reduction at these $\alpha$-balanced nodes already 
	reduces the $n$ initial elements to
	$\le \alpha^{\log_{\alpha}(n_0/n)} n = n_0$ elements.
	From there on, at most $n_0$ additional $\alpha$-balanced nodes can follow
	since each reduces the subproblem size by at least one.

	\paragraph{Balanced is not good enough}
	We are now formalizing the idea sketched above to obtain \emph{independent} indicator variables.
	For a node with $n(v)\ge n_0$, we define 
	\(p_b(v) = \Prob{\text{$v$ $\alpha$-balanced} \given n(v)}\).
	Note that $p_b(v)$ only depends on $n(v)$ and $k$,
	but since the number of its possible subproblems sizes is finite, $p_b(v)$ will necessarily differ
	for different values of $n(v)$, 
	even without pivot sampling ($k=1$) and without our threshold ($n_0 = 0$).
	
	However, we will show below that we can find choices for $n_0$ and $\alpha$ so that at least 
	$p_b(v) \ge p$ for a given constant $p = p(\alpha)$ in all possible trees.
	
	For such a $p$, we call a node $v$ \emph{\((\alpha,p)\)-good} if it is $\alpha$-balanced 
	\emph{and additionally} $B(v) = 1$,
	where $B(v) \eqdist \bernoulli(\frac{p}{p_b(v)})$ which is independent of all other random variables.
	The distribution of $B(v)$ is conditional on the given tree via $n(v)$,
	so one might imagine first drawing a random recursion tree, and then assigning its
	nodes labels good or bad.
	
	Since every good node is also balanced, we cannot have more than $\frac1{\ln(1/\alpha)}\ln(n) + n_0$ 
	good nodes on any path in a recursion tree for input size $n$.

	\paragraph{Probability of long paths}
	
	We can now bound the probability of having a tree of height $\ge c \ln n$.
	First note that the overall number of nodes is bounded by $n$ 
	(at least one pivot is removed in each step), 
	so the number of leaves in the recursion tree is trivially bounded by $n$.
	By the union bound, the probability that \emph{any} of these leaves has depth $\ge h$ is at
	most $n$ times the probability that one leaf has depth $\ge h$.
	Let hence $v$ be one of the leaves in the recursion tree and let $v_1,\ldots,v_{d(v)}=v$ 
	be the nodes on the path from the root $v_1$ down to $v$; 
	recall that $d(v)$ is the depth of leaf $v$.
	
	The sublist corresponding to $v$ is the result of $d(v)-1$ successive partitioning steps, 
	each of which is either $(\alpha,p)$-good or not.
	Let $G_1,\ldots,G_{d(v)-1}$ be the corresponding indicator random variables where
	$G_i$ is $1$ if and only if $v_i$ is good.
	By construction, we have $\Prob{G_i = 1} = p_b(v) \cdot \frac{p}{p_b(v)} = p$
	independently of the tree.
	We now extend $G_1,\ldots,G_{d(v)-1}$ to an infinite sequence of random variables 
	by \iid $\bernoulli(p)$ variables for all $i\ge d(v)$;
	the $G_i$ then form an infinite sequence of \iid random variables.

	Now recall that the number of $\alpha$-balanced nodes on any path is 
	at most $n_0 + \frac1{\ln(1/\alpha)} \ln n$.
	For any $h\in\N$, we thus have 
	\begin{align*}
			\Prob[\big]{d(v) \ge h}
		&\wwrel\le
			\Prob*{G_1+\cdots+G_h \rel\le n_0 + \frac1{\ln(1/\alpha)} \ln n }
	\\	&\wwrel=
			\Prob[\bigg]{X_h \rel\le \tilde\gamma \ln n }
	\\	&\wwrel\le
			\Prob[\bigg]{X_h \rel\le \gamma \ln n },
	\\
		\shortintertext{where}
			X_h
		&\wwrel\eqdist
			\binomial(h,p),
	\\
			\tilde\gamma
		&\wwrel=
			\frac1{\ln(1/\alpha)} + \frac{n_0}{\ln n},
	\\	
			\gamma
		&\wwrel=
			\frac1{\ln(1/\alpha)} + 1
		\wwrel\ge
			\tilde\gamma
			, \qquad (n\ge e^{n_0}).
	\end{align*}
	For $n \ge e^{n_0}$
	and $h$ so that $\delta \ce p - \gamma \ln(n) / h > 0$ we then have
	\begin{align*}
			\Prob[\big]{d(v) \ge h}
		&\wwrel\le
			\Prob[\big]{X_h \le \gamma \ln n}
	\\	&\wwrel=
			\Prob*{p - \frac{X_h}h \rel\ge p - \gamma \,\frac{\ln(n)}h } ;
	\\	&\wwrel\le
			\Prob*{
				\biggl|\frac{X_h}h - p \biggr| 
				\rel\ge \delta
			}
	\\ &\wwrel{\relwithref[r]{lem:chernoff-bound-binomial}\le}
			2 \exp(-2 \delta^2 h)
			.
	\end{align*}
	With $h = c \ln(n)$, we have 
	\(
			\delta
		\wwrel=
			p - \frac\gamma c
	\),
	which is independent of $n$, and positive for any $c > \gamma/p$.
	With this bound we finally find 
	\begin{align*}
			\Prob[\big]{\mathcal T \text{ has height } \ge c \ln n}
		&\wwrel\le
			n \,\Prob[\big]{d(v) \ge c \ln n}
	\\	&\wwrel\le 
			2n \exp(-2 c \delta^2 \ln n)
	\\	&\wwrel=
			2n^{1-2c\delta^2},
	\end{align*}
	which implies the claim.

	\paragraph{A lower bound for balanced nodes: how to choose the constants}
	It remains to show that we can actually find values $n_0$ and $p = p(\alpha)$ 
	(at least for some choices of $\alpha$) so that 
	so that $p_b(v) \ge p$ in all nodes $v$ in all possible trees $\mathcal T$.
	
	To this end, we derive an \emph{upper} bound for the probability $1-p_b(v)$ that the root $v$ of 
	$\mathcal T$ is \emph{not} $(\alpha,n_0)$-balanced.
	For $n \le n_0$, we are done since $p_b(v) = 1$ and any $p>0$ will do.
	So assume $n > n_0$.
	By the union bound we have
	\begin{align*}
	\numberthis\label{eq:equals-prob-non-alpha-balanced-union-bound}
			\Prob{v \text{ not $\alpha$-balanced}} 
		\wwrel= 
			\Prob[\big]{ \exists r : \ui{J_r}n \ge \alpha n } 
		\wwrel\le 
			\sum_r \Prob[\big]{ \ui{J_r}n \ge \alpha n},
	\end{align*}
	so it suffices to consider the subproblems $r=1,2$ in isolation.
	
	Recall that after a pivot value $P$ is chosen according to
	\wpeqref{eq:dist-P}, the probabilities $V_r$ for any other
	element to belong to the $r$th subproblem are fully determined.
	$P$ in turn is fully determined by the choice of $\vect D$.
	Hence, conditionally on $\vect D$, we have
	$\binomial(n-k,V_r)$ elements that go to the $r$th subproblem,
	plus up to $t$ from the sample.
	Moreover we always have $V_r \le D_r$ (cf.\ \wref{fig:equals-stochastic-models-probs}).
	Conditional on $\vect D$, $\ui{J}n_r$ is hence smaller than
	$\tilde J_r \eqdist \binomial(n,D_r) + t$ in stochastic order, \ie, for all $j$ we have that
	$\Prob[\big]{J_r \ge j \given\vect D} \le \Prob[\big]{ \tilde J_r \ge j \given\vect D}$.
	(By averaging over all choices for $\vect D$ the same relation holds also unconditionally.)

	This is nothing but the precise formulation of the fact that 
	(in stochastic order)
	subproblem sizes for inputs with duplicates 
	are no larger than for random-permutation inputs,
	since we potentially exclude duplicates of pivots from recursive calls.
	
	The good thing about $D_r$ is that \parenthesisclause{unlike $J_r$}
	it does not depend on $n$: $D_r \eqdist \betadist(t+1,t+1)$ in every node.
	Also \parenthesisclause{unlike $V_r$} it does not depend on $\vect q$.
	Since $\tilde J_r$ is concentrated around $n D_r$,
	$J_r$ is likely to be $\ge \alpha n$ only for $D_r > \alpha$.
	Precisely for a constant $\delta > 0$ we have
	\begin{align*}
			\Prob[\bigg]{\tilde J_r \ge (D_r + \delta)n \given D_r} 
		&\wwrel\le
			\Prob[\Bigg]{\biggl|\frac{\tilde J_r}n - D_r \biggr| \ge \delta \given D_r} 
	\\	&\wwrel{\relwithref[r]{lem:chernoff-bound-binomial}\le}
			2\exp(-2\delta^2 n)
	\\	&\wwrel\le
			2\exp(-2\tilde\delta^2 n)
			\qquad\text{for any $\tilde\delta < \delta$}.
	\end{align*}
	Using this and separately considering $D_r$ larger resp.\ smaller $\alpha-\delta$ yields
	\begin{align*}
			\Prob[\big]{J_r \ge \alpha n} 
		&\wwrel\le 
			\Prob[\big]{\tilde J_r \ge \alpha n } 
	\\	&\wwrel=
			\Eover[\big] D {
				\indicator{D_r < \alpha - \delta} \cdot 
					\Prob{J_r \ge \alpha n \given D_r = d, d < \alpha-\delta}
			}
	\\*	&\wwrel\ppe{}
			\bin+
			\Eover[\big] D {
				\indicator{D_r > \alpha-\delta} \cdot 
					\Prob{J_r \ge \alpha n \given D_r = d, d > \alpha-\delta}
			}
	\\	&\wwrel\le
			\Prob{D_r < \alpha-\delta}\cdot 2\exp(-2\delta^2 n)
			\bin+ 
			\Prob{D_r > \alpha-\delta}\cdot 1
	\intertext{%
		if we choose, say $\delta=0.01$, we have 
		$2\exp(-2\delta^2 n) \le 0.005$ for $n\ge n_0 = 30\,000$
	}
		&\wwrel\le
			\Prob{D_r > \alpha-0.01}
			\bin+ 0.005
			\qquad(n\ge n_0).
	\end{align*}
	Plugging in above, we find that with $n_0 = 30\,000$, we can choose
	\begin{align*}
	\numberthis\label{eq:equals-prob-node-not-alpha-balanced}
			p
		\wwrel=
			0.99 \wbin- 2 \cdot I_{\alpha-0.01,1}(t+1,t+1)
		\wwrel\le
			\Prob{v \text{ $\alpha$-balanced}}
			\qquad(n\ge n_0).
	\end{align*}
	
	Since $p=p(\alpha)$ is continuous and $\ge 0.97$ for $\alpha=1$ there is always
	a valid choice $\alpha<1$ with $p>0$.
	We are free to choose any such $\alpha$; 
	the resulting constant $c$ for the achieved height bound then 
	has to satisfy $c > \bigl(1+\frac1{\ln(1/\alpha)}\bigr)\big/p$.
	It is not clear in general which choice yields the best bounds, 
	so we keep it as a parameter in the analysis.
\end{proof}

\noindent
Two further remarks are in order about \wref{lem:equals-recursion-trees-log-height-whp}.
\begin{itemize}
\item \textbf{Height-bound for any input.}\\
	The attentive reader might have noticed that we do not make use of
	the assumption that the input consists of $\distFromWeights{\vect q}$
	elements.
	In fact, the above proof works for \emph{any} randomly permuted input,
	since we actually compute the subproblem sizes in the most unfavorable
	case: when all elements are distinct.
	
	For randomized Quicksort, the random-order assumption is also vacuous;
	we thus have proved the more general statement that randomized Quicksort
	has $\Oh(\log n)$ recursion depth \whp for any input.
	In particular, this proves \wref{pro:tree-log-height-weak-whp}.
	
\item \textbf{Height-bound in terms of $\vect q$.}\\
	For \emph{saturated} trees, our bound on the height in terms of $n$ is meaningless.
	By similar arguments as above we can show that the height is in $\Oh(\log(1/\mu))$ 
	with high probability as $1/\mu \to \infty$,
	where $\mu$ is the smallest probability $q_v$: 
	Intuitively, after $c\ln(1/\mu)$ balanced subdivisions of the unit interval,
	we are left with segments of size less than $\mu$, so after so many partitioning rounds,
	we have reduced the subuniverse sizes to $1$.
	The subproblems are then solved in one further partitioning step.
	
	This bound is intuitively more appealing, but for our use case in the proof of 
	the separation theorem (\wref{thm:separating-n-and-q}), 
	we are dealing with non-saturated trees and the $\log n$ bound turns out to be more convenient.
	(There we only require that $1/\mu$ does not grow \emph{too fast} with $n$, 
	but we do not have any guarantee that it grows at all.
	The height-bound $c \log(1/\mu)$ only holds with high probability as $1/\mu$ goes to infinity;
	we would then need a case distinction on the growth rate of $1/\mu$ \dots)
\end{itemize}

%% file: equal-keys-main.bbl
\begin{thebibliography}{48}
\providecommand{\natexlab}[1]{#1}
\providecommand{\url}[1]{\texttt{#1}}
\expandafter\ifx\csname urlstyle\endcsname\relax
  \providecommand{\doi}[1]{doi: #1}\else
  \providecommand{\doi}{doi: \begingroup \urlstyle{rm}\Url}\fi

\bibitem[Allen and Munro(1978)]{AllenMunro1978}
B.~Allen and I.~Munro.
\newblock Self-organizing binary search trees.
\newblock \emph{Journal of the ACM}, 25\penalty0 (4):\penalty0 526--535,
  October 1978.
\newblock \doi{10.1145/322092.322094}.

\bibitem[Archibald and Cl{\'{e}}ment(2006)]{Archibald2006}
M.~Archibald and J.~Cl{\'{e}}ment.
\newblock Average depth in a binary search tree with repeated keys.
\newblock In \emph{Colloquium on Mathematics and Computer Science}, volume~0,
  pages 309--320, 2006.

\bibitem[Bayer(1975)]{Bayer1975}
P.~J. Bayer.
\newblock \emph{Improved Bounds on the Cost of Optimal and Balanced Binary
  Search Trees}.
\newblock {M}aster's {T}hesis, Massachusetts Institute of Technology, 1975.

\bibitem[Bentley and McIlroy(1993)]{Bentley1993}
J.~L. Bentley and M.~D. McIlroy.
\newblock Engineering a sort function.
\newblock \emph{Software: Practice and Experience}, 23\penalty0 (11):\penalty0
  1249--1265, 1993.

\bibitem[Burge(1976)]{Burge1976}
W.~H. Burge.
\newblock An analysis of binary search trees formed from sequences of
  nondistinct keys.
\newblock \emph{Journal of the ACM}, 23\penalty0 (3):\penalty0 451--454, July
  1976.
\newblock \doi{10.1145/321958.321965}.

\bibitem[Devroye(1983)]{Devroye1983}
L.~Devroye.
\newblock The equivalence of weak, strong and complete convergence in $l_1$ for
  kernel density estimates.
\newblock \emph{The Annals of Statistics}, 11\penalty0 (3):\penalty0 896--904,
  September 1983.
\newblock \doi{10.1214/aos/1176346255}.

\bibitem[Devroye(1986)]{Devroye1986}
L.~Devroye.
\newblock \emph{{Non-Uniform Random Variate Generation}}.
\newblock Springer New York, 1986.
\newblock (available on author's website
  \url{http://luc.devroye.org/rnbookindex.html}).

\bibitem[Devroye(1993)]{Devroye1993}
L.~Devroye.
\newblock On the expected height of fringe-balanced trees.
\newblock \emph{Acta Informatica}, 30\penalty0 (5):\penalty0 459--466, May
  1993.
\newblock \doi{10.1007/BF01210596}.

\bibitem[{\relax DLMF}()]{DLMF}
{\relax DLMF}.
\newblock {NIST Digital Library of Mathematical Functions}.
\newblock Release 1.0.10; Release date 2015-08-07.
\newblock URL \url{http://dlmf.nist.gov}.

\bibitem[Dor and Zwick(2001)]{DorZwick2001}
D.~Dor and U.~Zwick.
\newblock Median selection requires {$(2+\epsilon)n$} comparisons.
\newblock \emph{SIAM Journal on Discrete Mathematics}, 14\penalty0
  (3):\penalty0 312--325, January 2001.
\newblock \doi{10.1137/S0895480199353895}.

\bibitem[Drmota(2009)]{Drmota2009}
M.~Drmota.
\newblock \emph{Random Trees}.
\newblock Springer, 2009.
\newblock ISBN 978-3-211-75355-2.

\bibitem[Dubhashi and Panconesi(2009)]{Dubhashi2009}
D.~Dubhashi and A.~Panconesi.
\newblock \emph{Concentration of measure for the analysis of randomized
  algorithms}.
\newblock Cambridge University Press, 2009.

\bibitem[Erickson(2017)]{Erickson2017}
J.~Erickson.
\newblock Tail inequalities.
\newblock
  \url{https://courses.engr.illinois.edu/cs473/sp2017/notes/04-chernoff.pdf},
  2017.

\bibitem[Estivill-Castro and Wood(1992)]{EstivillCastro1992}
V.~Estivill-Castro and D.~Wood.
\newblock {A survey of adaptive sorting algorithms}.
\newblock \emph{ACM Computing Surveys}, 24\penalty0 (4):\penalty0 441--476,
  1992.

\bibitem[Flajolet and Sedgewick(2009)]{Flajolet2009}
P.~Flajolet and R.~Sedgewick.
\newblock \emph{Analytic Combinatorics}.
\newblock Cambridge University Press, 2009.
\newblock ISBN 978-0-52-189806-5.
\newblock (available on author's website:
  \url{http://algo.inria.fr/flajolet/Publications/book.pdf}).

\bibitem[Gradshteyn and Ryzhik(2007)]{Gradshteyn2007}
I.~Gradshteyn and I.~Ryzhik.
\newblock \emph{Table of Integrals, Series, and Products}.
\newblock Academic Press, 7th edition, 2007.
\newblock ISBN 978-0-12-373637-6.

\bibitem[Greene(1983)]{Greene1983}
D.~H. Greene.
\newblock \emph{{Labelled formal languages and their uses}}.
\newblock {Ph.\hspace{.2ex}D.} thesis, Stanford University, January 1983.

\bibitem[Hibbard(1962)]{Hibbard1962}
T.~N. Hibbard.
\newblock Some combinatorial properties of certain trees with applications to
  searching and sorting.
\newblock \emph{Journal of the ACM}, 9\penalty0 (1):\penalty0 13--28, January
  1962.
\newblock \doi{10.1145/321105.321108}.

\bibitem[Hoare(1961)]{Hoare1961a}
C.~A.~R. Hoare.
\newblock {Algorithm 64: Quicksort}.
\newblock \emph{Communications of the ACM}, 4\penalty0 (7):\penalty0 321, July
  1961.

\bibitem[Hoare(1962)]{Hoare1962}
C.~A.~R. Hoare.
\newblock {Quicksort}.
\newblock \emph{The Computer Journal}, 5\penalty0 (1):\penalty0 10--16, January
  1962.

\bibitem[Huang and Wong(1983)]{HuangWong1983}
S.-H.~S. Huang and C.~K. Wong.
\newblock Binary search trees with limited rotation.
\newblock \emph{BIT}, \penalty0 (4):\penalty0 436--455, December 1983.
\newblock \doi{10.1007/BF01933619}.

\bibitem[Huang and Wong(1984)]{HuangWong1984}
S.-H.~S. Huang and C.~K. Wong.
\newblock Average number of rotations and access cost in {iR}-trees.
\newblock \emph{BIT}, 24\penalty0 (3):\penalty0 387--390, September 1984.
\newblock \doi{10.1007/BF02136039}.

\bibitem[{Java Core Library Development Mailing List}(2009)]{javacoredevel2009}
{Java Core Library Development Mailing List}.
\newblock Replacement of quicksort in java.util.arrays with new dual-pivot
  quicksort, 2009.
\newblock URL
  \url{https://www.mail-archive.com/core-libs-dev@openjdk.java.net/msg02608.html}.

\bibitem[Katajainen and Pasanen(1992)]{KatajainenPasanen1992}
J.~Katajainen and T.~Pasanen.
\newblock Stable minimum space partitioning in linear time.
\newblock \emph{BIT}, 32\penalty0 (4):\penalty0 580--585, December 1992.
\newblock ISSN 0006-3835.
\newblock \doi{10.1007/BF01994842}.

\bibitem[Katajainen and Pasanen(1994)]{KatajainenPasanen1994}
J.~Katajainen and T.~Pasanen.
\newblock Sorting multisets stably in minimum space.
\newblock \emph{Acta Informatica}, 31\penalty0 (4):\penalty0 301--313, April
  1994.
\newblock \doi{10.1007/BF01178508}.

\bibitem[Kemp(1996)]{Kemp1996}
R.~Kemp.
\newblock Binary search trees constructed from nondistinct keys with/without
  specified probabilities.
\newblock \emph{Theoretical Computer Science}, 156\penalty0 (1-2):\penalty0
  39--70, March 1996.
\newblock \doi{10.1016/0304-3975(95)00302-9}.

\bibitem[Knuth(1998)]{Knuth1998}
D.~E. Knuth.
\newblock \emph{{The Art Of Computer Programming: Searching and Sorting}}.
\newblock Addison Wesley, 2nd edition, 1998.
\newblock ISBN 978-0-20-189685-5.

\bibitem[Mahmoud(1992)]{Mahmoud1992evolution}
H.~M. Mahmoud.
\newblock \emph{Evolution of Random Search Trees}.
\newblock Wiley, 1992.
\newblock ISBN 0-471-53228-2.

\bibitem[Mahmoud(2000)]{Mahmoud2000sorting}
H.~M. Mahmoud.
\newblock \emph{Sorting: A distribution theory}.
\newblock John Wiley \& Sons, 2000.
\newblock ISBN \mbox{1-118-03288-8}.

\bibitem[Mart\'{\i}nez and Roura(2001)]{Martinez2001}
C.~Mart\'{\i}nez and S.~Roura.
\newblock Optimal sampling strategies in {Q}uicksort and {Q}uickselect.
\newblock \emph{SIAM Journal on Computing}, 31\penalty0 (3):\penalty0 683--705,
  2001.
\newblock \doi{10.1137/S0097539700382108}.

\bibitem[McDiarmid(1998)]{McDiarmid1998}
C.~J.~H. McDiarmid.
\newblock Concentration.
\newblock In M.~Habib, C.~McDiarmid, J.~Ramirez-Alfonsin, and B.~Reed, editors,
  \emph{Probabilistic Methods for Algorithmic Discrete Mathematics}, pages
  195--248. Springer, Berlin, 1998.

\bibitem[McDiarmid and Hayward(1996)]{McDiarmidHayward1996}
C.~J.~H. McDiarmid and R.~B. Hayward.
\newblock Large deviations for {Quicksort}.
\newblock \emph{Journal of Algorithms}, 21\penalty0 (3):\penalty0 476--507,
  November 1996.
\newblock \doi{10.1006/jagm.1996.0055}.

\bibitem[Mitzenmacher and Upfal(2005)]{MitzenmacherUpfal2005}
M.~Mitzenmacher and E.~Upfal.
\newblock \emph{Probability and Computing: Randomized Algorithms and
  Probabilistic Analysis}.
\newblock Cambridge University Press, 2005.
\newblock ISBN 0-521-83540-2.

\bibitem[Munro and Spira(1976)]{MunroSpira1976}
I.~Munro and P.~M. Spira.
\newblock Sorting and searching in multisets.
\newblock \emph{SIAM Journal on Computing}, 5\penalty0 (1):\penalty0 1--8,
  March 1976.
\newblock \doi{10.1137/0205001}.

\bibitem[Munro and Raman(1991)]{MunroRaman1991}
J.~I. Munro and V.~Raman.
\newblock Sorting multisets and vectors in-place.
\newblock In \emph{Workshop on Algorithms and Data Structures (WADS)}, volume
  519 of \emph{LNCS}, pages 473--480, Berlin/Heidelberg, 1991. Springer.
\newblock \doi{10.1007/BFb0028285}.

\bibitem[Poblete and Munro(1985)]{PobleteMunro1985}
P.~V. Poblete and J.~I. Munro.
\newblock The analysis of a fringe heuristic for binary search trees.
\newblock \emph{Journal of Algorithms}, 6\penalty0 (3):\penalty0 336--350,
  September 1985.
\newblock \doi{10.1016/0196-6774(85)90003-3}.

\bibitem[Sedgewick(1977{\natexlab{a}})]{Sedgewick1977}
R.~Sedgewick.
\newblock The analysis of {Quicksort} programs.
\newblock \emph{Acta Informatica}, 7\penalty0 (4):\penalty0 327--355,
  1977{\natexlab{a}}.

\bibitem[Sedgewick(1977{\natexlab{b}})]{Sedgewick1977a}
R.~Sedgewick.
\newblock Quicksort with equal keys.
\newblock \emph{SIAM Journal on Computing}, 6\penalty0 (2):\penalty0 240--267,
  1977{\natexlab{b}}.

\bibitem[Sedgewick and Bentley(1999)]{SedgewickBentley1999}
R.~Sedgewick and J.~Bentley.
\newblock New research on theory and practice of sorting and searching (talk
  slides), 1999.
\newblock URL \url{http://www.cs.princeton.edu/~rs/talks/Montreal.pdf}.

\bibitem[Sedgewick and Bentley(2002)]{SedgewickBentley2002}
R.~Sedgewick and J.~Bentley.
\newblock {Quicksort} is optimal (talk slides), 2002.
\newblock URL
  \url{http://www.cs.princeton.edu/~rs/talks/QuicksortIsOptimal.pdf}.

\bibitem[Sedgewick and Wayne(2011)]{Sedgewick2011algorithms}
R.~Sedgewick and K.~Wayne.
\newblock \emph{Algorithms}.
\newblock Addison-Wesley, 4th edition, 2011.
\newblock ISBN 978-0-32-157351-3.

\bibitem[Seidel(1993)]{Seidel1993}
R.~Seidel.
\newblock Backwards analysis of randomized geometric algorithms.
\newblock In J.~Pach, editor, \emph{New Trends in Discrete and Computational
  Geometry}, pages 37--67. Springer, 1993.
\newblock \doi{10.1007/978-3-642-58043-7_3}.

\bibitem[Sen and Gupta(1999)]{SandeepGupta1999}
S.~Sen and N.~Gupta.
\newblock Distribution-sensitive algorithms.
\newblock \emph{Nordic Journal of Computing}, 6\penalty0 (2):\penalty0 194,
  1999.

\bibitem[Shannon(1948)]{Shannon1948}
C.~E. Shannon.
\newblock A mathematical theory of communication.
\newblock \emph{Bell System Technical Journal}, 27\penalty0 (3):\penalty0
  379--423, July 1948.
\newblock \doi{10.1002/j.1538-7305.1948.tb01338.x}.

\bibitem[Sleator and Tarjan(1985)]{SleatorTarjan1985}
D.~D. Sleator and R.~E. Tarjan.
\newblock Self-adjusting binary search trees.
\newblock \emph{Journal of the ACM}, 32\penalty0 (3):\penalty0 652--686, July
  1985.
\newblock \doi{10.1145/3828.3835}.

\bibitem[Walker and Wood(1976)]{Walker1976}
A.~Walker and D.~Wood.
\newblock Locally balanced binary trees.
\newblock \emph{The Computer Journal}, 19\penalty0 (4):\penalty0 322--325,
  April 1976.
\newblock \doi{10.1093/comjnl/19.4.322}.

\bibitem[Wild(2016)]{Wild2016}
S.~Wild.
\newblock \emph{Dual-Pivot Quicksort and Beyond: Analysis of Multiway
  Partitioning and Its Practical Potential}.
\newblock Doktorarbeit (\phdthesis), Technische Universit{\"a}t Kaiserslautern,
  2016.
\newblock URL
  \url{http://nbn-resolving.de/urn/resolver.pl?urn:nbn:de:hbz:386-kluedo-44682}.
\newblock ISBN 978-3-00-054669-3.

\bibitem[Wild and Nebel(2012)]{Wild2012}
S.~Wild and M.~E. Nebel.
\newblock Average case analysis of {Java}~7's dual pivot {Quicksort}.
\newblock In L.~Epstein and P.~Ferragina, editors, \emph{European Symposium on
  Algorithms (ESA)}, volume 7501 of \emph{LNCS}, pages 825--836. Springer,
  2012.
\newblock URL \url{http://arxiv.org/abs/1310.7409}.

\end{thebibliography}
